\documentclass[12pt]{article}
\usepackage{graphicx}
\usepackage{amsmath,amsthm,amssymb,mathrsfs,mathtools,bm}
\usepackage[OT1]{fontenc}
\usepackage{float}
\usepackage{titlesec}
\usepackage{multirow,array}
\usepackage{diagbox}
\usepackage[a4paper,top=1in, bottom=1in,left=1in,right=1in]{geometry}
\usepackage{fancyhdr}

\titleformat{\section}
		{\scshape \large \bfseries\center}
         {\thesection}%
        {0.5em}%
        {}%
        []%
\titleformat{\subsection}%
        {\normalfont\bfseries}%
        {\thesubsection}%
        {0.5em}%
        {}%
        []%

\allowdisplaybreaks

\def\subsectiontitle{}
\def\subsubsectiontitle{}
\usepackage{comment}
\usepackage[backgroundcolor=orange!50!white]{todonotes}
\usepackage{tikz}
\usetikzlibrary{automata, positioning}

\usepackage{kpfonts}
\usepackage{inconsolata}

\usepackage{makecell}
\usepackage{cellspace}
\usepackage[shortlabels]{enumitem}

\usepackage[colorlinks,breaklinks=true]{hyperref}
\usepackage{breakcites}

\usepackage{xcolor}
\AtBeginDocument{%
	\hypersetup{%
    linkcolor={red!50!black},%
    citecolor={blue!50!black},%
    urlcolor={blue!80!black}%
}%
}%

\usepackage[nameinlink,noabbrev,sort,capitalise]{cleveref}
\makeatletter
\def\ps@pprintTitle{%
 \let\@oddhead\@empty
 \let\@evenhead\@empty
 \def\@oddfoot{\emph{Very preliminary version}\hfill\emph{This draft: \today}}%
 \let\@evenfoot\@oddfoot}
\makeatother
\usepackage{etoolbox}

\newtheorem{prop}{Proposition}
\crefname{prop}{Proposition}{Propositions}

\newtheorem{thm}{Theorem}
\crefname{thm}{Theorem}{Theorems}

\newtheorem{cor}{Corollary}[thm]
\crefname{cor}{Corollary}{Corollaries}

\newtheorem{lem}{Lemma}
\crefname{lem}{Lemma}{Lemmas}

\newtheorem{ass}{Assumption}
\crefname{ass}{Assumption}{Assumptions}

\crefname{axiom}{Axiom}{Axioms}

\newtheorem{defi}{Definition}
\crefname{defi}{Definition}{Definitions}

\theoremstyle{remark}
\newtheorem{remark}{Remark}
\crefname{remark}{Remark}{Remarks}

\theoremstyle{claim}

\crefname{claim}{Claim}{Claims}

\theoremstyle{definition}
\newtheorem{eg}{Example}
\crefname{eg}{Example}{Examples}

\crefname{problem}{Problem}{Problems}

\newcommand{\MM}{M}
\newcommand{\mm}{m}

\usepackage{indentfirst}
\allowdisplaybreaks

\let\oldfootnote\footnote
\renewcommand\footnote[1]{\oldfootnote{\hspace{.4mm}#1}}

\renewcommand{\footnotesize}{\scriptsize}

\makeatletter
\renewenvironment{proof}[1][\proofname] {\par\pushQED{\qed}\normalfont\topsep6\p@\@plus6\p@\relax\trivlist\item[\hskip\labelsep\bfseries#1\@addpunct{.}]\ignorespaces}{\popQED\endtrivlist\@endpefalse}
\makeatother

\let\oldFootnote\footnote
\newcommand\nextToken\relax

\renewcommand\footnote[1]{%
    \oldFootnote{#1}\futurelet\nextToken\isFootnote}

\newcommand\isFootnote{%
    \ifx\footnote\nextToken\textsuperscript{,}\fi}

\usepackage[american]{babel}
\usepackage[authoryear,round]{natbib}
\def\d{\mathrm{d}}
\newcommand{\rp}[1]{\langle #1 \rangle}
\def\E{\mathbb{E}}
\def\R{\mathbb{R}}
\DeclareMathOperator*{\argmax}{argmax}

\titlespacing\section{0pt}{6pt}{4pt}
\titlespacing\subsection{0pt}{4pt}{2pt}
\titlespacing\subsubsection{0pt}{2pt}{2pt}
\titlespacing*{\paragraph}{0pt}{1.25ex plus 1ex minus .2ex}{0.5em}

\begin{document}

\title{Persuasion and Optimal Stopping}

\date{This version: \today}

\author{\makebox[.25\linewidth]{{Andrew Koh}\thanks{MIT Department of Economics; email: \protect\texttt{ajkoh@mit.edu}}}\\{MIT}
\and
\makebox[.25\linewidth]{Sivakorn Sanguanmoo\thanks{MIT Department of Economics; email: \protect\texttt{sanguanm@mit.edu}}} \\ {MIT}
\and \makebox[.25\linewidth]{Weijie Zhong\thanks{Stanford GSB; email: \protect\texttt{weijie.zhong@stanford.edu}
\newline
We are grateful to Daron Acemoglu, Ian Ball, Ben Brooks, Jeffery Ely, Drew Fudenberg, Jan Knoepfle, Stephen Morris, and Alex Wolitzky for insightful discussions. We also thank seminar participants at various seminars and conferences for comments.
}} \\
{Stanford GSB}
}

\maketitle
    \vspace{-2em}
\begin{abstract}
    We develop a duality-based first-order approach to dynamic persuasion in optimal stopping problems with general action-, state-, and time-dependent preferences. A direct-communication reduction recasts the design problem as a semi-static program over joint distributions of stopping beliefs and times; strong duality and a near-necessary first-order condition then reduce it to a one-dimensional differential equation in a multiplier that prices the agent's continuation incentive, characterizing the optimum as a \emph{concavification} of a multiplier-augmented payoff. We demonstrate the method in three applications: dynamic binary persuasion, where optimal policies combine \emph{suspense generation} with \emph{action-targeting}; a structural result by which the principal's time-risk preferences alone determine whether suspense is optimal; and dynamic linear persuasion, where the optimum is \emph{dynamic tail-censorship}.
\end{abstract}

\noindent \textbf{Keywords:} dynamic persuasion, optimal stopping, information design, duality, first-order conditions\\
\textbf{JEL Classification:} D80, D82, D83\\

\clearpage

\section{Introduction}\label{sec:intro}
Uncertainty and irreversibility lie at the heart of many economic decisions. In such environments, decision-makers must choose \emph{when} to stop gathering information as well as \emph{what} irrevocable action to take. As \cite{bernanke1983irreversibility} observed, ``the optimizing investor must decide not only which projects to undertake but also which point in time is the best for making a commitment\ldots the dynamics of investment become very sensitive to expectations about the rate of information arrival.''\footnote{See also classical applications of irreversible decision-making to environmental preservation \citep{arrow1974environmental} and labor markets \citep{dornbusch1987open}.} A large literature accordingly fixes the flow of information and asks how decision-makers should optimally stop and act. But the \emph{design} of information itself is a powerful instrument: a company times its earnings guidance to manage expectations and smooth stock-price volatility; a central bank discloses information about fundamentals to shape the timing of investment; an advisor reveals information about a client's prospects to guide a career choice. The principal in such settings cares about the agent's action, the state, and the timing of the decision, and her preferences along these three dimensions can be rich and non-stationary. How should she design information over time?

The methodological challenge is sharp. The principal's choice variable is a stochastic belief process---an object of infinite dimension---and the agent's optimal stopping problem couples obedience constraints across every continuation history: an incentive to wait today depends on the information promised tomorrow, which depends on the information promised the day after. Classical concavification \citep{kamenica2011bayesian} handles persuasion when the design space is a one-dimensional belief at a fixed time, but it has no native treatment of the timing dimension. Most existing analyses progress by specialization---Markovian belief processes, exponential discounting, particular payoff structures, or principal preferences for full revelation---each restriction buying tractability at the cost of generality.

\textbf{Outline of contribution.} There is an unknown, persistent, payoff-relevant state. A principal chooses a sequence of history-dependent statistical experiments, and an agent, observing their realizations, makes an irreversible choice of when to stop and which action to take. Both players' preferences over the state, action, and stopping time are arbitrary, subject only to mild regularity. We develop a general operational method for this problem and demonstrate it in three applications.

\textit{A duality-based first-order characterization.} \citet{koh2024attention} established a direct-communication reduction for discrete-time dynamic persuasion in stopping problems; we extend that reduction to continuous time (\cref{lem:direct_communication}). The dynamic problem then reduces to a semi-static optimization over joint distributions $f$ of stopping beliefs and stopping times $(\mu,t)\in\Delta(\Theta)\times T$, subject to a single conditional obedience constraint \eqref{eqn:OCC} (in the spirit of \citealp{forges1986approach,myerson1986multistage}). We establish strong duality (\cref{lem:duality}) and derive a sufficient and near-necessary first-order condition (\cref{thm:foc}): the optimal joint distribution \emph{concavifies} a payoff $l_{f,\Lambda}$ that augments the principal's direct stopping payoff $V$ with the agent's intertemporal incentives through a Lagrange multiplier $\Lambda$, a time-dependent shadow price of the continuation incentive. Solving the design problem is then equivalent to solving a one-dimensional ordinary differential equation for $\Lambda$. The recipe is operational: conjecture the support of the optimum, derive the multiplier ODE along it, and verify the off-support concavification inequality. Both of the closed-form applications below follow this recipe.

\textit{Application 1: Dynamic binary persuasion.} In \cref{sec:binary} we apply the method to a canonical binary environment: a binary state $\theta\in\{L,R\}$, a binary action $\{\ell,r\}$, an agent who wishes to match the state, and a principal with state-independent action preferences whose payoffs over actions and time evolve flexibly. Optimal strategies combine two qualitatively distinct building blocks: a \emph{suspense-generation} stage that releases inconclusive ``plot twists'' to keep the agent simultaneously indifferent between stopping and continuing, and an \emph{action-targeting} stage that releases conclusive Poisson news to drive the decision. The characterization (\cref{thm:binary_suspense_L,cor:binary_suspense_R}) makes precise three sharp trade-offs: the magnitude of the persuasion gain determines the \emph{scope} of persuasion; the principal's time-risk preferences determine the \emph{duration} of suspense; and the complementarity or substitutability between persuasion and delay determines the \emph{direction} of targeting. The framework rationalizes ``good news'' versus ``bad news'' arrival processes endogenously.

\textit{Application 2: General time-risk preferences.} In \cref{sec:time_risk} we turn the method on itself to derive a structural result: the dependence of optimal persuasion on time-risk preferences is a feature of the framework, not of the binary specialization. \Cref{thm:time_risk_loving} shows that, regardless of the state and action spaces, time-risk loving preferences rule out suspense: the obedience constraint binds everywhere on the support of the optimum and the implementing strategy is a direct recommendation. \Cref{thm:time_risk_averse} shows that time-risk averse preferences instead bound the persuasion window and force a suspense interval before any stopping, implemented through informative indirect interim messages. Together with the first-order condition, these results predict the qualitative \emph{form} of the optimum from time preferences alone, before any computation.

\textit{Application 3: Dynamic linear persuasion.} In \cref{sec:linear} we extend the framework to a linear environment in which the state lies in a continuous interval and indirect utilities depend on the belief only through its mean, generalizing classical static linear persuasion \citep{dworczak2019simple,kolotilin2018optimal,kolotilin2022censorship,kleiner2021extreme} to a dynamic setting. The continuum first-order condition (\cref{thm:foc_continuum}) features two multipliers, and the dual \emph{decouples} them: a time-dependent multiplier $\Lambda$ for the obedience constraint is belief-invariant, and a belief-dependent convex multiplier $\Gamma$ for the mean-preservation constraint is time-invariant. This decoupling is the technical reason the dynamic problem inherits the tractability of its static counterpart. In a central-bank--investor problem (\cref{thm:tail_censorship}), the optimal disclosure is \emph{dynamic tail-censorship}: the state is revealed only when it lies in a time-expanding interval $[\alpha(t),\beta(t)]$, and persistent silence becomes increasingly informative about the tails---a pattern that resonates with the communication practices of monetary authorities and with corporate financial disclosure.

\subsection{Related literature}\label{ssec:related}

\paragraph{Dynamic information design and optimal stopping.} A growing literature studies the design of information when the receiver's action is a stopping decision: \citet{ely2017beeps} on the timing of news; \citet{ely2020moving} on dynamic persuasion with endogenous effort; \citet{knoepfle2020dynamic} and \citet{hebert2022engagement} on engagement-maximizing disclosure; \citet{orlov2020persuading} and \citet{bizzotto2021dynamic} on persuading an agent to wait under public signals; \citet{ball2023dynamic} and \citet{smolin2021dynamic} on dynamic disclosure with one-time stopping; \citet{henry2019research} on Wald-style learning with an information-selling principal; \citet{escude2023slow} on disclosure under flow-rate constraints; \citet{che2023keeping}, \citet{mekonnen2023persuaded}, and \citet{siegel2020economic} on binary-state persuasion under exogenous frictions; \citet{renault2017optimal} and \citet{zhao2024contracting} on dynamic disclosure with long-run states; and \citet{saeedi2024getting} on attention capture. Each of these studies narrower preferences than ours, typically state-independent payoffs or restricted time-dependence. This paper builds directly on \citet{koh2024attention}, who study discrete-time dynamic persuasion in stopping problems and establish the direct- and indirect-communication principles; our contribution relative to theirs is methodological---the duality-based first-order characterization and the recipe that operationalizes it---and applied---the binary, time-risk, and linear demonstrations.

\paragraph{Revelation principles, mechanism selection, and static design.} The pooling argument underlying \cref{ssec:relaxed} adapts the classical direct-revelation principles of \citet{forges1986approach} and \citet{myerson1986multistage} to a dynamic information-design setting. The static information-design literature originating with \citet{aumann1995repeated} and \citet{kamenica2011bayesian} characterizes commitment-optimal information through belief concavification; our framework recovers these static benchmarks as the time-stationary limit \citep[see][for a survey]{bergemann2019information}. The literature on mechanism-selection games \citep{laffont1988dynamics,bester2001contracting,skreta2015optimal,doval2022mechanism} studies analogous reductions for direct mechanisms under limited commitment; our dynamic-stopping setting is the informational counterpart, and admits a no-commitment-gap result via the indirect-communication principle of \citet{koh2024attention} rather than a commitment gap.

\paragraph{Linear persuasion, suspense, and real options.} The linear application connects to static linear persuasion with mean-measurable utilities \citep{gentzkow2016rothschild,dworczak2019simple,kolotilin2018optimal,kolotilin2022censorship,arieli2023optimal,kleiner2021extreme,yang2024design}, where censorship policies emerge as optimal disclosures and convex-order methods pin down their shape; our dynamic counterpart inherits this tractability through the dual decoupling. The suspense-generation stage is conceptually related to but mechanically distinct from the suspense-optimal information of \citet{ely2015suspense}, in which the agent derives \emph{non-instrumental} utility from variance reduction \citep[see also][]{augenblick2018belief,georgiadis2021information,chen2024information}; in our framework the continuation value is purely instrumental and the suspense paths are pinned down by zero interim surplus. The central-bank application interfaces with the real-options literature \citep{bernanke1983irreversibility,mcdonald1986value,dixit1994investment,stokey2008economics,stokey2016wait}. Finally, holding the principal's policy fixed, the agent solves a Wald-style stopping problem \citep{wald1947foundations,arrow1949bayes}; modern treatments of endogenous information acquisition span no agent control \citep{fudenberg2018speed,gonccalves2024speed}, partial control \citep{moscarini2001optimal,che2019optimal,liang2022dynamically}, and full control \citep{steiner2017rational,zhong2022optimal,hebert2023rational,sannikov2024embedding}. Our paper inverts the perspective: information is endogenous to the principal's design, and the agent retains only the stopping decision.

\section{The Framework}\label{sec:framework}

\subsection{Model}\label{ssec:model}

\paragraph*{Primitives.} $\Theta$ is a finite set of payoff-relevant states. $\mu_0 \in \Delta(\Theta)$ is the common prior. $A$ is a set of actions. The time space is a compact set $T \subset \R_+$, which we take to be a continuous interval throughout the body.\footnote{The discrete-time analog of this framework was introduced in \citet{koh2024attention}.} There are two players, a principal (she) and an agent (he). The agent makes a one-time irreversible choice of action $a$ at a time $t$ of his choosing. For any state, action, and time triple $(\theta, a, t) \in \Theta \times A \times T$, the agent obtains utility $u(\theta, a, t)$ and the principal obtains utility $v(\theta, a, t)$; both $u$ and $v$ are bounded.

\paragraph*{Information.} At the start of the game the principal commits to a persuasion strategy: a c\`adl\`ag martingale $\rp{\mu_t}$ in $\Delta(\Theta)$, accompanied by a suitable underlying probability space $(\Omega, \mathcal{F} = \rp{\mathcal{F}_t}, \mathcal{P})$, describing the common random posterior belief process induced by the principal's sequential experimentation.\footnote{\citet{koh2024attention} establish a no-commitment-gap result for this framework via the \emph{indirect-communication principle}: every commitment-optimal outcome is implementable as the equilibrium path of a subgame-perfect equilibrium of a dynamic experiment-selection game, even without intertemporal commitment. We invoke that result for the dynamic-consistency interpretations of the strategies derived below; see \cref{rmk:dynamic_consistency_binary,rmk:dynamic_consistency_linear}.} We write $\mu$ for a generic belief in $\Delta(\Theta)$, $\rp{\mu_t}$ for a stochastic belief process, and $\delta_\theta \in \Delta(\Theta)$ for the degenerate belief on $\theta$.

\paragraph*{Stopping and action.} The agent observes the realized path of the belief process and makes an irreversible choice of stopping time and action. Define indirect utilities $U, V : \Delta(\Theta) \times T \to \R$ via
\begin{align*}
	U(\mu, t) &:= \max_{a \in A, s \ge t} \sum_{\theta \in \Theta} u(\theta, a, s) \mu(\theta); \\
	V(\mu, t) &:= \max_{a^*, s^*} \sum_{\theta \in \Theta} v(\theta, a^*, s^*) \mu(\theta) \quad \text{s.t.} \quad (a^*, s^*) \in \argmax_{a \in A, s \ge t} \sum_{\theta \in \Theta} u(\theta, a, s) \mu(\theta).
\end{align*}
By construction $U$ is convex in $\mu$ and weakly decreasing in $t$. The principal chooses a belief martingale $\rp{\mu_t}$ and the agent chooses a stopping time $\tau$ optimally with respect to it. The principal's commitment problem is
\begin{align}
	\sup_{\rp{\mu_t}, \tau} & \E\left[ V(\mu_\tau, \tau) \right] \tag{P} \label{prob:1} \\
	\text{s.t.} \quad & \tau \in \argmax_{\tau'} \E\left[ U(\mu_{\tau'}, \tau') \right]. \tag{OC}\label{prob:1:oc}
\end{align}
We extend $U$ from $\Delta(\Theta)$ to $\R_+^{|\Theta|}$ homogeneously of degree 1 by setting $U(\mu, t) := \sum_\theta \mu(\theta) \cdot U(\mu / \sum_\theta \mu(\theta), t)$ for $\mu \ne 0$, a notational device that lets conditional distributions be evaluated without normalization.

\subsection{The relaxed problem}\label{ssec:relaxed}

Solving \eqref{prob:1} over belief martingales directly is intractable: the choice variable is infinite-dimensional and the constraint couples every continuation history. We instead work with a relaxed semi-static problem in which the principal chooses a joint distribution over stopping beliefs and stopping times. The economic idea is the standard direct-communication argument: a designer who wants the agent to continue is best off pooling every continuation message into a single ``continue'' instruction, since pooling weakly tightens the conditional obedience constraint and never affects the joint distribution of stopping outcomes. Pooling reduces the problem to choosing a joint distribution $f \in \Delta_{\mu_0} := \{f \in \Delta(\Delta(\Theta) \times T) : \E_f[\mu] = \mu_0\}$ subject to a single conditional obedience constraint at each interim time. We formalize this reduction (\cref{lem:direct_communication}) in \cref{appendix:reduction}; its discrete-time analog appears in \citet{koh2024attention}.

The relaxed problem is
\begin{align}
	\sup_{f \in \Delta_{\mu_0}} & \int V(\mu, t) f(\d \mu, \d t) \tag{R} \label{prob:relaxed} \\
	\text{s.t.} & \int_{y > t} U(\mu, y) f(\d \mu, \d y) \ge U\left( \int_{y > t} \mu f(\d \mu, \d y), t \right), \quad \forall t \in T^\circ, \tag{OC-C}\label{eqn:OCC}
\end{align}
where $T^\circ := T \setminus \sup T$. The reduction is that \eqref{prob:1} and \eqref{prob:relaxed} have the same optimal value, and any solution to \eqref{prob:relaxed} is implemented in the original game by the simple recommendation strategy constructed in \cref{appendix:reduction}. We work with \eqref{prob:relaxed} for the remainder of the paper.

\subsection{Strong duality and the first-order characterization}\label{ssec:foc}

Since \eqref{prob:relaxed} is a canonical convex program with infinitely many linear constraints (indexed by $t$), we solve it via the Lagrangian dual. Define $\mathcal{L} : \Delta_{\mu_0} \times \mathcal{B}(T^\circ) \to \R$ by
\begin{align*}
	\mathcal{L}(f, \Lambda) := \int V(\mu, \tau) f(\d \mu, \d \tau) + \int_{t \in T^\circ} \left( \int_{\tau > t} U(\mu, \tau) f(\d \mu, \d \tau) - U\left( \int_{\tau > t} \mu f(\d \mu, \d \tau), t \right) \right) \d \Lambda(t),
\end{align*}
where $\mathcal{B}(T^\circ)$ is the set of positive Borel measures on $T^\circ$ and a generic $\Lambda \in \mathcal{B}(T^\circ)$ assigns each interim time $t$ a multiplier on its obedience constraint; with slight abuse of notation we write $\Lambda(t) := \int_{s < t} \d \Lambda(s)$ for its cumulative distribution function. We maintain two regularity assumptions for the remainder of the paper.

\begin{ass}[Information is strictly valuable]\label{ass:positive}
	$\E_{\mu_0}[U(\delta_\theta, t)] > U(\mu_0, t)$ for all $t \in T$.
\end{ass}

\begin{ass}[Continuity]\label{ass:cont}
	$U$ is continuous in $(\mu, t)$, $V$ is upper-semicontinuous in $(\mu, t)$, and $\{V(\mu, \cdot)\}_{\mu \in \Delta(\Theta)}$ is an equicontinuous family of functions from $T$ to $\R$.
\end{ass}

\Cref{ass:positive} fails only when information is valueless at the prior, in which case the problem is degenerate. \Cref{ass:cont} is a standard continuity requirement ensuring existence in both the primal and the dual.

\begin{lem}[Strong duality]\label{lem:duality}
	Under \cref{ass:positive,ass:cont}, strong duality holds:
	\begin{align}
		\sup_f \min_\Lambda \mathcal{L}(f, \Lambda) = \min_\Lambda \sup_f \mathcal{L}(f, \Lambda). \tag{D}\label{eqn:dual}
	\end{align}
	Furthermore, the minimum is achieved by some $\Lambda^* \in \mathcal{B}(T^\circ)$ and the supremum by some $f^* \in \Delta_{\mu_0}$.
\end{lem}

Strong duality ensures that solving the constrained problem \eqref{prob:relaxed} is equivalent to solving the unconstrained $\sup_f \mathcal{L}(f, \Lambda)$ at the optimal multiplier; the proof is in \cref{appendix:duality_proof}. To state the first-order characterization, define the \emph{derivative of the Lagrangian} at a generic belief-time pair $(\mu, t)$ by
\begin{align*}
	l_{f, \Lambda}(\mu, t) := V(\mu, t) + \Lambda(t) \cdot U(\mu, t) - \left( \int_{\tau < t} \nabla_\mu U(\widehat{\mu}_\tau, \tau) \, \d \Lambda(\tau) \right) \cdot \mu,
\end{align*}
where $\widehat{\mu}_t := \E_{f(\mu, \tau)}[\mu \mid \tau > t]$ is the conditional expectation of the belief among continuing draws; when $U(\cdot, t)$ is not differentiable at $\widehat{\mu}_t$ we let $\nabla_\mu U(\widehat{\mu}_t, t)$ denote any selection of subgradients.

\begin{thm}[First-order characterization]\label{thm:foc}
	A joint distribution $f \in \Delta_{\mu_0}$ solves \eqref{prob:relaxed} if there exist $\Lambda \in \mathcal{B}(T^\circ)$, $a \in \R^{|\Theta|}$, and a selection of subgradients such that
	\begin{align}
		l_{f, \Lambda}(\mu, t) \le a \cdot \mu, \quad \text{with equality on the support of } f, \tag{FOC} \label{eqn:foc}
	\end{align}
	$f$ satisfies \eqref{eqn:OCC}, and the complementary slackness condition $\mathcal{L}(f, \Lambda) = \E_f[V]$ holds. Conversely, under \cref{ass:positive,ass:cont}, for every solution $f$ of \eqref{prob:relaxed} there exist $\Lambda$ solving \eqref{eqn:dual}, $a \in \R^{|\Theta|}$, and a selection of subgradients such that \eqref{eqn:foc} holds.
\end{thm}

\Cref{thm:foc} gives a sufficient and near-necessary first-order characterization of optimality, with the necessity part relying on strong duality (the proof is in \cref{appendix:foc_proof}). \Cref{eqn:foc} is a ``concavification'' condition---the dynamic analog of the static condition of \citet{kamenica2011bayesian}: the derivative-of-Lagrangian $l_{f,\Lambda}$ touches its upper supporting hyperplane $a\cdot\mu$ only on the support of the optimal $f$. The function being concavified is a $\Lambda$-augmented payoff: $V(\mu,t)$ is the principal's direct stopping payoff; $\Lambda(t)\cdot U(\mu,t)$ is her shadow benefit from relaxing the agent's continuation incentive at time $t$; and the final term subtracts a shadow cost, since pulling mass to time $t$ tightens the obedience constraints at all earlier times, scaled by the gradient of $U$ at the realized continuation beliefs. The principal's effective strategy is therefore a mean-preserving spread of the prior in the product space $\Delta(\Theta)\times T$, with intertemporal incentives baked in via $\Lambda$.

\subsection{Recipe for solving the dual}\label{ssec:recipe}

\Cref{thm:foc} reduces optimal dynamic persuasion to finding the multiplier $\Lambda$, a one-dimensional function of time. The three-step recipe is:
\begin{enumerate}[leftmargin=2em]
	\item \textbf{Conjecture the support of $f$.} The optimal joint distribution is typically supported on a low-dimensional submanifold of $\Delta(\Theta) \times T$---a path, a pair of paths, or a band.
	\item \textbf{Derive the multiplier ODE.} The equality $l_{f, \Lambda}(\mu, t) = a \cdot \mu$ on the support gives an ordinary differential equation pinning down $\Lambda$ along the conjectured support.
	\item \textbf{Verify the off-support inequality.} The condition $l_{f, \Lambda}(\mu, t) \le a \cdot \mu$ off the support is checked globally, typically by showing that $\widehat{l}(\mu) := \sup_{f' \in \Delta_\mu} \E_{f'}[l_{f, \Lambda}]$ is concavified by $a \cdot \mu$ at $\mu_0$.
\end{enumerate}
We apply this recipe in \cref{sec:binary,sec:linear}; closed-form solutions emerge in both. The recipe also admits an efficient numerical implementation: a finite discretization of $T \times \Delta(\Theta)$ reduces the design problem to a $|\Theta \times T|$-dimensional gradient descent on the dual, developed in \cref{appendix:numerical}.

\section{Application 1: Binary Persuasion}\label{sec:binary}

We first apply the duality method to a canonical binary persuasion environment. The setting is sparse enough to admit closed-form solutions, yet rich enough to expose the three structural trade-offs that recur in dynamic persuasion: the persuasion gain determines the \emph{scope} of persuasion, the principal's time-risk preferences determine the \emph{duration} of suspense, and the complementarity or substitutability between persuasion and delay determines the \emph{direction} of targeting. The first-order condition of \cref{thm:foc} delivers all three through a transparent one-dimensional ODE in $\Lambda$.

\subsection{Setup}\label{ssec:binary_setup}

The state $\theta$ is either $L$ or $R$. The common prior is $\mu_0(R) \in (0, 1/2)$, and $\mu_t := \mathbb{P}(\theta = R \mid \mathcal{F}_t)$ is the agent's posterior on $R$ at time $t$. The agent takes a binary action $a \in \{\ell, r\}$ and wishes to match the state, with a linear delay cost:
\begin{align*}
	u(L, \ell, t) = u(R, r, t) = 1 - c(t), \qquad u(L, r, t) = u(R, \ell, t) = -c(t), \qquad c(t)=t.
\end{align*}
The agent's indirect utility is therefore $U(\mu, t) = \max\{\mu, 1 - \mu\} - t$. The principal's preferences over the agent's action are state-independent but evolve over time:
\begin{align*}
	\text{payoff from $\ell$ at $t$:} \quad v_\ell + h_\ell(t), \qquad \text{from $r$ at $t$:} \quad v_r + h_r(t),
\end{align*}
where $h_\ell, h_r$ are strictly increasing and differentiable, and we normalize $v_r \ge v_\ell$, $h_\ell(0) = h_r(0) = 0$.\footnote{Setting $v_r \ge v_\ell$ is without loss, since otherwise we relabel the states; $h_\ell(0) = h_r(0) = 0$ is without loss after a translation of $v_\ell, v_r$; and $c(t) = t$ is obtained by reparameterizing the time space $t \to c(t)$, which converts $h_\ell, h_r$ to $h_\ell \circ c^{-1}, h_r \circ c^{-1}$.} Two derived primitives organize the analysis: $\Delta v := v_r - v_\ell$ is the \emph{initial} persuasion gain at $t = 0$, and $\Delta h(t) := h_r(t) - h_\ell(t)$ is the \emph{additional} persuasion gain at time $t$. The marginal delay gains $h_\ell', h_r'$ and the marginal persuasion gain $\Delta h'$ may be arbitrary functions of time.

This binary environment is the workhorse of dynamic persuasion: it captures jury--prosecutor exchanges, advisor--politician communication, supervisor--worker feedback, and analyst--investor disclosure \citep{kamenica2011bayesian, morris2001political, fudenberg2019training, ely2023feedback}. In each, the principal cares about both the agent's eventual action and the timing of the decision, and the two motives may reinforce or oppose each other over time. Existing dynamic persuasion work in this setting typically imposes exogenous frictions on the principal's information technology \citep{che2023keeping, escude2023slow, henry2019research, mekonnen2023persuaded}. By contrast, persuasion here is fully frictionless---the principal can run any sequence of experiments---and delay arises endogenously from differential time-risk preferences between principal and agent.

\subsection{Optimal strategy}\label{ssec:binary_characterization}

We construct the optimal strategy from two building-block belief paths.

\paragraph{Suspense-generation paths.} A \emph{suspense-generation} stage induces interim beliefs on one of two paths: the $L$-path $\mu^L_t$ stays below $1/2$ and the $R$-path $\mu^R_t$ stays above $1/2$. In each period the principal runs an experiment that, with appropriate Poisson intensity, generates a \emph{plot twist}---the belief jumps from one path to the other; absent the twist, the belief drifts along its current path. The two paths are constructed so that the agent obtains \emph{zero continuation surplus} throughout the stage: he is indifferent between stopping immediately and waiting at every interim history. This indifference pins down the two paths uniquely, given the suspense duration $t_1$ and the direction of targeting that follows.

\paragraph{Action-targeting paths.} An \emph{action-targeting} stage runs a single Poisson experiment that, upon arrival, reveals one state conclusively. If $L$ is revealed, the agent stops and takes $\ell$ ($\ell$-targeting); if $R$ is revealed, he stops and takes $r$ ($r$-targeting). Absent revelation, the belief drifts gradually toward the unrevealed state. The Poisson intensity is calibrated to keep the conditional obedience constraint \eqref{eqn:OCC} \emph{binding} at every interim time.

\paragraph{Composite strategies.} We combine the two blocks into two-stage strategies parameterized by a triple $(t_1, a, t_2)$, where $t_1 \in [0, t_2)$ is the duration of suspense, $a \in \{\ell, r\}$ is the targeted action, and $t_2 > t_1$ is the terminal time at which any agent still continuing must stop:
\begin{itemize}[leftmargin=2em]
	\item \textbf{Suspense${}^{t_1}$-$\ell^{t_2}$}: During the suspense stage $[0, t_1)$, the principal induces $\mu^L_t = \mu_0 (t_1 - t)/t_1$ and $\mu^R_t = \mu_0 (1 - t_1 + t)/(2\mu_0 - t_1)$, which leave the agent indifferent throughout.\footnote{Implicitly, $t_1 < \mu_0$ is required to guarantee $\mu^R_{t_1} \in (0, 1)$.} At $t_1$, the agent either learns that the state is $L$ ($\mu^L_{t_1} = 0$) and stops with action $\ell$, or continues at the higher belief $\mu^R_{t_1} = \mu_0/(2\mu_0 - t_1)$. In the $\ell$-targeting stage $[t_1, t_2]$, the principal reveals $L$ at a Poisson rate calibrated to keep \eqref{eqn:OCC} binding, pinning down the continuing belief path $\mu^*_{t_1}(t) := \mu^R_{t_1} e^{t - t_1}$. If $L$ has not been revealed by $t_2$, the agent stops and takes $r$ at $\mu^*_{t_1}(t_2) \in [1/2, 1]$. See \cref{fig:alibi:belief}.
	\item \textbf{Suspense${}^{t_1}$-$r^{t_2}$}: Symmetric, with the roles of $L$ and $R$ (and $\ell$ and $r$) swapped. The suspense paths are $\mu^L_t = ((1 + t - t_1) \mu_0 - t)/(1 - t_1)$ and $\mu^R_t = (1 + t - (2 + t - t_1) \mu_0)/(1 + t_1 - 2\mu_0)$. At $t_1$, the agent either learns $R$ ($\mu^R_{t_1} = 1$) and stops with action $r$, or continues at the lower belief $\mu^L_{t_1} = (\mu_0 - t_1)/(1 - t_1)$. In the $r$-targeting stage, the principal reveals $R$ at a calibrated rate, pinning down $\mu^\dagger_{t_1}(t) := 1 - (1 - \mu^L_{t_1}) e^{t - t_1}$. If $R$ has not been revealed by $t_2$, the agent stops and takes $\ell$ at $\mu^\dagger_{t_1}(t_2) \in [0, 1/2]$. See \cref{fig:fingerprint:belief}.
\end{itemize}

\begin{figure}[htbp]
	\centering
	\begin{minipage}{0.49\linewidth}
		\centering
		\resizebox{0.95\linewidth}{!}{\tikzset{every picture/.style={line width=0.75pt}} 

\begin{tikzpicture}[x=0.75pt,y=0.75pt,yscale=-1,xscale=1]

\draw [color={rgb, 255:red, 0; green, 122; blue, 255 }  ,draw opacity=1 ]   (65.16,111.08) -- (145.23,111.08) ;
\draw   (145.9,139.7) .. controls (145.9,144.37) and (148.23,146.7) .. (152.9,146.7) -- (178.7,146.7) .. controls (185.37,146.7) and (188.7,149.03) .. (188.7,153.7) .. controls (188.7,149.03) and (192.03,146.7) .. (198.7,146.7)(195.7,146.7) -- (224.5,146.7) .. controls (229.17,146.7) and (231.5,144.37) .. (231.5,139.7) ;
\draw   (65.1,139.7) .. controls (65.1,144.37) and (67.43,146.7) .. (72.1,146.7) -- (94.8,146.7) .. controls (101.47,146.7) and (104.8,149.03) .. (104.8,153.7) .. controls (104.8,149.03) and (108.13,146.7) .. (114.8,146.7)(111.8,146.7) -- (137.5,146.7) .. controls (142.17,146.7) and (144.5,144.37) .. (144.5,139.7) ;
\draw  [draw opacity=0][fill={rgb, 255:red, 0; green, 0; blue, 0 }  ,fill opacity=0.17 ] (65.09,20) -- (314.82,20) -- (314.82,80) -- (65.09,80) -- cycle ;
\draw   (65.09,20) -- (314.82,20) -- (314.82,140) -- (65.09,140) -- cycle ;
\draw  [dash pattern={on 4.5pt off 4.5pt}]  (230.65,28.83) .. controls (238.53,26.17) and (245.2,24.17) .. (251.87,20.5) ;
\draw  [dash pattern={on 0.84pt off 2.51pt}]  (170.53,46.17) .. controls (177.59,77.49) and (178.71,110.1) .. (171.65,137.59) ;
\draw [shift={(171.21,139.27)}, rotate = 285.16] [fill={rgb, 255:red, 0; green, 0; blue, 0 }  ][line width=0.08]  [draw opacity=0] (7.2,-1.8) -- (0,0) -- (7.2,1.8) -- cycle    ;
\draw  [color={rgb, 255:red, 255; green, 255; blue, 255 }  ,draw opacity=1 ] (24.4,0) -- (334.4,0) -- (334.4,160) -- (24.4,160) -- cycle ;
\draw  [color={rgb, 255:red, 208; green, 2; blue, 27 }  ,draw opacity=1 ][fill={rgb, 255:red, 208; green, 2; blue, 27 }  ,fill opacity=1 ] (143.37,139.88) .. controls (143.37,138.73) and (144.17,137.8) .. (145.16,137.8) .. controls (146.15,137.8) and (146.95,138.73) .. (146.95,139.88) .. controls (146.95,141.03) and (146.15,141.96) .. (145.16,141.96) .. controls (144.17,141.96) and (143.37,141.03) .. (143.37,139.88) -- cycle ;
\draw  [dash pattern={on 0.84pt off 2.51pt}]  (204.2,37.5) .. controls (211.25,68.82) and (211.77,110) .. (204.69,137.84) ;
\draw [shift={(204.25,139.53)}, rotate = 285.16] [fill={rgb, 255:red, 0; green, 0; blue, 0 }  ][line width=0.08]  [draw opacity=0] (7.2,-1.8) -- (0,0) -- (7.2,1.8) -- cycle    ;
\draw [color={rgb, 255:red, 0; green, 122; blue, 255 }  ,draw opacity=1 ]   (145.14,52.2) .. controls (175.53,45.5) and (199.27,38.8) .. (230.65,28.83) ;
\draw  [color={rgb, 255:red, 208; green, 2; blue, 27 }  ,draw opacity=1 ][fill={rgb, 255:red, 208; green, 2; blue, 27 }  ,fill opacity=1 ] (228.86,28.75) .. controls (228.86,27.6) and (229.66,26.67) .. (230.65,26.67) .. controls (231.64,26.67) and (232.44,27.6) .. (232.44,28.75) .. controls (232.44,29.9) and (231.64,30.83) .. (230.65,30.83) .. controls (229.66,30.83) and (228.86,29.9) .. (228.86,28.75) -- cycle ;
\draw [color={rgb, 255:red, 208; green, 2; blue, 27 }  ,draw opacity=1 ][line width=1.5]    (145.16,139.88) -- (231.33,139.88) ;
\draw [color={rgb, 255:red, 189; green, 16; blue, 224 }  ,draw opacity=1 ] [dash pattern={on 2.25pt off 1.5pt}]  (145.14,52.2) -- (64.75,72.28) ;
\draw [color={rgb, 255:red, 65; green, 117; blue, 5 }  ,draw opacity=1 ] [dash pattern={on 2.25pt off 1.5pt}]  (145.16,139.88) -- (65.16,111.08) ;
\draw  [fill={rgb, 255:red, 255; green, 255; blue, 255 }  ,fill opacity=1 ] (63.37,111.08) .. controls (63.37,109.93) and (64.17,109) .. (65.16,109) .. controls (66.15,109) and (66.95,109.93) .. (66.95,111.08) .. controls (66.95,112.23) and (66.15,113.16) .. (65.16,113.16) .. controls (64.17,113.16) and (63.37,112.23) .. (63.37,111.08) -- cycle ;
\draw  [dash pattern={on 0.84pt off 2.51pt}]  (89.42,66.25) .. controls (96.05,81.31) and (97.42,101.05) .. (90.65,118.42) ;
\draw [shift={(90,120.03)}, rotate = 292.91] [fill={rgb, 255:red, 0; green, 0; blue, 0 }  ][line width=0.08]  [draw opacity=0] (7.2,-1.8) -- (0,0) -- (7.2,1.8) -- cycle    ;
\draw  [dash pattern={on 0.84pt off 2.51pt}]  (120.25,130.78) .. controls (126.56,115.26) and (128.63,79.98) .. (120.76,59.85) ;
\draw [shift={(120,58.03)}, rotate = 66.1] [fill={rgb, 255:red, 0; green, 0; blue, 0 }  ][line width=0.08]  [draw opacity=0] (7.2,-1.8) -- (0,0) -- (7.2,1.8) -- cycle    ;
\draw [color={rgb, 255:red, 189; green, 16; blue, 224 }  ,draw opacity=1 ] [dash pattern={on 2.25pt off 1.5pt}]  (264,54.78) -- (245.5,54.78) ;
\draw [color={rgb, 255:red, 65; green, 117; blue, 5 }  ,draw opacity=1 ] [dash pattern={on 2.25pt off 1.5pt}]  (264,74.53) -- (245.5,74.53) ;
\draw [color={rgb, 255:red, 0; green, 122; blue, 255 }  ,draw opacity=1 ]   (246.41,94.08) -- (263.75,94.08) ;
\draw [color={rgb, 255:red, 208; green, 2; blue, 27 }  ,draw opacity=1 ]   (246.41,113.83) -- (263.75,113.83) ;
\draw  [color={rgb, 255:red, 0; green, 122; blue, 255 }  ,draw opacity=1 ][fill={rgb, 255:red, 255; green, 255; blue, 255 }  ,fill opacity=1 ] (143.44,111.08) .. controls (143.44,109.93) and (144.24,109) .. (145.23,109) .. controls (146.22,109) and (147.02,109.93) .. (147.02,111.08) .. controls (147.02,112.23) and (146.22,113.16) .. (145.23,113.16) .. controls (144.24,113.16) and (143.44,112.23) .. (143.44,111.08) -- cycle ;

\draw (47.43,131.8) node [anchor=north west][inner sep=0.75pt]  [font=\footnotesize]  {$L$};
\draw (311.4,144.4) node [anchor=north west][inner sep=0.75pt]  [font=\small]  {$t$};
\draw (46.63,12) node [anchor=north west][inner sep=0.75pt]  [font=\footnotesize]  {$R$};
\draw (44.69,74.4) node [anchor=north west][inner sep=0.75pt]  [font=\small]  {$0.5$};
\draw (141.1,146.9) node [anchor=north west][inner sep=0.75pt]  [font=\small]  {$t_{1}$};
\draw (26.15,102.4) node [anchor=north west][inner sep=0.75pt]  [font=\small]  {$Prior$};
\draw (229.1,147.8) node [anchor=north west][inner sep=0.75pt]  [font=\small]  {$t_{2}$};
\draw (70.41,153.11) node [anchor=north west][inner sep=0.75pt]  [font=\footnotesize]  {$ \begin{array}{l}
Suspense-\\
generating
\end{array}$};
\draw (165.2,152.21) node [anchor=north west][inner sep=0.75pt]  [font=\footnotesize]  {$ \begin{array}{l}
\ell -\\
targeting
\end{array}$};
\draw (265.75,45.55) node [anchor=north west][inner sep=0.75pt]  [font=\footnotesize]  {$:\mu _{t}^{R}$};
\draw (265.5,64.8) node [anchor=north west][inner sep=0.75pt]  [font=\footnotesize]  {$:\mu _{t}^{L}$};
\draw (265.25,84.3) node [anchor=north west][inner sep=0.75pt]  [font=\footnotesize]  {$:\hat{\mu }_{t}$};
\draw (265.25,104.05) node [anchor=north west][inner sep=0.75pt]  [font=\footnotesize]  {$:supp( f)$};

\end{tikzpicture}}
		\caption{\emph{Suspense-$\ell$} strategy.}
		\label{fig:alibi:belief}
	\end{minipage}
	\begin{minipage}{0.49\linewidth}
		\centering
		\resizebox{0.95\linewidth}{!}{\tikzset{every picture/.style={line width=0.75pt}} 

\begin{tikzpicture}[x=0.75pt,y=0.75pt,yscale=-1,xscale=1]

\draw  [draw opacity=0][fill={rgb, 255:red, 0; green, 0; blue, 0 }  ,fill opacity=0.17 ] (65.09,20) -- (314.82,20) -- (314.82,80) -- (65.09,80) -- cycle ;
\draw [color={rgb, 255:red, 0; green, 122; blue, 255 }  ,draw opacity=1 ]   (66.48,103.51) -- (134.33,103.51) ;
\draw   (134.67,139.7) .. controls (134.67,144.37) and (137,146.7) .. (141.67,146.7) -- (160.58,146.7) .. controls (167.25,146.7) and (170.58,149.03) .. (170.58,153.7) .. controls (170.58,149.03) and (173.91,146.7) .. (180.58,146.7)(177.58,146.7) -- (199.49,146.7) .. controls (204.16,146.7) and (206.49,144.37) .. (206.49,139.7) ;
\draw   (65.1,139.7) .. controls (65.1,144.37) and (67.43,146.7) .. (72.1,146.7) -- (89.88,146.7) .. controls (96.55,146.7) and (99.88,149.03) .. (99.88,153.7) .. controls (99.88,149.03) and (103.21,146.7) .. (109.88,146.7)(106.88,146.7) -- (127.67,146.7) .. controls (132.34,146.7) and (134.67,144.37) .. (134.67,139.7) ;
\draw   (65.09,20) -- (314.82,20) -- (314.82,140) -- (65.09,140) -- cycle ;
\draw  [fill={rgb, 255:red, 255; green, 255; blue, 255 }  ,fill opacity=1 ] (62.9,103.51) .. controls (62.9,102.36) and (63.7,101.43) .. (64.69,101.43) .. controls (65.68,101.43) and (66.48,102.36) .. (66.48,103.51) .. controls (66.48,104.66) and (65.68,105.59) .. (64.69,105.59) .. controls (63.7,105.59) and (62.9,104.66) .. (62.9,103.51) -- cycle ;
\draw  [dash pattern={on 4.5pt off 4.5pt}]  (206.49,134.03) .. controls (212.49,135.17) and (219.91,137.46) .. (225.63,140.03) ;
\draw  [color={rgb, 255:red, 208; green, 2; blue, 27 }  ,draw opacity=1 ][fill={rgb, 255:red, 208; green, 2; blue, 27 }  ,fill opacity=1 ] (133.97,20.11) .. controls (133.97,18.96) and (134.77,18.03) .. (135.76,18.03) .. controls (136.75,18.03) and (137.55,18.96) .. (137.55,20.11) .. controls (137.55,21.26) and (136.75,22.19) .. (135.76,22.19) .. controls (134.77,22.19) and (133.97,21.26) .. (133.97,20.11) -- cycle ;
\draw [color={rgb, 255:red, 0; green, 0; blue, 0 }  ,draw opacity=1 ] [dash pattern={on 0.84pt off 2.51pt}]  (163.06,122.6) .. controls (170.11,91.64) and (172.07,49.64) .. (165.04,22.07) ;
\draw [shift={(164.6,20.4)}, rotate = 74.67] [fill={rgb, 255:red, 0; green, 0; blue, 0 }  ,fill opacity=1 ][line width=0.08]  [draw opacity=0] (7.2,-1.8) -- (0,0) -- (7.2,1.8) -- cycle    ;
\draw  [color={rgb, 255:red, 255; green, 255; blue, 255 }  ,draw opacity=1 ] (24.4,0) -- (334.4,0) -- (334.4,160) -- (24.4,160) -- cycle ;
\draw [color={rgb, 255:red, 0; green, 0; blue, 0 }  ,draw opacity=1 ] [dash pattern={on 0.84pt off 2.51pt}]  (197.34,130.6) .. controls (204.4,99.64) and (204.54,49.96) .. (197.44,22.08) ;
\draw [shift={(197,20.4)}, rotate = 74.67] [fill={rgb, 255:red, 0; green, 0; blue, 0 }  ,fill opacity=1 ][line width=0.08]  [draw opacity=0] (7.2,-1.8) -- (0,0) -- (7.2,1.8) -- cycle    ;
\draw [color={rgb, 255:red, 0; green, 122; blue, 255 }  ,draw opacity=1 ]   (133.06,116.89) .. controls (154.77,120.31) and (198.49,130.6) .. (206.49,134.03) ;
\draw [color={rgb, 255:red, 208; green, 2; blue, 27 }  ,draw opacity=1 ][line width=1.5]    (135.76,20.11) -- (207.06,20.11) ;
\draw [color={rgb, 255:red, 189; green, 16; blue, 224 }  ,draw opacity=1 ] [dash pattern={on 2.25pt off 1.5pt}]  (264,52.78) -- (245.5,52.78) ;
\draw [color={rgb, 255:red, 65; green, 117; blue, 5 }  ,draw opacity=1 ] [dash pattern={on 2.25pt off 1.5pt}]  (264,72.53) -- (245.5,72.53) ;
\draw [color={rgb, 255:red, 0; green, 122; blue, 255 }  ,draw opacity=1 ]   (246.41,92.08) -- (263.75,92.08) ;
\draw [color={rgb, 255:red, 208; green, 2; blue, 27 }  ,draw opacity=1 ]   (246.41,111.83) -- (263.75,111.83) ;
\draw [color={rgb, 255:red, 189; green, 16; blue, 224 }  ,draw opacity=1 ] [dash pattern={on 2.25pt off 1.5pt}]  (135.76,20.11) -- (65.32,45.71) ;
\draw [color={rgb, 255:red, 65; green, 117; blue, 5 }  ,draw opacity=1 ] [dash pattern={on 2.25pt off 1.5pt}]  (133.06,116.89) -- (64.69,103.51) ;
\draw  [dash pattern={on 0.84pt off 2.51pt}]  (86.85,37.96) .. controls (93.47,53.03) and (95.04,87.82) .. (88.28,106.1) ;
\draw [shift={(87.63,107.74)}, rotate = 292.91] [fill={rgb, 255:red, 0; green, 0; blue, 0 }  ][line width=0.08]  [draw opacity=0] (7.2,-1.8) -- (0,0) -- (7.2,1.8) -- cycle    ;
\draw  [dash pattern={on 0.84pt off 2.51pt}]  (110.77,112.31) .. controls (117.08,96.79) and (119.19,52.09) .. (111.33,31.39) ;
\draw [shift={(110.57,29.54)}, rotate = 66.1] [fill={rgb, 255:red, 0; green, 0; blue, 0 }  ][line width=0.08]  [draw opacity=0] (7.2,-1.8) -- (0,0) -- (7.2,1.8) -- cycle    ;
\draw  [color={rgb, 255:red, 208; green, 2; blue, 27 }  ,draw opacity=1 ][fill={rgb, 255:red, 208; green, 2; blue, 27 }  ,fill opacity=1 ] (204.69,134.03) .. controls (204.69,132.88) and (205.5,131.95) .. (206.49,131.95) .. controls (207.47,131.95) and (208.28,132.88) .. (208.28,134.03) .. controls (208.28,135.18) and (207.47,136.11) .. (206.49,136.11) .. controls (205.5,136.11) and (204.69,135.18) .. (204.69,134.03) -- cycle ;
\draw  [color={rgb, 255:red, 0; green, 122; blue, 255 }  ,draw opacity=1 ][fill={rgb, 255:red, 255; green, 255; blue, 255 }  ,fill opacity=1 ] (132.54,103.51) .. controls (132.54,104.64) and (133.34,105.56) .. (134.33,105.56) .. controls (135.32,105.56) and (136.12,104.64) .. (136.12,103.51) .. controls (136.12,102.37) and (135.32,101.45) .. (134.33,101.45) .. controls (133.34,101.45) and (132.54,102.37) .. (132.54,103.51) -- cycle ;

\draw (48.86,133.37) node [anchor=north west][inner sep=0.75pt]  [font=\footnotesize]  {$L$};
\draw (311.4,144.4) node [anchor=north west][inner sep=0.75pt]  [font=\small]  {$t$};
\draw (48.48,14) node [anchor=north west][inner sep=0.75pt]  [font=\footnotesize]  {$R$};
\draw (44.69,74.4) node [anchor=north west][inner sep=0.75pt]  [font=\small]  {$0.5$};
\draw (130.1,147.57) node [anchor=north west][inner sep=0.75pt]  [font=\small]  {$t_{1}$};
\draw (26.25,100.97) node [anchor=north west][inner sep=0.75pt]  [font=\small]  {$Prior$};
\draw (204.24,146.56) node [anchor=north west][inner sep=0.75pt]  [font=\small]  {$t_{2}$};
\draw (65.46,154.07) node [anchor=north west][inner sep=0.75pt]  [font=\footnotesize]  {$ \begin{array}{l}
Suspense-\\
generating
\end{array}$};
\draw (152.29,154.17) node [anchor=north west][inner sep=0.75pt]  [font=\footnotesize]  {$ \begin{array}{l}
r-\\
targeting
\end{array}$};
\draw (265.75,43.55) node [anchor=north west][inner sep=0.75pt]  [font=\footnotesize]  {$:\mu _{t}^{R}$};
\draw (265.5,62.8) node [anchor=north west][inner sep=0.75pt]  [font=\footnotesize]  {$:\mu _{t}^{L}$};
\draw (265.25,82.3) node [anchor=north west][inner sep=0.75pt]  [font=\footnotesize]  {$:\hat{\mu }_{t}$};
\draw (265.25,102.05) node [anchor=north west][inner sep=0.75pt]  [font=\footnotesize]  {$:supp( f)$};

\end{tikzpicture}}
		\caption{\emph{Suspense-$r$} strategy.}
		\label{fig:fingerprint:belief}
	\end{minipage}
\end{figure}

The triple $(t_1, a, t_2)$ admits a natural reading: $t_1$ is the \emph{duration} of suspense (larger $t_1$ pushes the agent's stopping mass further into the future), $a$ is the \emph{direction} of targeting, and $t_2$ is the \emph{scope} of persuasion (larger $t_2$ requires a more dispersed terminal belief and reveals more information). Rather than reducing these strategies to simple recommendations, we describe them in terms of the indirect messages that make them robust to the level of intertemporal commitment.

\begin{remark}[Dynamic consistency]\label{rmk:dynamic_consistency_binary}
	The plot twists in the suspense-generating stage are not simple recommendations to stop or continue; they are indirect interim messages that leave the agent with zero interim surplus at every continuation history. \citet{koh2024attention} establishes that zero interim surplus is exactly the condition needed for an outcome to be implementable as the equilibrium path of a subgame-perfect equilibrium of the dynamic experiment-selection game \emph{without} intertemporal commitment. The strategies derived below are therefore robust to whether the principal can commit.
\end{remark}

The next theorem characterizes when Suspense${}^{t_1}$-$\ell^{t_2}$ is optimal. Define
\begin{align}
	\Psi(t, t_2) := \int_t^{t_2} e^{t - s} h_\ell'(s) \, \d s + e^{t - t_2} h_r'(t_2), \label{eqn:Psi}
\end{align}
which aggregates the marginal benefit of extending the suspense window by $\d t_1$ over the remaining $\ell$-targeting interval $[t_1, t_2]$, exponentially weighted.

\begin{thm}[Optimality of Suspense-$\ell$]\label{thm:binary_suspense_L}
	Let $\Psi(t, t_2)$ be as in \eqref{eqn:Psi} and let $0 \le t_1 < t_2 < \mu_{t_1}^{*-1}(1)$. If the \emph{Suspense${}^{t_1}$-$\ell^{t_2}$} strategy is optimal, then
	\begin{enumerate}[(a), noitemsep]
		\item \textbf{FOC for $t_2$:} $\Delta v + \Delta h(t_2) = h_r'(t_2)$,
		\item \textbf{FOC for $t_1$:} $\Psi(t_1, t_2) \ge h_\ell'(t_1)$, with equality when $t_1 > 0$,
		\item \textbf{Local SOC:} $h_\ell''(t_1) \cdot t_1 \le 0$ and $h_r''(t_2) \le \Delta h'(t_2)$.
	\end{enumerate}
	Conversely, if conditions (a)--(c) hold and, in addition,
	\begin{enumerate}[(d), noitemsep]
		\item \textbf{Global SOC:}
		\(
			\begin{cases}
				h_\ell''(t), h_r''(t) \le 0 & \forall\, t < t_1, \\
				\max\{h_r''(t), 0\} \le \Psi_t(t, t_2) - h_\ell'(t) & \forall\, t \ge t_1,
			\end{cases}
		\)
	\end{enumerate}
	then the \emph{Suspense${}^{t_1}$-$\ell^{t_2}$} strategy is optimal.
\end{thm}
\begin{proof}
	See \cref{appendix:binary_proofs}.
\end{proof}

The conditions are exactly the first- and second-order conditions for the multiplier-augmented payoff $l_{f, \Lambda}$ to be concavified by a supporting hyperplane on the support of the conjectured strategy: (a) pins down the terminal time $t_2$, (b) pins down the suspense duration $t_1$, and (c)--(d) are second-order checks. By symmetry, an analogous characterization holds for Suspense${}^{t_1}$-$r^{t_2}$.

\begin{cor}[Optimality of Suspense-$r$]\label{cor:binary_suspense_R}
	Let $\Psi(t_1, t_2) = \int_{t_1}^{t_2} e^{-s + t_1} h_r'(s) \, \d s + e^{-t_2 + t_1} h_\ell'(t_2)$ and $0 \le t_1 < t_2 < \mu_{t_1}^{\dagger -1}(0)$. If the \emph{Suspense${}^{t_1}$-$r^{t_2}$} strategy is optimal, then
	\begin{enumerate}[(a), noitemsep]
		\item \textbf{FOC for $t_2$:} $\Delta v + \Delta h(t_2) = -h_\ell'(t_2)$,
		\item \textbf{FOC for $t_1$:} $\Psi(t_1, t_2) \ge h_r'(t_1)$, with equality when $t_1 > 0$,
		\item \textbf{Local SOC:} $h_r''(t_1) \cdot t_1 \le 0$ and $h_\ell''(t_2) \le -\Delta h'(t_2)$.
	\end{enumerate}
	Conversely, if conditions (a)--(c) hold and, in addition,
	\begin{enumerate}[(d), noitemsep]
		\item \textbf{Global SOC:}
		\(
			\begin{cases}
				h_\ell''(t), h_r''(t) \le 0 & \forall\, t < t_1, \\
				\max\{h_\ell''(t), 0\} \le e^{t - t_1}\left( h_r'(t_1) - \int_{t_1}^t e^{-s + t_1} h_r'(s) \, \d s \right) - h_r'(t) & \forall\, t \ge t_1,
			\end{cases}
		\)
	\end{enumerate}
	then the \emph{Suspense${}^{t_1}$-$r^{t_2}$} strategy is optimal.
\end{cor}
\begin{proof}
	See \cref{appendix:corollary_3_1}: \cref{cor:binary_suspense_R} follows from \cref{thm:binary_suspense_L} by switching state labels $L \leftrightarrow R$ and action labels $\ell \leftrightarrow r$.
\end{proof}

Together, \cref{thm:binary_suspense_L,cor:binary_suspense_R} pin down the three parameters $(t_1, a, t_2)$ and make the three structural trade-offs explicit. We discuss each in turn.

\paragraph{The persuasion gain shapes the scope of persuasion.} The terminal time $t_2$ is pinned down by condition (a), which equates the marginal additional persuasion gain to the marginal delay gain at $t_2$. The magnitude of $\Delta v$ relative to $\Delta h(t)$ governs its location. When $\Delta v$ is large, the principal maximizes the probability of action $r$, choosing an early $t_2$ with a tightly clustered terminal belief; in the limit $\Delta v \to \infty$, the optimum collapses to the static benchmark of \citet{kamenica2011bayesian}, with posteriors at $0$ and $1/2$. When $\Delta v$ is small relative to $\Delta h$, the principal delays $t_2$ and compensates the agent with a more dispersed terminal belief $\mu^*_{t_1}(t_2)$.

\paragraph{Time-risk preferences shape the duration of suspense.} The suspense duration $t_1$ is pinned down by condition (b). The function $\Psi(t_1, t_2)$ is the marginal loss from extending the suspense window: lengthening it by $\d t_1$ delays stopping but requires additional information over $[t_1, t_2]$ to compensate the agent. Fixing $t_2$ and the level $h_\ell'(t_1)$, suppose $h_\ell$ becomes more concave; then $h_\ell'$ is smaller for every $t > t_1$, which lowers $\Psi(t_1, t_2)$ and pushes $t_1$ upward. Hence \emph{more concave delay payoffs lengthen the suspense stage}. In the extreme of convex $h_\ell$, the local second-order condition rules out interior suspense: the principal sets $t_1 = 0$ and the optimum reduces to pure action-targeting. This is the binary manifestation of a much more general principle---in \cref{sec:time_risk} we show that the principal's time-risk attitude alone, independent of state and action spaces, determines whether suspense is ever optimal.

\paragraph{Persuasion-delay complementarity shapes the direction of targeting.} The choice between Suspense-$\ell$ and Suspense-$r$ hinges on whether persuasion and delay are complements or substitutes, as encoded by the monotonicity of $\Delta h$. Suspense-$\ell$ is optimal when $\Delta h$ is increasing (complementarity): delay enhances the marginal persuasion gain, so the principal benefits from \emph{positively correlating} action $r$ with late stopping. $\ell$-targeting achieves this: over $[t_1, t_2)$ the agent may receive conclusive news that $L$ is the state (early stopping, action $\ell$); absent such news, he stops at $t_2$ and takes $r$. The strategy maximizes the joint probability of late stopping \emph{and} action $r$. When $\Delta h$ is decreasing (substitutability)---delay erodes the marginal persuasion gain, although both $h_\ell, h_r$ remain increasing so the principal still enjoys delay---she instead \emph{negatively correlates} action $r$ with late stopping, which $r$-targeting achieves: action $r$ follows early Poisson arrival of conclusive news that $R$ is the state, while late stopping at $t_2$ yields $\ell$. The strategy maximizes the probability that \emph{either} action $r$ occurs \emph{or} stopping is late. \Cref{fig:direction_illustration} illustrates the two induced joint distributions. This endogenizes the ``good news''/``bad news'' arrival processes that the literature often assumes exogenously \citep{keller2005strategic,che2023keeping}, and delineates which one is optimal.

\begin{figure}[h]
	\centering
	\resizebox{0.7\linewidth}{!}{\tikzset{every picture/.style={line width=0.75pt}} 

\begin{tikzpicture}[x=0.75pt,y=0.75pt,yscale=-1,xscale=1]

\draw [color={rgb, 255:red, 208; green, 2; blue, 27 }  ,draw opacity=1 ]   (456.17,85.07) -- (505.74,85.07) ;
\draw    (270.48,183.92) -- (505.74,183.92) ;
\draw    (270.48,55.42) -- (270.48,183.92) ;
\draw [color={rgb, 255:red, 74; green, 144; blue, 226 }  ,draw opacity=1 ]   (270.48,181.92) .. controls (291.11,181.92) and (340.5,182.26) .. (340,182) .. controls (339.5,181.74) and (339.5,140.26) .. (340,140) .. controls (340.5,139.74) and (351.27,124.66) .. (380.39,112.2) .. controls (409.5,99.74) and (451.97,97.19) .. (455.33,97.19) .. controls (458.69,97.19) and (453.65,97.19) .. (505.74,97.19) ;
\draw [color={rgb, 255:red, 208; green, 2; blue, 27 }  ,draw opacity=1 ]   (270.48,182.51) .. controls (305.77,182.51) and (455.33,182.51) .. (455.33,182.51) .. controls (455.33,182.51) and (455.33,130.26) .. (455.33,85.07) ;
\draw  [color={rgb, 255:red, 208; green, 2; blue, 27 }  ,draw opacity=1 ][fill={rgb, 255:red, 208; green, 2; blue, 27 }  ,fill opacity=1 ] (452.81,85.07) .. controls (452.81,83.9) and (453.94,82.95) .. (455.33,82.95) .. controls (456.72,82.95) and (457.85,83.9) .. (457.85,85.07) .. controls (457.85,86.24) and (456.72,87.19) .. (455.33,87.19) .. controls (453.94,87.19) and (452.81,86.24) .. (452.81,85.07) -- cycle ;
\draw [color={rgb, 255:red, 74; green, 144; blue, 226 }  ,draw opacity=1 ]   (194.85,55.42) -- (244.43,55.42) ;
\draw    (10,182.51) -- (245.27,182.51) ;
\draw    (10,55.42) -- (10,182.51) ;
\draw [color={rgb, 255:red, 208; green, 2; blue, 27 }  ,draw opacity=1 ]   (10,181.1) .. controls (31.54,181.1) and (79.5,180.59) .. (80,181) .. controls (80.5,181.41) and (79.5,140.26) .. (80,140) .. controls (80.5,139.74) and (87.5,125.26) .. (117,109) .. controls (146.5,92.74) and (195.62,94.47) .. (198.98,94.47) .. controls (202.34,94.47) and (197.3,94.47) .. (249.4,94.47) ;
\draw [color={rgb, 255:red, 74; green, 144; blue, 226 }  ,draw opacity=1 ]   (10,181.1) .. controls (21.03,181.1) and (43.21,181.1) .. (68.5,181.1) .. controls (124.15,181.1) and (194.85,181.1) .. (194.85,181.1) .. controls (194.85,181.1) and (194.85,100.61) .. (194.85,55.42) ;
\draw  [color={rgb, 255:red, 74; green, 144; blue, 226 }  ,draw opacity=1 ][fill={rgb, 255:red, 74; green, 144; blue, 226 }  ,fill opacity=1 ] (192.33,55.42) .. controls (192.33,54.25) and (193.46,53.3) .. (194.85,53.3) .. controls (196.25,53.3) and (197.37,54.25) .. (197.37,55.42) .. controls (197.37,56.59) and (196.25,57.54) .. (194.85,57.54) .. controls (193.46,57.54) and (192.33,56.59) .. (192.33,55.42) -- cycle ;
\draw  [color={rgb, 255:red, 208; green, 2; blue, 27 }  ,draw opacity=1 ][fill={rgb, 255:red, 208; green, 2; blue, 27 }  ,fill opacity=1 ] (160,20) -- (180,20) -- (180,40) -- (160,40) -- cycle ;
\draw  [color={rgb, 255:red, 74; green, 144; blue, 226 }  ,draw opacity=1 ][fill={rgb, 255:red, 74; green, 144; blue, 226 }  ,fill opacity=1 ] (290,20) -- (310,20) -- (310,40) -- (290,40) -- cycle ;
\draw   (10.5,186) .. controls (10.5,190.67) and (12.83,193) .. (17.5,193) -- (33.75,193) .. controls (40.42,193) and (43.75,195.33) .. (43.75,200) .. controls (43.75,195.33) and (47.08,193) .. (53.75,193)(50.75,193) -- (70,193) .. controls (74.67,193) and (77,190.67) .. (77,186) ;
\draw   (79.5,186.5) .. controls (79.5,191.17) and (81.83,193.5) .. (86.5,193.5) -- (126,193.5) .. controls (132.67,193.5) and (136,195.83) .. (136,200.5) .. controls (136,195.83) and (139.33,193.5) .. (146,193.5)(143,193.5) -- (185.5,193.5) .. controls (190.17,193.5) and (192.5,191.17) .. (192.5,186.5) ;
\draw   (271,186.5) .. controls (271,191.17) and (273.33,193.5) .. (278,193.5) -- (294.25,193.5) .. controls (300.92,193.5) and (304.25,195.83) .. (304.25,200.5) .. controls (304.25,195.83) and (307.58,193.5) .. (314.25,193.5)(311.25,193.5) -- (330.5,193.5) .. controls (335.17,193.5) and (337.5,191.17) .. (337.5,186.5) ;
\draw   (340,187) .. controls (340,191.67) and (342.33,194) .. (347,194) -- (386.5,194) .. controls (393.17,194) and (396.5,196.33) .. (396.5,201) .. controls (396.5,196.33) and (399.83,194) .. (406.5,194)(403.5,194) -- (446,194) .. controls (450.67,194) and (453,191.67) .. (453,187) ;
\draw  [color={rgb, 255:red, 208; green, 2; blue, 27 }  ,draw opacity=1 ][fill={rgb, 255:red, 208; green, 2; blue, 27 }  ,fill opacity=1 ] (78.35,139.95) .. controls (78.35,138.85) and (79.24,137.95) .. (80.35,137.95) .. controls (81.45,137.95) and (82.35,138.85) .. (82.35,139.95) .. controls (82.35,141.06) and (81.45,141.95) .. (80.35,141.95) .. controls (79.24,141.95) and (78.35,141.06) .. (78.35,139.95) -- cycle ;
\draw  [color={rgb, 255:red, 74; green, 144; blue, 226 }  ,draw opacity=1 ][fill={rgb, 255:red, 74; green, 144; blue, 226 }  ,fill opacity=1 ] (338,140) .. controls (338,138.9) and (338.9,138) .. (340,138) .. controls (341.1,138) and (342,138.9) .. (342,140) .. controls (342,141.1) and (341.1,142) .. (340,142) .. controls (338.9,142) and (338,141.1) .. (338,140) -- cycle ;

\draw (334.72,188.81) node [anchor=north west][inner sep=0.75pt]  [font=\scriptsize] [align=left] {$\displaystyle t_{1}$};
\draw (450.67,189.22) node [anchor=north west][inner sep=0.75pt]  [font=\scriptsize] [align=left] {$\displaystyle t_{2}$};
\draw (497.72,190.63) node [anchor=north west][inner sep=0.75pt]  [font=\scriptsize] [align=left] {$\displaystyle \overline{t}$};
\draw (74.24,189.39) node [anchor=north west][inner sep=0.75pt]  [font=\scriptsize] [align=left] {$\displaystyle t_{1}$};
\draw (190.19,187.81) node [anchor=north west][inner sep=0.75pt]  [font=\scriptsize] [align=left] {$\displaystyle t_{2}$};
\draw (237.25,189.22) node [anchor=north west][inner sep=0.75pt]  [font=\scriptsize] [align=left] {$\displaystyle \overline{t}$};
\draw (21,172) node [anchor=north west][inner sep=0.75pt]  [font=\small] [align=left] {$ $};
\draw (11,54) node [anchor=north west][inner sep=0.75pt]  [font=\scriptsize] [align=left] {$\displaystyle  \begin{array}{{>{\displaystyle}l}}
persuasion-delay\ complementarity:\\
-\ late\ stopping\ on\ action\ r\\
-early\ stopping\ on\ action\ l
\end{array}$};
\draw (272,54) node [anchor=north west][inner sep=0.75pt]  [font=\scriptsize] [align=left] {$\displaystyle  \begin{array}{{>{\displaystyle}l}}
persuasion-delay\ substitutability:\\
-\ late\ stopping\ on\ action\ l\\
-\ early\ stopping\ on\ action\ r
\end{array}$};
\draw (189,22) node [anchor=north west][inner sep=0.75pt]  [font=\small] [align=left] {$\displaystyle \mathbb{P}( \tau \leq t,a=l)$};
\draw (318,22) node [anchor=north west][inner sep=0.75pt]  [font=\small] [align=left] {$\displaystyle \mathbb{P}( \tau \leq t,a=r)$};
\draw (13,201) node [anchor=north west][inner sep=0.75pt]  [font=\small] [align=left] {suspense};
\draw (103,199) node [anchor=north west][inner sep=0.75pt]  [font=\small] [align=left] {$\displaystyle \ell $-targeting};
\draw (277,201) node [anchor=north west][inner sep=0.75pt]  [font=\small] [align=left] {suspense};
\draw (367,199) node [anchor=north west][inner sep=0.75pt]  [font=\small] [align=left] {$\displaystyle r$-targeting};

\end{tikzpicture}}
	\caption{Persuasion-delay complementarity versus substitutability.}
	\label{fig:direction_illustration}
\end{figure}

\paragraph{Generalized completeness.} \Cref{thm:binary_suspense_L,cor:binary_suspense_R} characterize the two leading patterns; a natural question is whether they are exhaustive. The next proposition gives an affirmative answer; the formal statement, which extends the patterns to several corner cases, is in \cref{appendix:binary_completeness_formal}.
\begin{prop}[Informal]\label{prop:binary_completeness_informal}
	If $h_\ell$ and $h_r$ are concave, then a ``generalized'' Suspense-$\ell$ (Suspense-$r$) strategy is optimal whenever persuasion and delay are complements, $\Delta h' > 0$ (substitutes, $\Delta h' < 0$).
\end{prop}
The qualifier ``generalized'' refers to corner specifications in which the suspense stage or one targeting stage degenerates; the take-away is that, under concave delay payoffs and monotone $\Delta h$, the qualitative form of the optimum is determined entirely by the sign of $\Delta h'$.

\subsection{Discussion}\label{ssec:binary_discussion}

The three trade-offs combine in informative ways. Consider the prosecutor--jury setting. A conviction-seeking prosecutor in a high-profile case benefits from prolonging the trial---publicity raises the marginal payoff from conviction over time---so persuasion and delay are complements, and \cref{thm:binary_suspense_L} predicts $\ell$-targeting: the prosecutor searches for conclusive evidence of \emph{innocence} (action $\ell$, acquittal) to be used as a Poisson arrival that ends the trial; absent such evidence, the trial ends at $t_2$ with the preferred verdict of conviction (action $r$). Conversely, a prosecutor whose payoff is supplemented by inducing a plea bargain through delay faces persuasion and delay as substitutes, and \cref{cor:binary_suspense_R} predicts $r$-targeting: she searches for conclusive evidence of \emph{guilt} (action $r$) and accepts the plea outcome (action $\ell$) if the search drags on. The same preferences over the action produce qualitatively opposite arrival processes.

Finally, suspense generation here is conceptually related to, but distinct from, the suspense-optimal information of \citet{ely2015suspense}. There, the two belief paths are chosen so that the residual variance of the terminal belief declines at a constant rate, smoothing the agent's ``consumption'' of variance reduction. In our model, the two paths are chosen so that the agent's continuation value declines at a rate that exactly offsets his loss from delay, leaving zero interim surplus. The two notions coincide in special cases but rest on different motivations: ours is driven by tying the principal's hands intertemporally, theirs by smoothing variance reduction.\footnote{In frictional dynamic learning under time-risk aversion, the \citeauthor{ely2015suspense} pattern reemerges \citep[e.g.,][]{georgiadis2021information}; the time-risk attitude plays a structurally similar role, but the underlying strategy and economic force differ.}

\section{Application 2: General Time-Risk Preferences}\label{sec:time_risk}

The binary analysis identified time-risk preferences as the determinant of the \emph{duration} of suspense. We now turn the method on itself to show that this is a \emph{general structural feature} of the framework: regardless of the state and action spaces, the curvature in time of the principal's and agent's payoffs governs whether the optimal policy generates suspense at all. The two results below predict the qualitative shape of the solution before any computation, and together exhaust the curvature dichotomy. Throughout, $T = [0, \overline{T}]$ and $U, V$ are twice differentiable in $t$.

\subsection{Time-risk loving rules out suspense}\label{ssec:time_risk_loving}

Under time-risk loving preferences---convexity of $V$ and $U$ in time---the principal's marginal gain from delay accelerates and the agent's marginal loss decelerates. The duality recipe then leaves no slack in which to insert a suspense interval: the multiplier $\Lambda$ must charge the obedience constraint continuously throughout the support of $f$, equivalently \eqref{eqn:OCC} binds everywhere on that support. Any pre-stopping interval on which the agent retains strict surplus would be dominated by reallocating mass toward a stopping time at which $V$ is more convex.

\begin{thm}[Time-risk loving]\label{thm:time_risk_loving}
Suppose $(V''_t > 0,\ U''_t \geq 0)$ or $(V'_t > 0,\ V''_t \geq 0,\ U''_t > 0)$. If $f$, $\Lambda$ solve \eqref{eqn:dual} and $\overline{t} = \sup \mathrm{supp}(f) < \overline{T}$, then \eqref{eqn:OCC} must be binding for all $t \le \overline{t}$.
\end{thm}
\begin{proof}
See \cref{appendix:time_risk_loving_proof}.
\end{proof}

\Cref{thm:time_risk_loving} has two interlocking implications. First, the implementing simple recommendation is \emph{dynamically consistent}: \eqref{eqn:OCC} binding throughout the support is exactly the zero-interim-surplus property that, by \citet{koh2024attention}, suffices for subgame-perfect implementability under limited commitment---the principal loses nothing from being unable to commit, since the agent follows the direct ``continue'' instruction at every interim time. Second, the binary-case analog is a Suspense${}^{0}$ strategy: setting $t_1 = 0$ in \cref{thm:binary_suspense_L,cor:binary_suspense_R} recovers exactly the no-suspense action-targeting strategy. A canonical example is exponential discounting: if $V(\mu, t) = e^{-rt} v(\mu)$ and $U(\mu, t) = e^{-r't} u(\mu)$ with $u > 0$ and $v < 0$, then $V''_t > 0$ and $U''_t > 0$, so the optimal policy contains no suspense stage. Affine-in-time specifications are the borderline case ($V''_t = U''_t = 0$), where the theorem holds with weak inequality. Both classes are pervasive in macroeconomic and financial applications.

\subsection{Time-risk aversion bounds the persuasion window}\label{ssec:time_risk_averse}

The complementary case---concavity of $V$ and $U$ in time---reverses the geometry. The principal's marginal gain from delay diminishes; the agent's marginal cost of delay rises. The optimum cannot pack all stopping mass into a single instant, because the principal's payoff is concave in time, but \eqref{eqn:OCC} cannot remain tight indefinitely either. The persuasion window is therefore non-degenerate and, as the theorem shows, bounded. We need a mild regularity condition.

\begin{defi}\label{defi:regular}
Suppose $f$ and $\Lambda$ solve \eqref{eqn:dual}. Let $\overline{t} = \sup_t \mathrm{supp}(f)$. $f, \Lambda$ are \emph{regular} if (i) $\exists\, \mu^* = \lim_{t \to \overline{t}^-} \widehat{\mu}_t$, (ii) $(\mu^*, \overline{t}) \in \mathrm{supp}(f)$, and (iii) $\dfrac{\Lambda(\overline{t}) - \Lambda(t)}{\overline{t} - t}$ is bounded for $t \to \overline{t}^-$.
\end{defi}

This requires the continuation belief $\widehat{\mu}_t$ to converge to a belief-time pair in the support of $f$ as $t \to \overline t$, with the multiplier accumulating at a bounded rate; it holds in all our applications. Define
\begin{align*}
\begin{dcases}
\overline{J}_V(t) = \max_{\mu \in \Delta(\Theta)} V'_t(\mu, t) \\
\underline{J}_V(t) = \min_{\mu \in \Delta(\Theta)} V'_t(\mu, t)
\end{dcases}
\quad \text{and} \quad
\begin{dcases}
\overline{J}_U(t) = \max_{\mu \in \Delta(\Theta)} U'_t(\mu, t) \\
\underline{J}_U(t) = \min_{\mu \in \Delta(\Theta)} U'_t(\mu, t)
\end{dcases},
\end{align*}
the range of $V'_t$ and $U'_t$ across beliefs at $t$; when $V$ is separable in $\mu$ and $t$, $\overline{J}_V = \underline{J}_V$.

\begin{thm}[Time-risk averse]\label{thm:time_risk_averse}
Suppose $V''_t < 0$, $U''_t < 0$, and regular $f$ and $\Lambda$ (per \cref{defi:regular}) solve \eqref{eqn:dual} with $\underline{t} = \inf_t \mathrm{supp}(f) > 0$. Then
\begin{align*}
\overline{t} \le \max\left\{ \overline{J}_V^{-1} \circ \underline{J}_V(\underline{t}),\ \overline{J}_U^{-1} \circ \underline{J}_U(\underline{t}) \right\},
\end{align*}
where the max is taken over the image of the correspondences $\overline{J}_V^{-1}$ and $\overline{J}_U^{-1}$.
\end{thm}
\begin{proof}
See \cref{appendix:time_risk_averse_proof}.
\end{proof}

The bound measures the time-concavity of $V$ and $U$ relative to their variation across beliefs: the more concave $U$ and $V$ are in $t$, and the less their slopes vary across beliefs, the smaller the persuasion window; in the separable extreme it must be degenerate. \Cref{thm:time_risk_averse} again carries two implications. First, a non-degenerate suspense interval $[0, \underline{t})$ necessarily precedes any stopping; if the agent is strictly impatient, \eqref{eqn:OCC} is slack on $[0, \underline{t})$ and the implementing simple recommendation is dynamically \emph{inconsistent}. By \citet{koh2024attention}, the principal must then implement the commitment optimum via informative indirect interim messages that exhaust the agent's interim surplus, tying her hands on the continuation game. Second, the Suspense${}^{t_1}$ strategies of \cref{sec:binary} are exactly the binary manifestation: $t_1 > 0$ is the pre-stopping interval $[0, \underline{t})$, and the belief paths $\mu^L_t, \mu^R_t$ realize the indirect-message implementation in the binary state space.

\subsection{Discussion}\label{ssec:time_risk_discussion}

\Cref{thm:foc} and \cref{thm:time_risk_loving,thm:time_risk_averse} play complementary roles: the first tells us \emph{how} to solve the design problem (write the Lagrangian, derive the multiplier, concavify), the latter two tell us \emph{what form} the solution takes as a function of time preferences alone. The three trade-offs of \cref{sec:binary} are consequences of \cref{thm:time_risk_averse} plus the binary-action geometry: the ``time-risk shapes duration'' trade-off is exactly \cref{thm:time_risk_averse} specialized to a binary state space, with $t_1$ the binary realization of $[0, \underline{t})$. Time-risk preferences are not a binary phenomenon; the binary application merely makes them visible in closed form. \Cref{thm:time_risk_loving,thm:time_risk_averse} are stated for a finite state space, but a continuum analog is implicit in \cref{sec:linear}: the ``investor indifference'' property of the optimal tail-censorship policy is the continuum counterpart of \eqref{eqn:OCC} binding for a time-discounted impatient agent. The qualitative dichotomy---loving rules out suspense, aversion necessitates it---is robust, though the exact bound on $\overline{t}$ is state-space dependent.

\section{Application 3: Linear Persuasion}\label{sec:linear}

Our third application is a \emph{linear persuasion} environment, in which the state lies in a continuous interval and both players' indirect utilities depend on the belief only through its mean. The static counterpart is extensively studied \citep{gentzkow2016rothschild,dworczak2019simple,kolotilin2018optimal,kolotilin2022censorship,kleiner2021extreme,yang2024design}; we develop its dynamic version. Two results follow. First, the first-order characterization extends to the continuous state with a transparent structural feature: the dual \emph{decouples} the time-coupled obedience constraint from the belief-coupled mean-preservation constraint. Second, in a canonical central-bank application, the optimal disclosure is \emph{dynamic tail-censorship}.

\subsection{Continuous-state framework and the first-order characterization}\label{ssec:linear_framework}

\paragraph*{Setup.} The framework of \cref{sec:framework} extends naturally to $\Theta = [0, 1]$, with prior $\mu_0 \in \Delta(\Theta)$ a Borel probability measure on $[0,1]$. The flow utilities are linear in the state, so $U$ and $V$ depend on a belief $\mu$ only through its mean $\E_\mu[\theta]$; we index indirect utilities by the posterior mean, writing $\mm \in \MM := [0, 1]$ and treating $U(\mm, t), V(\mm, t)$ as functions on $\MM \times T$. Pooling continuation messages reduces the problem to choosing a joint distribution over posterior means and stopping times, $g \in \overline{\Delta}_{\mu_0} := \{g \in \Delta(\MM \times T) \mid \int \mm \, g(\d \mm, \d t) = \E_{\mu_0}[\theta]\}$. Two constraints structure feasibility: the conditional obedience constraint of \cref{ssec:relaxed}, reformulated in posterior means, and---new to the continuum---a Bayes-plausibility (mean-preserving-spread) requirement.

The relaxed problem is
\begin{align}
	\sup_{g \in \overline{\Delta}_{\mu_0}} & \int V(\mm, t) g(\d \mm, \d t) \tag{R1} \label{prob:cont:relaxed} \\
	\text{s.t.} \quad & \int_{y > t} U(\mm, y) g(\d \mm, \d y) \ge \int_{y > t} g(\d \mm, \d y) \cdot U\!\left( \frac{\int_{y > t} \mm \, g(\d \mm, \d y)}{\int_{y > t} g(\d \mm, \d y)}, t \right), \quad \forall t \in T^\circ, \tag{OC-C}\label{eqn:OCC-C} \\
	& \int_0^{\mm} \!\!\int_0^x \!\left( \mu_0(\d z) - \int_{t \in T} g(\d z, \d t) \right) \d x \ge 0, \quad \forall \mm \in [0, 1]. \tag{MPS}\label{eqn:MPS}
\end{align}
\eqref{eqn:OCC-C} is the continuous-state obedience constraint; \eqref{eqn:MPS} encodes mean-preserving-spread feasibility. Together they characterize implementability: a revelation-principle argument that simultaneously collapses continuation messages and continuation beliefs establishes the equivalence of \eqref{prob:1} and \eqref{prob:cont:relaxed}.\footnote{See \cref{appendix:reduction_continuum} for the continuum revelation argument; it generalizes \cref{lem:direct_communication} by additionally collapsing continuation beliefs, the extra reduction step needed in the linear environment.}

\paragraph*{The dual.} Define the Lagrangian $\mathcal{L} : \overline{\Delta}_{\mu_0} \times \mathcal{B}(T^\circ) \times \mathcal{C} \to \R$ by
\begin{align*}
	\mathcal{L}(g, \Lambda, \Gamma) := & \int V(\mm, \tau) g(\d \mm, \d \tau) + \int_{t \in T^\circ} \!\!\left( \int_{\tau > t} U(\mm, \tau) g(\d \mm, \d \tau) - \int_{\tau > t} g(\d \mm, \d \tau) \cdot U(\widehat{\mm}_t, t) \right) \d \Lambda(t) \\
	& + \int_0^1 \!\!\left( \int_0^\mm \!\!\int_0^x \!\left( \mu_0(\d y) - \int_{\tau \in T} g(\d y, \d \tau) \right) \d x \right) \d^2 \Gamma(\mm),
\end{align*}
where $\widehat{\mm}_t := \E_g[\mm \mid \tau > t]$, $\Lambda \in \mathcal{B}(T^\circ)$ is as in \cref{lem:duality}, and $\Gamma \in \mathcal{C}$ is a convex function on $[0, 1]$ whose second derivative (a positive measure) is the multiplier on \eqref{eqn:MPS}. The derivative of the Lagrangian along $(\mm, t)$ is
\begin{align*}
	l_{g, \Lambda, \Gamma}(\mm, t) := V(\mm, t) + \Lambda(t) U(\mm, t) - \int_{\tau < t} \!\left[ U(\widehat{\mm}_\tau, \tau) + U'_\mm(\widehat{\mm}_\tau, \tau) (\mm - \widehat{\mm}_\tau) \right] \d \Lambda(\tau) - \Gamma(\mm),
\end{align*}
adopting the convention that $U'_\mm(\widehat{\mm}_\tau, \tau)$ denotes any subgradient selection at non-differentiability. The Lagrangian specializes to that of \cref{lem:duality} when $\Gamma \equiv 0$; the additional $-\Gamma(\mm)$ term encodes \eqref{eqn:MPS}.

\begin{thm}[First-order characterization for the continuum]\label{thm:foc_continuum}
	A joint distribution $g \in \overline{\Delta}_{\mu_0}$ solves \eqref{prob:cont:relaxed} if there exist $\Lambda \in \mathcal{B}(T^\circ)$, a time-invariant convex function $\Gamma \in \mathcal{C}$, and a selection of subgradients such that
	\begin{align}
		l_{g, \Lambda, \Gamma}(\mm, t) \le 0, \quad \text{with equality on the support of } g, \tag{FOC-C} \label{eqn:foc-c}
	\end{align}
	$g$ satisfies \eqref{eqn:OCC-C} and \eqref{eqn:MPS}, and the complementary slackness condition $\mathcal{L}(g, \Lambda, \Gamma) = \E_g[V(\mm, \tau)]$ holds.
\end{thm}
\begin{proof}
	See \cref{appendix:foc_continuum_proof}.
\end{proof}

The structural feature worth emphasizing is that the two multipliers serve qualitatively different roles: $\Lambda$ is \emph{belief-invariant} (a function of time only) and $\Gamma$ is \emph{time-invariant} (a function of the posterior mean only). The dual thus \emph{decouples} the time-coupled obedience constraint from the belief-coupled mean-preservation constraint: $\Lambda$ takes the same shape as in \cref{thm:foc}, while $\Gamma$ is borrowed from the linear-programming approach to static linear persuasion \citep{kolotilin2018optimal,dworczak2019simple,kolotilin2022censorship,kleiner2021extreme,yang2024design}. This decoupling is the technical reason the dynamic problem inherits the tractability of its static counterpart: the recipe of \cref{ssec:recipe} applies separately along each dimension, with $\Lambda$ pinned down by a one-dimensional ODE in time and $\Gamma$ by the static-style concavification in $\mm$ at each cross-section.

\subsection{Optimal dynamic tail-censorship}\label{ssec:tail}

We apply the recipe of \cref{ssec:recipe} with \cref{thm:foc_continuum} to a canonical macroeconomic environment: dynamic disclosure of fundamentals to an impatient investor facing a real-option choice.

\paragraph*{The central-bank--investor problem.} The state $\theta \in [0, 1]$ is the economy's fundamental. A representative investor (the agent) faces a real-option choice between a safe asset yielding $1/2$ and a risky asset yielding the state-dependent return $\theta$. The investor is impatient and discounts at rate $r > 0$, so $U(\theta, t) = e^{-r t} \max\{\theta, 1/2\}$.\footnote{An equivalent real-options formulation has the agent be a firm choosing between a risky project with payout $\theta$ and a safe action (e.g., buybacks or dividends) paying $1/2$, with $r$ the opportunity cost of waiting in the sense of \cite{stokey2008economics}.} The investor's \emph{wait-and-see} motive is a canonical source of investment fluctuations \citep{bernanke1983irreversibility,mcdonald1986value,dixit1994investment,stokey2016wait}. The central bank (the principal) sequentially conducts stress tests to disclose information about $\theta$. Its preferences are countercyclical---it wishes to slow investment when the economy is overheating ($\theta > \underline{\theta}$) and accelerate it in contraction ($\theta < \underline{\theta}$)---which we capture by $V(\mm, t) = (\mm - \underline{\theta}) \cdot t$. Let $\mu(\cdot)$ denote the prior PDF and $\mathcal{M}(\cdot)$ its CDF.

\paragraph*{Tail-censorship.} A \emph{tail-censorship} policy reveals the state only if it lies within a deterministically evolving interval $[\alpha(t), \beta(t)]$; silence at $t$ conveys that the state is in one of the tails $[0, \alpha(t)) \cup (\beta(t), 1]$. The interval expands over time ($\alpha$ falls, $\beta$ rises), so persistent silence becomes increasingly informative about the tails (\cref{fig:tail:censor}).

\begin{figure}[htbp]
	\centering
	\tikzset{every picture/.style={line width=0.75pt}}
	\begin{tikzpicture}[x=0.75pt,y=0.75pt,yscale=-1,xscale=1]
	\draw   (75.75,55.33) -- (405.75,55.33) -- (405.75,166.85) -- (75.75,166.85) -- cycle ;
	\draw [color={rgb, 255:red, 0; green, 122; blue, 255 }  ,draw opacity=1 ]   (75.75,90.54) .. controls (161.25,84.85) and (262.5,73.11) .. (330.75,55.33) ;
	\draw [color={rgb, 255:red, 0; green, 122; blue, 255 }  ,draw opacity=1 ]   (75.75,143.37) .. controls (111.75,156.46) and (198.75,165.85) .. (323.25,166.85) ;
	\draw   (233.25,162.89) .. controls (237.92,162.89) and (240.25,160.56) .. (240.25,155.89) -- (240.25,129) .. controls (240.25,122.33) and (242.58,119) .. (247.25,119) .. controls (242.58,119) and (240.25,115.67) .. (240.25,109)(240.25,112) -- (240.25,82.1) .. controls (240.25,77.43) and (237.92,75.1) .. (233.25,75.1) ;
	\draw   (77.25,141.87) .. controls (81.92,141.87) and (84.25,139.54) .. (84.25,134.87) -- (84.25,126.88) .. controls (84.25,120.21) and (86.58,116.88) .. (91.25,116.88) .. controls (86.58,116.88) and (84.25,113.55) .. (84.25,106.88)(84.25,109.88) -- (84.25,98.89) .. controls (84.25,94.22) and (81.92,91.89) .. (77.25,91.89) ;
	\draw (402.75,164.61) node [anchor=north west][inner sep=0.75pt]    {$t$};
	\draw (64.5,159.25) node [anchor=north west][inner sep=0.75pt]    {$0$};
	\draw (64.5,47.73) node [anchor=north west][inner sep=0.75pt]    {$1$};
	\draw (255.75,109.6) node [anchor=north west][inner sep=0.75pt]  [font=\small] [align=left] {Revealed up to $\displaystyle t$};
	\draw (96.75,100.97) node [anchor=north west][inner sep=0.75pt]  [font=\small] [align=left] {Revealed \\at $\displaystyle t=0$};
	\draw (178,142.4) node [anchor=north west][inner sep=0.75pt]  [font=\small]  {$\alpha ( t)$};
	\draw (181,62.4) node [anchor=north west][inner sep=0.75pt]  [font=\small]  {$\beta ( t)$};
	\end{tikzpicture}
	\caption{Illustration of a tail-censorship policy.}
	\label{fig:tail:censor}
\end{figure}

\begin{thm}[Optimality of tail-censorship]\label{thm:tail_censorship}
	The optimal policy in the central-bank--investor problem is a tail-censorship policy with thresholds $\alpha$ and $\beta$ satisfying the ODE pair
	\begin{align}
		\begin{dcases}
			\alpha'(t) &= -\frac{\mathcal{M}(\alpha(t)) + 1 - \mathcal{M}(\beta(t))}{\mu(\alpha(t)) (1/2 - \alpha(t))} \cdot r \cdot \beta(t) \\
			\beta'(t) &= \frac{\beta(t) - \alpha(t)}{1/2 - \alpha(t)} \cdot r \cdot \beta(t),
		\end{dcases}
		\label{eqn:ODE_continuum}
	\end{align}
	with terminal conditions $\alpha(\overline{T}) = 0$, $\beta(\overline{T}) = 1$ and initial condition $\beta(0) = \max\{\underline{\theta}, \theta^*\}$ for some $\overline{T} > 0$ and $\theta^* \in [1/2, (1 + \E[\theta \mathbf{1}\{\theta < 1/2\}]) / (1 + \mathcal{M}(1/2))]$.
\end{thm}
\begin{proof}
	See \cref{appendix:tail_proof}.
\end{proof}

Two characterizing properties capture the economics of the policy. \emph{No upward surprise:} revelation at the upper threshold never induces an upward jump in the posterior mean; equivalently the interim mean equals the upper threshold, $\widehat{\mm}_t = \beta(t)$. Distinguishing states above $1/2$ has no informational value for the agent (it does not change his optimal action), so delaying or expediting their revelation affects his incentive only through the symmetric discounting loss; at the optimum the principal is indifferent across delays, expediting states $\theta < \widehat{\mm}_t$ and delaying $\theta > \widehat{\mm}_t$ until $\widehat{\mm}_t = \beta(t)$. \emph{Investor indifference:} the agent is indifferent between stopping at $t$ and continuing to $t + \d t$, which translates into
\begin{align*}
	\underbrace{(\mathcal{M}(\alpha(t)) + 1 - \mathcal{M}(\beta(t)))}_{\text{Continuation probability}} \cdot \underbrace{r \widehat{\mm}_t \, \d t}_{\text{Discounting loss}} + \underbrace{(1/2 - \alpha(t))}_{\text{Gain from info}} \cdot \underbrace{\d \mathcal{M}(\alpha(t))}_{\text{Prob.\ of arrival}} = 0,
\end{align*}
balancing the discounting loss against the instrumental value of the marginal revelation of $\theta = \alpha(t)$.

\begin{remark}[Dynamic consistency]\label{rmk:dynamic_consistency_linear}
	Investor indifference is exactly the dynamic-consistency property of \citet{koh2024attention}: it leaves the agent with zero interim surplus at every continuation time along the optimal path, so by the indirect-communication principle the tail-censorship policy is SPE-implementable under limited commitment.
\end{remark}

\paragraph*{Economic interpretation.} The principal would naively prefer the agent to stop late when $\theta > \underline{\theta}$ and early when $\theta < \underline{\theta}$. This first-best is incentive-incompatible: a recommendation to wait would reveal $\theta > \underline{\theta}$, after which information has no value and the agent stops. The principal must therefore \emph{stagger} the revelation of low states to sustain the incentive to wait. The lower threshold $\alpha(t)$ reveals low states with high informational value $|1/2 - \theta|$ latest, trading discounting loss against instrumental value at each instant; the upper threshold $\beta(t)$ does the dual job for high states, revealing those with the smallest delay gain (close to $\underline{\theta}$) earliest and those with the largest (close to $1$) latest.

\subsection{Real-world institutions}\label{ssec:linear_institutions}

The expanding-revelation-interval pattern resonates with disclosure practice in monetary policy and corporate finance, where designers visibly withhold information in extreme states and reveal in moderate ones. At the onset of COVID-19, the Federal Reserve \emph{withheld} its Summary of Economic Projections entirely in March 2020 and resumed only partially in June \citep{FOMC2020SEP,FOMC2020Minutes}; the Bank of Canada's April 2020 \emph{Monetary Policy Report} declined to present a base-case projection given the uncertainty \citep{BoC2020MPR}; and the Bank of England replaced point forecasts with an illustrative scenario \citep{BoE2020MPR,BoE2020MayMinutes}. Public companies similarly \emph{reaffirm} earnings guidance mid-quarter when results lie in the anticipated band, while silence is read as a material miss or beat \citep{billings2015guidance,wsj2024guidance}; a delayed or omitted dividend declaration is immediately interpreted as a negative tail event \citep{reuters2023dividends}. In each case the institution discloses in moderate states and stays silent in extreme states, with the revelation interval expanding over time---the pattern \cref{thm:tail_censorship} captures.

\subsection{Discussion}\label{ssec:linear_discussion}

The dual decoupling is what makes the dynamic linear-persuasion problem tractable: the cross-sections of the optimal joint distribution at each $t$ are static censorship policies whose thresholds evolve according to the dynamic multiplier. A natural variant has the principal prefer \emph{late} investment regardless of the state (uniform stimulus rather than business-cycle smoothing); we conjecture the optimum is then a \emph{dynamic upper-censorship} policy in which $\beta$ is flat and only $\alpha$ falls over time, and the same recipe applies with a single-threshold ODE for $\alpha$. More broadly, the duality method of \cref{thm:foc} and its continuum analog \cref{thm:foc_continuum} apply to any linear-persuasion environment with flexibly specified preferences over the posterior mean and time; the decoupling reduces the dynamic problem to two essentially one-dimensional sub-problems, opening a tractable path to a wide range of dynamic disclosure problems with continuous state.

\section{Extensions}\label{sec:extensions}

The framework and the duality method extend to several variations of the basic setup; each fits inside \eqref{prob:relaxed} after a localized modification of the feasibility set or the principal's effective payoff, and \cref{thm:foc} applies to the modified problem without change. We highlight two.

\subsection{Evolving state and public information}\label{ssec:evolving_state}

Suppose there are finitely many state variables $\{\theta_{t_i}\}$ and the principal can design experiments only about variables already realized ($\theta_{t_i}$ may enter an experiment at $t$ only if $t_i \le t$). This models gradual learning about a complex state, and accommodates the interpretation in which each $\theta_{t_i}$ is an increment of a process and $X_t = \sum_{t_i \le t} \theta_{t_i}$ is an evolving observable. Let $\widehat{\theta} = (\theta_{t_i})$, $\Theta$ its state space, and $\nu_t(\widehat{\theta})$ the posterior conditional on realizations up to $t$. When the $\theta_{t_i}$ are unobservable to both players, the belief process is restricted to $\overline{D} := \{(\mu, t) \in D \mid \mu \in \mathrm{conv}(\bigcup_\theta \{\nu_t(\theta)\})\}$; when all variables except $\theta_0$ are public (as in \citealp{orlov2020persuading,bizzotto2021dynamic}), it is restricted to $\underline{D} := \{(\mu, t) \in D \mid \mu \in \bigcup_{\theta_{t_i}} \mathrm{conv}(\bigcup_{\theta_0}\{\nu_t(\theta)\})\}$. Since the time-section is monotone in $t$ under inclusion, the direct-communication reduction of \cref{appendix:reduction} carries over and \cref{thm:foc} applies to the restricted problem directly.

\begin{eg}[Beeps]\label{eg:beeps}
For each $t_i$, $\theta_{t_i} \in \{0, 1\}$; conditional on $\theta_{t_{i-1}} = 0$ the probability that $\theta_{t_i} = 1$ is $1 - e^{-\gamma(t_i - t_{i-1})}$, and $\theta_{t_i} = 1$ absorbs. The variable describes whether an email has arrived by $t_i$ (Poisson rate $\gamma$), observable to the principal once $t \ge t_i$. The agent chooses when to stop working and check, with $U(\mu, t) = e^{-rt} \mathrm{Prob}_\mu(\sum_{t_i \le t} \theta_{t_i} > 0)$ and $V(\mu, t) = t$. This is a forward-looking-agent variant of \citet{ely2017beeps} and a special case of \eqref{prob:relaxed} over $\overline D$.
\end{eg}

\begin{eg}[Persuading the agent to wait]\label{eg:wait}
Let $T = \{0, \ldots, I\}$. $\theta_0 \in \{\theta_L, \theta_H\}$ is project quality; for $t \ge 1$, $\theta_t \sim N(0, \sigma)$ are i.i.d.\ public increments, $X_t = (m + \theta_t) X_{t-1}$ the observed market size. The agent decides whether to exercise a real option, paying $I_A$ for payoff $\theta_0 X_t$, with $U(\mu, t) = \E_\mu[\max_{a \in \{0,1\}} a \, e^{-rt}(\theta_0 X_t - I_A)]$ and $V(\mu, t) = \E_\mu[a(\mu, t) e^{-rt}(\theta_0 X_t - I_P)]$. This is the main model of \citet{orlov2020persuading} and a special case of \eqref{prob:relaxed} over $\underline D$.
\end{eg}

\subsection{Costly or constrained information generation}\label{ssec:costly_info}

\paragraph{Costly information.} Suppose the principal bears a separable flow cost $c_t := \d \E[H(\mu_t) \mid \mathcal{F}_t]/\d t$ of experimentation, where $H$ is convex on $\Delta(\Theta)$ (canonically, $-H$ is Shannon entropy). By uniform posterior separability, $\E[\int_0^\tau c_t \, \d t] = \E[H(\mu_\tau) - H(\mu_0)]$ \citep{steiner2017rational,zhong2022optimal,hebert2022engagement}, so redefining $\widehat V(\mu, t) := V(\mu, t) - H(\mu)$ embeds the cost into the baseline problem and \cref{thm:foc} applies to $\widehat V$ directly.

\paragraph{Constrained information.} A more delicate variant imposes a flow-information capacity constraint $\d \E[H(\mu_t) \mid \mathcal{F}_t]/\d t \le \chi$, the standard constraint in the rational-inattention literature. Because simple recommendations are typically not smooth in the flow of information, this renders the relaxed and original problems non-equivalent in general. Using the implementability characterization of \citet{sannikov2024embedding}, the constrained relaxed problem adds
\begin{align}
	\int_{y \le t} H(\mu) f(\d \mu, \d y) + H\!\left(\int_{y > t} \mu f(\d \mu, \d y)\right) - H(\mu_0) \le \chi \int \min\{t, y\} f(\d \mu, \d y), \quad \forall t \in T, \tag{ICC-C}\label{eqn:icc_c}
\end{align}
a convex constraint, so the problem remains a linear program and \cref{thm:foc} extends with an additional multiplier. If the optimum keeps \eqref{eqn:icc_c} binding for every $t$, the simple recommendation is feasible in the original dynamic problem, embedding the engagement-maximization setting of \citet{hebert2022engagement} into our framework.

\section{Conclusion}\label{sec:conclude}

We have developed a duality-based first-order approach to optimal dynamic persuasion in stopping problems. Strong duality (\cref{lem:duality}) and the first-order condition (\cref{thm:foc}) reduce the infinite-dimensional design problem to a one-dimensional ordinary differential equation for a Lagrange multiplier $\Lambda$ tracking the shadow value of continuation incentives; the recipe converts dynamic persuasion into ``concavifying'' a $\Lambda$-augmented payoff and is operational on a broad class of problems. We demonstrated the method in three applications. In binary persuasion, optimal strategies combine suspense generation with action-targeting, and three sharp trade-offs determine the scope, duration, and direction of persuasion. A structural result (\cref{thm:time_risk_loving,thm:time_risk_averse}) showed that the principal's time-risk preferences alone determine the qualitative form of the optimum---loving rules out suspense, aversion necessitates it. In linear persuasion, the dual decouples time and belief constraints, and the optimal disclosure is dynamic tail-censorship with an expanding revelation interval that resonates with the communication practices of monetary authorities and corporate financial disclosure. Non-common-prior extensions in the spirit of \citet{saeedi2024getting}, and the action-dependent variants of \cref{ssec:linear_discussion}, are natural next steps. Finally, \citet{koh2024attention} establishes that the commitment outcomes derived here survive the removal of intertemporal commitment via an indirect-communication principle; together, the duality method of this paper and that principle provide a complete toolkit for optimal dynamic persuasion in stopping problems.

\fontsize{10pt}{12pt}\selectfont
\setlength{\bibsep}{-2.5pt}
\bibliography{references}

\newgeometry{margin=0.8in}
\linespread{1.2}
\normalsize
\appendix
\setlength\abovedisplayskip{1pt}
\setlength\belowdisplayskip{1pt}


\section{Direct Communication and Strong Duality}\label{appendix:reduction}

This appendix collects the technical underpinnings of \cref{ssec:relaxed,ssec:foc}. \Cref{appendix:reduction_proof} formalizes the simple-recommendation construction sketched in \cref{ssec:relaxed} and proves the direct-communication reduction (\cref{lem:direct_communication}) establishing that the original dynamic problem \eqref{prob:1} and the relaxed problem \eqref{prob:relaxed} have the same value and the same outcome set. \Cref{appendix:duality_proof} proves strong duality (\cref{lem:duality}). \Cref{appendix:reduction_continuum} extends the direct-communication reduction to the continuum-state framework of \cref{sec:linear}.

\subsection{Simple recommendations and the direct-communication reduction}\label{appendix:reduction_proof}

We begin by formalizing the class of ``direct mechanisms'' alluded to in \cref{ssec:relaxed}. The construction takes a candidate joint distribution $f$ over stopping beliefs and stopping times and realizes $f$ via a particularly simple persuasion strategy: the principal continuously sends a pooled ``continue'' message and only releases informative messages at the instant of stopping. We then show that such simple recommendations exhaust all outcome distributions implementable by feasible obedient strategies in \eqref{prob:1}.

Let $D = \Delta(\Theta) \times T$. We call elements of $\Delta(D)$ belief-time distributions, and we recall from \cref{ssec:relaxed} that $\Delta_{\mu_0} = \{f \in \Delta(D) \mid \E_f[\mu] = \mu_0\}$ collects belief-time distributions consistent with the prior. Fixing $f \in \Delta_{\mu_0}$, let $\bar{t} = \sup \{t \mid (\mu, t) \in \mathrm{Supp}(f)\}$. For $t < \bar{t}$, define $\widehat{\mu}_t = \E_{f(\mu, \tau)}[\mu \mid \tau > t]$ as the continuation belief at time $t$ conditional on not stopping. For all $\mu \in \Delta(\Theta)$ and $t \le \bar{t}$, let $x^{(\mu, t)}$ denote the belief path
\begin{align*}
	x^{(\mu, t)}(s) =
	\begin{cases}
		\widehat{\mu}_s & s < t \\
		\mu & s \ge t.
	\end{cases}
\end{align*}
That is, we consider simple belief paths that jump from $\widehat{\mu}_t$ only once (which the agent reads as a recommendation to stop) and stay constant thereafter. Since each path $x^{(\mu, t)}$ leads to stopping at $(\mu, t)$, to achieve the belief-time distribution $f$, the persuasion strategy simply needs to randomize among all such paths according to distribution $f$. Define $\Phi : (\mu, t) \mapsto (x^{(\mu, t)}, t)$ as the map sending a belief-time pair to the path constructed above together with its jump time. The following lemma ensures that $\Phi$ is Borel-measurable, so the law of the induced belief process and stopping time is well defined.

\begin{lem}\label{lem:measurablepath}
	Let $\left\{ x^{(\mu, t)} \right\} \subset D_{\infty}$ be equipped with the Skorokhod topology. The mapping $\Phi : (\mu, t) \mapsto (x^{(\mu, t)}, t)$ is Borel measurable.
\end{lem}

\begin{proof}
	Let $\sigma$ be the Skorokhod metric induced by a metric $d_\Theta$ on $\Delta(\Theta)$. Then $\forall (\mu, t), (\mu', t')$ where $t < t'$:
	\begin{align*}
		\sigma(x^{(\mu, t)}, x^{(\mu', t')}) \le \max\left\{ d_\Theta(\mu, \mu'), |t - t'|, \max_{s, s' \in [t, t')}\left\{ d_\Theta(\widehat{\mu}_s, \widehat{\mu}_{s'}) \right\} \right\},
	\end{align*}
	where the right-hand side is achieved via contracting time on $[t, t')$.

	If $\widehat{\mu}_t$ is continuous at $t$, then $(\mu_n, t_n) \to (\mu, t)$ implies $\sigma(x^{(\mu_n, t_n)}, x^{(\mu, t)}) \to 0$. Hence the mapping is Borel measurable (continuous) at $(\mu, t)$.

	If $\widehat{\mu}_t$ is discontinuous at $t$, then for all $t' > t$, $\sigma(x^{(\mu, t)}, x^{(\mu', t')}) \ge d_\Theta(\widehat{\mu}_t, \widehat{\mu}_{t+}) > 0$. Meanwhile, $(\mu_n, t_n) \to (\mu, t-)$ implies $\sigma(x^{(\mu_n, t_n)}, x^{(\mu, t)}) \to 0$. Then, for any sufficiently small open ball around $x^{(\mu, t)}$, the inverse image is an open ball around $(\mu, t)$ truncated by $t' \le t$; hence, a Borel set. Therefore, the mapping is Borel measurable at $(\mu, t)$.
\end{proof}

We now define the simple recommendation strategy associated with $f$.

\begin{defi}[Simple recommendation]\label{defi:simple_recommendation}
	For $f \in \Delta_{\mu_0}$, define the law $\mathcal{P}$ of the corresponding belief process $\rp{\mu^f_t}$ and stopping time $\tau^f$ as $\mathcal{P} \circ \Phi = f$, i.e., for all Borel sets $B_x \subset D_{\infty}$ and $B_t \subset T$:
	\begin{align*}
		\mathcal{P}\left( \mu^f_t \in B_x, \tau^f \in B_t \right) = \int_{\Phi^{-1}(B_x \times B_t)} f(\d \mu, \d t).
	\end{align*}
\end{defi}

In words, a simple recommendation tells the agent to continue and nothing else, which induces the deterministic continuation belief path $\widehat{\mu}_t$. The principal only releases information when the agent is meant to stop. The following lemma is the direct-communication reduction underlying \cref{ssec:relaxed}.

\begin{lem}[Direct communication]\label{lem:direct_communication}
	For every feasible and obedient strategy $(\rp{\mu_t}, \tau)$, let $f \sim (\mu_{\tau}, \tau)$. Then $f \in \Delta_{\mu_0}$ and $f$ satisfies \eqref{eqn:OCC}. Conversely, for every $f \in \Delta_{\mu_0}$ that satisfies \eqref{eqn:OCC}, the simple recommendation strategy $(\rp{\mu^f_t}, \tau^f)$ is feasible and obedient and satisfies $f \sim (\mu^f_{\tau^f}, \tau^f)$.
\end{lem}

\Cref{lem:direct_communication} states that the outcomes (belief-time distributions) of all feasible and obedient strategies are exactly the outcomes of all obedient simple recommendations. Therefore, \eqref{prob:1} is equivalent to \eqref{prob:relaxed}, and any solution to \eqref{prob:relaxed} is implemented in the original game by the corresponding simple recommendation.

\begin{proof}[Proof of \cref{lem:direct_communication}]
The first statement: Suppose \eqref{eqn:OCC} is violated, i.e.\ $\exists t$ such that
\begin{align*}
	\int_{y > t} U(\mu, y) f(\d \mu, \d y) < U\left( \frac{\int_{y > t} \mu f(\d \mu, \d y)}{\int_{y > t} f(\d \mu, \d y)}, t \right).
\end{align*}
Define $\tau' = \min\{\tau, t\}$.
\begin{align*}
	\E\left[ U(\mu_{\tau'}, \tau') \right] = & \mathrm{Prob}(\tau \le t) \E\left[ U(\mu_{\tau}, \tau) \mid \tau \le t \right] + \mathrm{Prob}(\tau > t) \E\left[ U(\mu_{t}, t) \mid \tau > t \right] \\
	\ge & \mathrm{Prob}(\tau \le t) \E\left[ U(\mu_{\tau}, \tau) \mid \tau \le t \right] + \mathrm{Prob}(\tau > t) U(\E\left[ \mu_{t} \mid \tau > t \right], t) \\
	& \text{(Jensen's inequality)} \\
	= & \mathrm{Prob}(\tau \le t) \E\left[ U(\mu_{\tau}, \tau) \mid \tau \le t \right] + \mathrm{Prob}(\tau > t) U(\E\left[ \mu_{\tau} \mid \tau > t \right], t) \\
	& \text{(Optional stopping theorem)} \\
	> & \mathrm{Prob}(\tau \le t) \E\left[ U(\mu_{\tau}, \tau) \mid \tau \le t \right] + \mathrm{Prob}(\tau > t) \E\left[ U(\mu_{\tau}, \tau) \mid \tau > t \right] \\
	= & \E[U(\mu_{\tau}, \tau)].
\end{align*}
Therefore \eqref{prob:1:oc} is violated. That $f \in \Delta_{\mu_0}$ is a direct implication of the optional stopping theorem.

The second statement: First, we show that $(\mu^f_{\tau^f}, \tau^f) \sim f$. For all Borel sets $B_{\mu} \subset \Delta(\Theta)$ and $B_t \subset T$,
\[
	(\mu^f_{\tau^f}, \tau^f) \in B_{\mu} \times B_t \iff t \in B_t \ \& \ x^{(\mu, t)}(t) \in B_{\mu} \iff (\mu, t) \in B_{\mu} \times B_t.
\]
Therefore, $\mathcal{P}\left( (\mu^f_{\tau^f}, \tau^f) \in B_{\mu} \times B_t \right) = \int_{B_{\mu} \times B_t} f(\d \mu, \d t)$.

Second, we show that $\rp{\mu^f_t}$ is a martingale. Take any Borel subset $B \subset \mathcal{F}_t$ such that $\mathcal{P}(B) > 0$. The set $B$ can be further divided into two events (which are elements of $\mathcal{F}_t$ since $\tau^f$ is a stopping time).
\begin{itemize}
	\item Case 1: $\tau^f \le t$. For all $t' > t$,
	\begin{align*}
		\E\left[ \mu^f_{t'} \mid B \right] = & \E\left[ \mu^f_{t'} \mid B, \tau^f \le t \right] \\
		= & \frac{1}{\mathcal{P}(B)} \int_{(x^{(\mu, s)}, s) \in B} x^{(\mu, s)}(t') \, \d \mathcal{P} \\
		= & \frac{1}{\mathcal{P}(B)} \int_{(x^{(\mu, s)}, s) \in B} x^{(\mu, s)}(t) \, \d \mathcal{P} \\
		= & \E\left[ \mu^f_t \mid B \right].
	\end{align*}
	The third line follows from the definition of $x^{(\mu, s)}$ and $t' > t \ge s$.
	\item Case 2: $\tau^f > t$. In this case, for all $\mu$ and $s > t$, $(x^{(\mu, s)}, s) \in B$ since all such $x^{(\mu, s)} = \widehat{\mu}_t$ up to period $t$. Therefore,
	\begin{align*}
		\E\left[ \mu^f_{t'} \mid B \right] = & \E\left[ \mu^f_{t'} \mid \tau^f > t \right] \\
		= & \frac{1}{\mathcal{P}(B)} \int_{s > t} x^{(\mu, s)}(t') \, \d \mathcal{P} \\
		= & \frac{1}{\mathcal{P}(B)} \left( \int_{s > t'} x^{(\mu, s)}(t') \, \d \mathcal{P} + \int_{t < s \le t'} x^{(\mu, s)}(t') \, \d \mathcal{P} \right) \\
		= & \frac{1}{\mathcal{P}(B)} \left( \int_{s > t'} \widehat{\mu}_{t'} \, \d \mathcal{P} + \int_{t < s \le t'} \mu f(\d \mu, \d s) \right) \\
		= & \frac{1}{\int_{s > t} f(\d \mu, \d s)} \left( \widehat{\mu} \cdot \int_{s > t'} f(\d \mu, \d s) + \int_{t < s \le t'} \mu f(\d \mu, \d s) \right) \\
		= & \widehat{\mu}_t \\
		= & \E[\mu^f_t \mid \tau^f > t].
	\end{align*}
\end{itemize}

Third, we verify \eqref{prob:1:oc}. By the definition of $\rp{\mu^f_t}$, conditional on the event $\tau^f \le t$, $\mu^f_t$ is constant in $t$. Let $\widetilde{\tau} := \min\{\tau^f, \tau'\}$ and $\delta \tau := (\tau' - \tau^f)^+$. Then $(\mu_{\widetilde{\tau}}, \tau')$ has the same distribution as $(\mu_{\tau'}, \tau')$.
\begin{align*}
	\E\left[ U(\mu_{\tau'}, \tau') \right] = & \mathcal{P}(\tau^f > \tau') \E\left[ U(\mu_{\tau'}, \tau') \mid \tau^f > \tau' \right] + \mathcal{P}(\tau^f \le \tau') \E\left[ U(\mu_{\tau^f}, \tau^f + \delta \tau) \mid \tau^f \le \tau' \right] \\
	\le & \mathcal{P}(\tau^f > \tau') \E\left[ U(\mu_{\tau'}, \tau') \mid \tau^f > \tau' \right] + \mathcal{P}(\tau^f \le \tau') \E\left[ U(\mu_{\tau^f}, \tau^f) \mid \tau^f \le \tau' \right] \\
	& \text{(Because $\delta \tau \ge 0$.)} \\
	= & \mathcal{P}(\tau^f > \tau') \E\left[ U(\widehat{\mu}_{\tau'}, \tau') \mid \tau^f > \tau' \right] + \mathcal{P}(\tau^f \le \tau') \E\left[ U(\mu_{\tau^f}, \tau^f) \mid \tau^f \le \tau' \right] \\
	\le & \mathcal{P}(\tau^f > \tau') \E\left[ \E_f\left[ U(\mu, t) \mid t > \tau' \right] \big| \tau^f > \tau' \right] + \mathcal{P}(\tau^f \le \tau') \E\left[ U(\mu_{\tau^f}, \tau^f) \mid \tau^f \le \tau' \right] \\
	& \text{(Implied by \eqref{eqn:OCC}.)} \\
	= & \mathcal{P}(\tau^f > \tau') \E\left[ U(\mu_{\tau^f}, \tau^f) \mid \tau^f > \tau' \right] + \mathcal{P}(\tau^f \le \tau') \E\left[ U(\mu_{\tau^f}, \tau^f) \mid \tau^f \le \tau' \right] \\
	= & \E\left[ U(\mu_{\tau^f}, \tau^f) \right]. \qedhere
\end{align*}
\end{proof}

\subsection{Proof of strong duality (\texorpdfstring{\cref{lem:duality}}{Lemma 1})}\label{appendix:duality_proof}

We restate \cref{lem:duality} for convenience.

\begin{lem}[Strong duality, restated]
	Under \cref{ass:positive,ass:cont}, strong duality holds:
	\begin{align*}
		\sup_{f} \min_{\Lambda} \mathcal{L}(f, \Lambda) = \min_{\Lambda} \sup_{f} \mathcal{L}(f, \Lambda).
	\end{align*}
	Furthermore, the minimum is achieved by some $\Lambda^* \in \mathcal{B}(T^{\circ})$ and the supremum is achieved by some $f^* \in \Delta_{\mu_0}$.
\end{lem}

\begin{proof}
The lemma is trivial when $T$ is finite. We prove the case when $T$ is a continuum. We first prove strong duality. Define the mapping $G$:
\begin{align*}
	G(f)(t) = \int_{\tau > t} U(\mu, \tau) f(\d \mu, \d \tau) - U\left( \int_{\tau > t} \mu f(\d \mu, \d \tau), t \right).
\end{align*}
Since $U$ is convex, $G(\cdot, t)$ is concave. Since $U$ is bounded, $G$ maps $\Delta_{\mu_0}$ into $L^{\infty}(T^{\circ})$.

Next, we verify that there exists some $f \in \Delta_{\mu_0}$ such that $G(f)(\cdot)$ is an interior point (with respect to the $L^{\infty}$ norm) of the positive cone. Let $\widehat{U}(\mu, t)$ denote $\E_{\mu}[U(\delta_\theta, t)]$. Since $T$ is compact and $U$ is continuous, \cref{ass:positive} implies that there exists $\epsilon > 0$ such that $\widehat{U}(\mu_0, t) - U(\mu_0, t) \ge 3\epsilon$ for all $t \in T$. Since $U$ is continuous, there exists a finite partition $\{t_i\}_{i=0}^{I} \subset T$ with $t_0 = 0$ such that $|U(\mu_0, t) - U(\mu_0, t_i)| \le \epsilon$ when $t \in T$ and $t_i$ are adjacent. Therefore, for all $t \in (t_i, t_{i+1}] \cap T$,
\begin{align*}
	U(\mu_0, t) \le \widehat{U}(\mu_0, t_{i+1}) - 2\epsilon.
\end{align*}
Let $M = \sup_t U(\mu_0, t)$. Without loss of generality, we can pick $\epsilon < \frac{1}{3} M$. Next, we define $f$:
\begin{align*}
	\text{for } i = 1, \dots, I - 1,
	& \begin{dcases}
		f(t_i) = \left( \frac{\epsilon}{M} \right)^{i-1} \left( 1 - \frac{\epsilon}{M} \right); \\
		f(\mu \mid t_i) \text{ achieves } \widehat{U}(\mu_0, t_i).
	\end{dcases} \\
	& \begin{dcases}
		f(t_I) = \left( \frac{\epsilon}{M} \right)^{I}; \\
		f(\mu \mid t_I) \text{ achieves } \widehat{U}(\mu_0, t_I).
	\end{dcases}
\end{align*}
By definition of $\widehat{U}$, $\forall i$, $f(\mu \mid t_i)$ is well-defined since the upper concave hull can be achieved by a finite distribution per Caratheodory's theorem. Then, for any $i = 0, \dots, |I| - 1$ and $t \in (t_i, t_{i+1}]$,
\begin{align*}
	G(f)(t) = & \sum_{j = i+1}^{I-1} \left( \left( \frac{\epsilon}{M} \right)^{j-1} \left( 1 - \frac{\epsilon}{M} \right) \right) \cdot \widehat{U}(\mu_0, t_j) \\
	& + \left( \frac{\epsilon}{M} \right)^{I} \widehat{U}(\mu_0, t_I) - U(\mu_0, t) \cdot \left( \sum_{j = i+1}^{I-1} \left( \left( \frac{\epsilon}{M} \right)^{j-1} \left( 1 - \frac{\epsilon}{M} \right) \right) + \left( \frac{\epsilon}{M} \right)^{I} \right) \\
	\ge & \left( \frac{\epsilon}{M} \right)^{i} \left( 1 - \frac{\epsilon}{M} \right) \widehat{U}(\mu_0, t_{i+1}) - \left( \frac{\epsilon}{M} \right)^{i} \cdot U(\mu_0, t) \\
	\ge & \left( \frac{\epsilon}{M} \right)^{i} \left( 1 - \frac{\epsilon}{M} \right) (U(\mu_0, t_{i+1}) + 2\epsilon) - \left( \frac{\epsilon}{M} \right)^{i} \cdot U(\mu_0, t) \\
	= & \left( \frac{\epsilon}{M} \right)^{i} \left( 2\epsilon - \frac{2\epsilon^2}{M} - \frac{\epsilon}{M} U(\mu_0, t) \right) \\
	\ge & \left( \frac{\epsilon}{M} \right)^{I} \cdot \frac{\epsilon}{3}.
\end{align*}
Therefore, $G(f)(\cdot)$ is bounded away from $0$ by at least $\left( \frac{\epsilon}{M} \right)^{I} \cdot \frac{\epsilon}{3}$ under the $L^{\infty}$ norm.

The following lemma establishes that for any target joint distribution $f$, we can find a ``nearby'' continuous distribution which does not alter the value of $G$ too much. Let $\Delta^{C}_{\mu_0}$ be the subset of $\Delta_{\mu_0}$ such that $G(f)(t)$ is uniformly continuous on $T^{\circ}$.

\begin{lem}\label{lem:cont}
	Given \cref{ass:cont}, for all $f \in \Delta_{\mu_0}$ and all $\epsilon$, there exists $f' \in \Delta^{C}_{\mu_0}$ such that $\E_{f'}[V] \ge \E_{f}[V] - \epsilon$ and $G(f')(t) \ge G(f)(t) - \epsilon$.
\end{lem}

We now take stock of the conditions to apply Theorem 1, Chapter 8.6 of \cite{luenberger1997optimization}. The objective $\int V(\mu, \tau) f(\d \mu, \d \tau)$ is a linear functional. $G$ is a concave mapping, and by \cref{lem:cont} there exists $f \in \Delta^{C}_{\mu_0}$ such that $G(f)(\cdot)$ is in the interior of the positive cone. Further, since $f \in \Delta^{C}_{\mu_0}$, $G(f)(\cdot)$ is uniformly continuous on $T^{\circ}$, hence $\mathcal{B}(T^{\circ})$ is the appropriate dual space. We then have
\begin{align*}
	\sup_{f \in \Delta^{C}_{\mu_0}} \inf_{\Lambda \in \mathcal{B}(T^{\circ})} \mathcal{L}(f, \Lambda) = \min_{\Lambda \in \mathcal{B}(T^{\circ})} \sup_{f \in \Delta^{C}_{\mu_0}} \mathcal{L}(f, \Lambda),
\end{align*}
where the minimum is achieved by $\Lambda \in \mathcal{B}(T^{\circ})$.

Then note that for all $f \in \Delta_{\mu_0}$, $\Lambda \in \mathcal{B}(T^{\circ})$ and $\epsilon > 0$, \cref{lem:cont} implies that there exists $f' \in \Delta^{C}_{\mu_0}$ such that $\mathcal{L}(f', \Lambda) \ge \mathcal{L}(f, \Lambda) - \epsilon$. Therefore,
\begin{align*}
	\sup_{f \in \Delta^{C}_{\mu_0}} \inf_{\Lambda \in \mathcal{B}(T^{\circ})} \mathcal{L}(f, \Lambda) \le \sup_{f \in \Delta_{\mu_0}} \inf_{\Lambda \in \mathcal{B}(T^{\circ})} \mathcal{L}(f, \Lambda) \le \min_{\Lambda \in \mathcal{B}(T^{\circ})} \sup_{f \in \Delta_{\mu_0}} \mathcal{L}(f, \Lambda) = \min_{\Lambda \in \mathcal{B}(T^{\circ})} \sup_{f \in \Delta^{C}_{\mu_0}} \mathcal{L}(f, \Lambda).
\end{align*}
The inner inequality must be an equality since the smallest term equals the largest term.

Finally, we verify that there exists $f$ achieving the maximum when $V$ is upper semicontinuous. This implies that $\E_f[V]$ is upper semicontinuous. Then, it is sufficient to verify that the set of $f$ such that $G(f) \ge 0$ is closed. Suppose not; then there exist $f_n \to f$ such that $G(f_n) \ge 0$ but $G(f)(t) < 0$. Consider
\begin{align*}
	\gamma_{\epsilon}(\tau) =
	\begin{cases}
		0 & \tau \le t \\
		\frac{\tau - t}{\epsilon} & \tau \in (t, t + \epsilon) \\
		1 & \tau > t + \epsilon.
	\end{cases}
\end{align*}
Then $\int U \cdot \gamma_{\epsilon} \, \d f \to \int_{\tau > t} U \, \d f$ and $\int \mu \cdot \gamma_{\epsilon} \, \d f \to \int_{\tau > t} \mu \, \d f$ as $\epsilon \to 0$. For $\epsilon$ sufficiently small, $\int U \cdot \gamma_{\epsilon} \, \d f - U\left( \int \mu \cdot \gamma_{\epsilon} \, \d f, t \right) < 0$. However, since $U \cdot \gamma_{\epsilon}$ and $\mu \cdot \gamma_{\epsilon}$ are bounded and continuous functions,
\begin{align*}
	& \int U \cdot \gamma_{\epsilon} \, \d f - U\left( \int \mu \cdot \gamma_{\epsilon} \, \d f, t \right) \\
	= & \lim_{n \to \infty} \int U \cdot \gamma_{\epsilon} \, \d f_n - U\left( \int \mu \cdot \gamma_{\epsilon} \, \d f_n, t \right).
\end{align*}
Then there exists $n$ such that $\int U \cdot \gamma_{\epsilon} \, \d f_n - U\left( \int \mu \cdot \gamma_{\epsilon} \, \d f_n, t \right) < 0$. Note that $f_n \cdot \gamma_{\epsilon}$ is a convex combination of $\bm{1}_{t > s} f_n$ for $s \in [t, t + \epsilon]$ and $G(\bm{1}_{t > s} f_n)(T) = G(f_n)(s) \ge 0$. Then $\int U \cdot \gamma_{\epsilon} \, \d f_n - U\left( \int \mu \cdot \gamma_{\epsilon} \, \d f_n, t \right) = G(f_n \cdot \gamma_{\epsilon}, T) \ge 0$ since $G$ is concave. Contradiction.
\end{proof}

\begin{proof}[Proof of \cref{lem:cont}]
For every $f \in \Delta_{\mu_0}$, $G(f)(t)$ has bounded variation and only jumps down. Therefore, $G(f)(t)$ can be decomposed into $g(t) + h(t)$, where $g$ is bounded and continuous and $h$ is bounded and decreasing. Define the ``delayed'' measure
\begin{align*}
	f^{s}(\mu, t) :=
	\begin{cases}
		0 & t < s \\
		f(\mu, t - s) & t \in [s, \sup(T)) \\
		f(\mu, [t - s, \sup(T)]) & t = \sup(T).
	\end{cases}
\end{align*}
In words, $f^s$ delays the distribution of $f$ by $s$. Pick $\delta > 0$ as a continuity parameter corresponding to $\frac{1}{2} \epsilon$ for $U$, $V$ and $g$ (all the continuity is uniform since $T$ and $D$ are compact). Then,
\begin{align*}
	G(f^s)(t) = & \int_{\tau > t - s} U \, \d f - U\left( \int_{\tau > t - s} \mu \, \d f, t \right) \\
	\ge & \int_{\tau > t - s} U \, \d f - U\left( \int_{\tau > t - s} \mu \, \d f, t - s \right) - \frac{1}{2} \epsilon \\
	= & g(t - s) + h(t - s) - \frac{1}{2} \epsilon \\
	\ge & g(t) + h(t) - \frac{1}{2} \epsilon \\
	= & G(f)(t) - \frac{1}{2} \epsilon; \\
	\E_{f^s}[V] = & \int_0^{\sup(T) - s} V(\mu, t + s) \, \d f + \int_{\sup(T) - s}^{\sup(T)} V(\mu, \sup(T)) \, \d f \\
	\ge & \E_f[V] - \frac{1}{2} \epsilon.
\end{align*}
Let $\widehat{f}$ be the uniform randomization of $f^s$ for $s \in [0, \delta]$. Then, since $\E_f[V]$ is a linear operator in $f$, $\E_{\widehat{f}}[V] \ge \E_f[V] - \frac{1}{2} \epsilon$. Since $G$ is a concave operator in $f$, $G(\widehat{f})(t) \ge G(f)(t) - \epsilon$.

Next, we prove the uniform continuity of $G(\widehat{f})$. Note that for all $t < t' < t + \delta$,
\begin{align*}
	\widehat{f}((t, t']) = & \frac{1}{\delta} \int_{s = 0}^{\delta} (F(t' - s) - F(t - s)) \, \d s \\
	\le & \frac{1}{\delta} \int_{(t, t'] \cup (t - \delta, t' - \delta]} F(s) \, \d s \\
	\le & \frac{2 |t - t'|}{\delta}.
\end{align*}
Therefore,
\begin{align*}
	\left| \int_{\tau > t} U \, \d \widehat{f} - \int_{\tau > t'} U \, \d \widehat{f} \right| = \left| \int_{(t, t']} U \, \d \widehat{f} \right| \le |U| \widehat{f}((t, t']) \le |U| \cdot \frac{2 |t - t'|}{\delta}.
\end{align*}
Then, for all $\epsilon$, let $|U| \cdot \frac{2 |t - t'|}{\delta} < \frac{\epsilon}{2}$ and $2 |t - t'| / \delta$ be less than the continuity parameter of $U$ for $\frac{\epsilon}{2}$. Then $\left| U\left( \int_{\tau > t} \mu \, \d \widehat{f}, t \right) - U\left( \int_{\tau > t'} \mu \, \d \widehat{f}, t \right) \right| < \frac{\epsilon}{2}$. To sum up, $|G(\widehat{f})(t) - G(\widehat{f})(t')| < \epsilon$.
\end{proof}

\subsection{Continuum-state direct-communication reduction}\label{appendix:reduction_continuum}

In \cref{sec:linear} we extend the framework to a continuous state $\Theta = [0, 1]$ with mean-measurable indirect utilities. The direct-communication reduction of \cref{lem:direct_communication} generalizes to this setting: the continuum analog \eqref{prob:cont:relaxed} of the relaxed problem \eqref{prob:relaxed} is equivalent to the original dynamic problem \eqref{prob:1}. We first restate the result and then prove it.

\begin{lem}[Continuum direct communication]\label{lem:direct_communication_continuum}
	Under the continuum-state framework of \cref{ssec:linear_framework}, \eqref{prob:1} and \eqref{prob:cont:relaxed} have the same optimal value, and any solution to \eqref{prob:cont:relaxed} is implemented in the original game by a simple recommendation strategy in the sense of \cref{defi:simple_recommendation}.
\end{lem}

\begin{proof}
We can still define a simple recommendation as we did in \cref{defi:simple_recommendation}, and the proof of \cref{lem:direct_communication} still goes through under the continuum-state setting. Therefore, \eqref{prob:1} is equivalent to \eqref{prob:relaxed} (reformulated over $\Delta(\Theta) \times T$ with $\Theta = [0, 1]$). It remains to show that \eqref{prob:cont:relaxed} is equivalent to \eqref{prob:relaxed}, which then implies the equivalence of \eqref{prob:1} and \eqref{prob:cont:relaxed}.

We begin with a useful lemma in the same spirit as the fact that a distribution of posterior means is induced by some information structure if and only if it is a mean-preserving contraction of the prior \citep{kolotilin2018optimal,gentzkow2016rothschild,dworczak2019simple}.

\begin{lem}\label{lem:MPC_equivalence}
	Fix a distribution $f$ over $\Delta(\Theta) \times T$. Define the corresponding distribution $g$ over $\MM \times T$ such that, for every pair of Borel sets $B_M \subset \MM$ and $B_T \subset T$,
	\begin{align}\label{eqn:MPC_equivalence}
		\mathbb{P}_{(m, t) \sim g}(m \in B_M, t \in B_T) = \mathbb{P}_{(\mu, t) \sim f}( \E_\mu[\theta] \in B_M, t \in B_T).
	\end{align}
	Then $g \in \overline{\Delta}_{\mu_0}$, and $g$ satisfies the constraint \eqref{eqn:MPS}. Conversely, for every distribution $g \in \overline{\Delta}_{\mu_0}$ that satisfies the constraint \eqref{eqn:MPS}, there exists a distribution $f$ over $\Delta(\Theta) \times T$ such that \cref{eqn:MPC_equivalence} holds.
\end{lem}

\begin{proof}[Proof of \cref{lem:MPC_equivalence}]
The forward direction is clear since the marginal distribution $g_M$ is a mean-preserving contraction of $\mu_0$. For the converse, fix a distribution $g \in \overline{\Delta}_{\mu_0}$ that satisfies the constraint \eqref{eqn:MPS}. Thus $g_M$ is a mean-preserving contraction of $\mu_0$, so there exists a random variable $\tilde{\mu} \in \Delta(\Theta)$ such that $\mathbb{E}[\tilde{\mu}] = \mu_0$ and the distribution of $\mathbb{E}_{\tilde{\mu}}[\theta]$ is $g_M$. We then construct a random variable $\tau \in T$ so that $g$ is the joint distribution $(\mathbb{E}_{\tilde{\mu}}[\theta], \tau)$. Finally, we define the law of $f$ as the joint distribution of $(\tilde{\mu}, \tau)$. We can see that, for every pair of Borel sets $B_M \subset \MM$ and $B_T \subset T$,
\begin{align*}
	\mathbb{P}_{(\mu, t) \sim f}( \E_\mu[\theta] \in B_M, t \in B_T) = \mathbb{P}(\mathbb{E}_{\tilde{\mu}}[\theta] \in B_M, \tau \in B_T) = g(m \in B_M, t \in B_T),
\end{align*}
so \cref{eqn:MPC_equivalence} holds, as required.
\end{proof}

Given \cref{eqn:MPC_equivalence}, we can rewrite the objective function and the obedience constraint of \eqref{prob:relaxed} as follows:
\begin{align*}
	\int V(\mu, t) f(\d \mu, \d t) & = \int V(\E_{\mu}[\theta], t) f(\d \mu, \d t) \tag{$V$ depends on $\mu$ via $\E_\mu[\theta]$} \\
	& = \int V(m, t) g(\d m, \d t), \tag{\cref{eqn:MPC_equivalence}}
\end{align*}
and
\begin{align*}
	& \mathbb{E}\big[ U(\mu_{\tau}, \tau) \mid \tau > t \big] \ge U\big( \mathbb{E}[\mu_\tau \mid \tau > t], t \big) \\
	\iff & \frac{\E\big[ U(\E_{\mu_\tau}[\theta], \tau) \mathbf{1}\{\tau > t\} \big]}{\mathbb{P}(\tau > t)} \ge U\left( \E[\E_{\mu_\tau}[\theta] \mid \tau > t], t \right) \tag{$U$ depends on $\mu$ via $\E_\mu[\theta]$} \\
	\iff & \int_{y > t} U(m, y) g(\d m, \d y) \ge \int_{y > t} g(\d m, \d y) \cdot U\bigg( \frac{\int_{y > t} m \, g(\d m, \d y)}{\int_{y > t} g(\d m, \d y)}, t \bigg), \tag{\cref{eqn:MPC_equivalence}}
\end{align*}
which are exactly the objective function and the \eqref{eqn:OCC-C} constraint in \eqref{prob:cont:relaxed}, respectively. Thus \eqref{prob:cont:relaxed} is equivalent to \eqref{prob:relaxed}, as desired.
\end{proof}


\section{Proof of the First-Order Characterization and Numerical Algorithm}\label{appendix:foc_section}

This appendix carries the proof of \cref{thm:foc} and the numerical algorithm that operationalizes the recipe when closed-form solutions are unavailable. Both pieces build on the strong-duality result of \cref{lem:duality}, which closes the duality gap and guarantees existence of an optimal multiplier $\Lambda$.

\subsection{Proof of \cref{thm:foc}}\label{appendix:foc_proof}

\Cref{thm:foc} characterizes the optimal joint distribution $f$ via a sufficient first-order condition and, by appealing to the strong-duality result of \cref{lem:duality}, a near-necessary converse: for every solution of \eqref{prob:relaxed}, there exist a multiplier $\Lambda$, a vector $a \in \R^{|\Theta|}$, and a subgradient selection certifying \cref{eqn:foc}.

\begin{proof}[Proof of \cref{thm:foc}]
	\emph{Sufficiency.} Suppose for the purpose of contradiction that $(f, a, \Lambda)$ satisfies \cref{eqn:foc}, \eqref{eqn:OCC}, and the complementary slackness condition $\mathcal{L}(f, \Lambda) = \E_f[V]$ but $f$ is suboptimal in \eqref{prob:relaxed}. Then there exists $f' \in \Delta_{\mu_0}$ such that $\mathcal{L}(f, \Lambda) < \mathcal{L}(f', \Lambda)$. Since $\mathcal{L}$ is concave, for all $\alpha \in (0, 1)$,
	\begin{align*}
		&\frac{\mathcal{L}(\alpha f' + (1 - \alpha) f, \Lambda) - \mathcal{L}(f, \Lambda)}{\alpha} \ge \mathcal{L}(f', \Lambda) - \mathcal{L}(f, \Lambda) > 0 \\
		\implies & \varliminf_{\alpha \to 0} \frac{\mathcal{L}(\alpha f' + (1 - \alpha) f, \Lambda) - \mathcal{L}(f, \Lambda)}{\alpha} \ge \mathcal{L}(f', \Lambda) > 0 \\
		\iff & \int V(\mu, t) (f' - f)(\d \mu, \d t) + \int_{t \in T^{\circ}} \int_{\tau > t} U(\mu, \tau) (f' - f)(\d \mu, \d \tau) \d \Lambda(t) \\
			& - \varlimsup_{\alpha \to 0} \int_{t \in T^{\circ}} \underbrace{\frac{U\left( \int_{\tau > t} \mu (\alpha f' + (1 - \alpha) f)(\d \mu, \d \tau), t \right) - U\left( \int_{\tau > t} \mu f(\d \mu, \d \tau), t \right)}{\alpha}}_{\ge \max_{\gamma \in \nabla_\mu U(\widehat{\mu}_t)} \gamma \cdot \int_{\tau > t} \mu (f' - f)(\d \mu, \d \tau)} \d \Lambda(t) > 0 \\
		\implies & \int V(\mu, t) (f' - f)(\d \mu, \d t) + \int_{t \in T^{\circ}} U(\mu, t) \Lambda(t) (f' - f)(\d \mu, \d t) \\
			& - \int_{t \in T^{\circ}} \nabla_{\mu} U\left( \int_{\tau > t} \mu f(\d \mu, \d \tau), t \right) \cdot \int_{\tau > t} \mu (f' - f)(\d \mu, \d \tau) \d \Lambda(t) > 0 \\
		\implies & \int l_{f, \Lambda} (f' - f)(\d \mu, \d t) > 0 \\
		\implies & \E_{f'}[l_{f, \Lambda}] > a \cdot \mu_0.
	\end{align*}
	The selection of subgradients in the fourth inequality can be arbitrary. The last inequality violates \cref{eqn:foc}.

	\emph{Necessity.} Suppose $f$ solves \eqref{prob:relaxed}. By \cref{lem:duality}, let $\Lambda$ be the minimizer in \eqref{eqn:dual}. When defining $l_{f, \Lambda}(\mu, t)$, make the following selection: for all $(\mu, t) \in D$, if $t < \overline{t}$, $\nabla_{\mu} U(\widehat{\mu}_t, t) = \arg\max_{\gamma \in \nabla_{\mu} U(\widehat{\mu}_t, t)} \gamma \cdot \mu$; if $t \ge \overline{t}$, $\nabla_{\mu} U(0, t)$ is chosen so that $\nabla_{\mu} U(0, t) \cdot \mu = U(\mu, t)$. Define
	\begin{align*}
		\widehat{l}(\mu) := \sup_{f' \in \Delta_{\mu}} \E_{f'}[l_{f, \Lambda}(\nu, t)].
	\end{align*}
	Let $a \cdot \mu$ be the supporting hyperplane of $\widehat{l}$ at $\mu_0$. Then $l_{f, \Lambda}(\mu, t) \le a \cdot \mu$. Next, we prove that $\E_f[l_{f, \Lambda}(\mu, t)] = a \cdot \mu_0$. Suppose for the purpose of contradiction that $\E_f[l_{f, \Lambda}(\mu, t)] \le a \cdot \mu_0 - \epsilon$ for some $\epsilon > 0$. There exists a finitely supported $f' \in \Delta_{\mu_0}$ such that $\E_{f'}[l_{f, \Lambda}] > a \cdot \mu_0 - \tfrac{1}{2} \epsilon$. Then,
	\begin{align*}
		& \varliminf_{\alpha \to 0} \frac{\mathcal{L}(\alpha f' + (1 - \alpha) f, \Lambda) - \mathcal{L}(f, \Lambda)}{\alpha} \\
		= & \int V(\mu, t) (f' - f)(\d \mu, \d t) + \int_{t \in T^{\circ}} \int_{\tau > t} U(\mu, t) (f' - f)(\d \mu, \d \tau) \d \Lambda(t) \\
			& - \varlimsup_{\alpha \to 0} \int_{t \in T^{\circ}} \frac{U\left( \int_{\tau > t} \mu (\alpha f' + (1 - \alpha) f)(\d \mu, \d \tau), t \right) - U\left( \int_{\tau > t} \mu f(\d \mu, \d \tau), t \right)}{\alpha} \d \Lambda(t) \\
		\ge & \int V(\mu, t) (f' - f)(\d \mu, \d t) + \int_{t \in T^{\circ}} \int_{\tau > t} U(\mu, t) (f' - f)(\d \mu, \d \tau) \d \Lambda(t) \\
			& - \sup_{\gamma_t \in \nabla_{\mu} U(\widehat{\mu}_t, t)} \int_{t < \overline{t}} \gamma_t \cdot \int_{\tau > t} \mu (f' - f)(\d \mu, \d \tau) \d \Lambda(t) \\
			& - \int_{t \ge \overline{t}} U\left( \int_{\tau > t} \mu f'(\d \mu, \d \tau), t \right) \d \Lambda(t) \\
		\ge & \int V(\mu, t) (f' - f)(\d \mu, \d t) + \int_{t \in T^{\circ}} \int_{\tau > t} U(\mu, t) (f' - f)(\d \mu, \d \tau) \d \Lambda(t) \\
			& - \int_{t < \overline{t}} \int_{\tau > t} \sup_{\gamma \in \nabla_{\mu} U(\widehat{\mu}_t, t)} \gamma \cdot \mu (f' - f)(\d \mu, \d \tau) \d \Lambda(t) \\
			& - \int_{t \ge \overline{t}} U\left( \int_{\tau > t} \mu f'(\d \mu, \d \tau), t \right) \d \Lambda(t) \\
		\ge & \int V(\mu, t) (f' - f)(\d \mu, \d t) + \int_{t \in T^{\circ}} \int_{\tau > t} U(\mu, t) (f' - f)(\d \mu, \d \tau) \d \Lambda(t) \\
			& - \int_{t < \overline{t}} \int_{\tau > t} \sup_{\gamma \in \nabla_{\mu} U(\widehat{\mu}_t, t)} \gamma \cdot \mu (f' - f)(\d \mu, \d \tau) \d \Lambda(t) \\
			& - \int_{t \ge \overline{t}} \int_{\tau > t} U(\mu, t) f'(\d \mu, \d t) \d \Lambda(t) \\
		= & \int V(\mu, t) (f' - f)(\d \mu, \d t) + \int_{t \in T^{\circ}} \int_{\tau > t} U(\mu, t) (f' - f)(\d \mu, \d \tau) \d \Lambda(t) \\
			& - \int_{t \in T^\circ} \left( \int_{\tau < t} \nabla_{\mu} U(\widehat{\mu}_{\tau}, \tau) \cdot \mu \d \Lambda(\tau) \right) (f' - f)(\d \mu, \d t) \\
		= & \E_{f' - f}[l_{f, \Lambda}(\mu, t)] > \tfrac{1}{2} \epsilon.
	\end{align*}
	This contradicts the optimality of $f$.
\end{proof}

\subsection{An efficient algorithm for computing the relaxed problem \eqref{prob:relaxed}}\label{appendix:numerical}

When closed-form solutions to \eqref{prob:relaxed} are unavailable, the recipe articulated in the body admits an efficient numerical implementation. A finite discretization of $T \times \Delta(\Theta)$ reduces the dual problem to a $|\Theta \times T|$-dimensional gradient descent: the inner problem (a finite-dimensional concavification step) is solved for fixed multipliers, and the outer problem (an affine function of the multipliers) is solved by gradient descent. The construction is grounded in \cref{thm:foc} and provides a saddle-point reformulation of the dual that is convex in the multiplier variables.

Define the function $\Phi$ by
\begin{align*}
	\Phi(\Lambda, b, \nu, \mu, \tau) := V(\mu, \tau) + \Lambda(\tau) U(\mu, \tau) - \Lambda(0) U(\mu_0, 0) - b_{\tau} \mu + b_0 \mu_0 + \int_{t \in (0, \tau)} \left( \beta_t \nu_t - U(\nu_t, t) \right) \d \Lambda(t),
\end{align*}
where $\Lambda \in \mathbb{L}$, $\nu \in \mathbb{D}(T)$ (the set of c\`adl\`ag paths on $T$), $(\mu, \tau) \in D$, and $b_t = b_0 + \int_{s < t} \beta_s \d \Lambda(s)$ for $\beta \in L^1(T)$.

\begin{thm}\label{thm:numerical}
	Suppose $(\Lambda^*, f, a^*)$ together with a selection of $\nabla U$ satisfies \eqref{eqn:foc}. Let $\beta_t^* = \nabla U(\widehat{\mu}_t, t)$ and $b_0^* = a^*$. Then
	\begin{align*}
		(\Lambda^*, b^*) \in \arg\min_{\Lambda, b} \max_{\nu, \mu, \tau} \Phi(\Lambda, b, \nu, \mu, \tau).
	\end{align*}
	Moreover, for all $(\mu^*, \tau^*) \in \mathrm{supp}(f)$,
	\begin{align*}
		(\widehat{\mu}, \mu^*, \tau^*) \in \arg\max_{\nu, \mu, \tau} \Phi(\Lambda^*, b^*, \nu, \mu, \tau).
	\end{align*}
\end{thm}

In the special case where $T \subset \mathbb{N}$ is a discrete and finite set, \cref{thm:numerical} yields a simple algorithm for computing the shadow prices, given by
\begin{align*}
	(\Lambda^*, b^*) \in \arg\min_{\Lambda, b} \max_{\nu, \mu, \tau} \bigg( & V(\mu, \tau) + \Lambda(\tau) U(\mu, \tau) - \Lambda(0) U(\mu_0, 0) - b_{\tau} \mu + b_0 \mu_0 \\
		& + \sum_{t = 1}^{\tau - 1} \left( (b_{t + 1} - b_t) \cdot \nu_t - (\Lambda(t + 1) - \Lambda(t)) U(\nu_t, t) \right) \bigg).
\end{align*}

\begin{itemize}
	\item \emph{The inner problem.} For each given $\tau$, the inner problem can be solved by maximizing $\mu$ and $\nu_t$ period by period; then $\tau$ is solved by enumerating finitely many values in $T$.
	\item \emph{The outer problem.} Fixing any $(\nu, \mu, \tau)$, the objective is an affine function of $(\Lambda, b)$. Therefore, the objective is convex in $(\Lambda, b)$ and can be solved efficiently via gradient descent.
\end{itemize}

\begin{proof}[Proof of \cref{thm:numerical}]
	For each $f \in \Delta_{\mu_0}$, consider two auxiliary objects.
	\begin{itemize}
		\item $\phi^f$ is a probability measure in $\Delta(\mathbb{D}(T) \times D)$ defined as follows. The marginal measures are $\phi^f|_{\mathbb{D}(T)} = \delta_{\widehat{\mu}}$ and $\phi^f|_{D} = f$.
		\item $(\rp{\mu_t^f}, \tau^f)$ is the simple recommendation defined by $f$.
	\end{itemize}
	We first prove the claim that $(\Lambda^*, b^*)$ and $\phi^f$ constitute a saddle point of
	\begin{align*}
		\min_{\Lambda, b} \max_{\phi \in \Delta(\mathbb{D}(T) \times D)} \E_{\phi}[\Phi(\Lambda, b, \nu, \mu, \tau)].
	\end{align*}
	\begin{itemize}
		\item Fixing $(\Lambda^*, b^*)$, we show the optimality of $\phi^f$. It is equivalent to show that for all $(\mu, \tau) \in \mathrm{supp}(f)$,
		\begin{align*}
			(\widehat{\mu}, \mu, \tau) \in \arg\max_{\nu, \mu, \tau} \Phi(\Lambda^*, b^*, \nu, \mu, \tau).
		\end{align*}
		Observe that
		\begin{align*}
			\Phi(\Lambda^*, b^*, \nu, \mu, \tau) = & \underbrace{V(\mu, \tau) + \Lambda(\tau) U(\mu, \tau) - b_{\tau} \mu}_{\text{maximized by } \mathrm{supp}(f)} - \Lambda(0) U(\mu_0, 0) + b_0 \mu_0 \\
				& + \int_{t \in (0, \tau)} \underbrace{(\beta_t \nu_t - U(\nu_t, t))}_{\text{maximized if } \nabla U(\nu_t, t) = \beta_t} \d \Lambda(t).
		\end{align*}
		The first maximizer is an implication of \cref{thm:foc}. The second maximizer is achieved by $\widehat{\mu}$.
		\item Fixing $\phi^f$, we show the optimality of $(\Lambda^*, b^*)$.
		\begin{align*}
			& \E_{\phi^f}[\Phi(\Lambda, b, \nu, \mu, \tau)] = \E_f[\Phi(\Lambda, b, \widehat{\mu}, \mu, \tau)] \\
			= & \int \left[ V(\mu, \tau) + \Lambda(\tau) U(\mu, \tau) - \Lambda(0) U(\mu_0, 0) - b_{\tau} \mu + b_0 \mu_0 + \int_{t \in (0, \tau)} (\beta_t \widehat{\mu}_t - U(\widehat{\mu}_t, t)) \d \Lambda(t) \right] f(\d \mu, \d \tau) \\
			= & \E_f[V(\mu, \tau)] - \underbrace{\int \left( \int_{t \in (0, \tau)} b_t \d \widehat{\mu}_t + b_{\tau} (\mu - \widehat{\mu}_t) \right) f(\d \mu, \d \tau)}_{= \E_{\mathcal{P}}[\int_0^{\tau^f} b_t \d \mu_t^f] = 0} \\
				& + \int \left( \int_{t \in (0, \tau)} (U(\mu, 0) - U(\widehat{\mu}_t, t)) \Lambda(\d t) \right) f(\d \mu, \d \tau) + \Lambda(0) \E_f[U(\mu, \tau) - U(\mu_0, 0)] \\
			= & \E_f[V(\mu, \tau)] + \underbrace{\int G(f)(t) \d \Lambda(t)}_{\text{minimized by } \Lambda^*}.
		\end{align*}
		The second equality is by integration by parts. The first underbrace is from $\rp{\mu^f_t}$ being a martingale. The last equality is from switching the order of integration. The second underbrace is from the complementary slackness condition.
	\end{itemize}
	Therefore, we have proved that
	\begin{align*}
		(\Lambda^*, b^*) \in & \arg\min_{\Lambda, b} \max_{\phi} \E_{\phi}[\Phi(\Lambda, b, \widehat{\mu}, \mu, \tau)] \\
		= & \arg\min_{\Lambda, b} \max_{\nu, \mu, \tau} \Phi(\Lambda, b, \widehat{\mu}, \mu, \tau).
	\end{align*}
\end{proof}


\section{Proofs for the Binary Application and Time-Risk Preferences}\label{appendix:binary_section}

This appendix collects the proofs of the results stated in \cref{sec:binary,sec:time_risk}, together with the formal version of the generalized binary-completeness statement (\cref{prop:binary_completeness_informal}). All four subsections apply the first-order characterization of \cref{thm:foc}: each proof identifies a candidate joint distribution, constructs an explicit Lagrange multiplier $\Lambda$, and verifies the concavification inequality \eqref{eqn:foc} on and off the support.

\subsection{Proof of \texorpdfstring{\cref{thm:binary_suspense_L}}{Theorem 2}}\label{appendix:binary_proofs}

The argument has two parts. We first establish that conditions (a)--(c) of \cref{thm:binary_suspense_L} are necessary for the optimality of the Suspense${}^{t_1}$-$\ell^{t_2}$ strategy; we then construct a multiplier $\Lambda$ that, under the additional condition (d), certifies the concavification inequality \eqref{eqn:foc} of \cref{thm:foc} and hence establishes sufficiency.

\begin{proof}[Proof of necessity]
    Under the $\text{Suspense}^{t_1}\text{-}\ell^{t_2}$ strategy, the corresponding joint distribution of stopping time and action is
    \begin{align*}
        \mathbb{P}(\tau \leq t, a = \ell) = 1-(2\mu_0 - t_1)e^{t_1-t}, \quad \quad \mathbb{P}(\tau = t_2, a = r) =(2\mu_0 - t_1)e^{t_1-t_2}
    \end{align*}
    for every $t \in [t_1,t_2].$ The principal's payoff under the $\text{Suspense}^{t_1}\text{-}\ell^{t_2}$ strategy is
    \begin{align*}
        V^*_\ell(t_1,t_2) \coloneqq (1-(2\mu_0-t_1))(v_\ell + h_\ell(t_1)) + (2\mu_0-t_1) \bigg(\int_{t_1}^{t_2} e^{-t+t_1} (v_\ell + h_\ell(t)) \d t + e^{-t_2+t_1}(v_r+h_r(t_2)) \bigg).
    \end{align*}
    Optimality implies the first order condition w.r.t. $t_2$:
    \begin{align*}
        \frac{\d V^*_\ell(t_1,t_2)}{\d t_2} = (2\mu_0 - t_1)e^{-t_2+t_1}(\Delta v + \Delta h(t_2) - h'_r(t_2)) = 0,
    \end{align*}
    which implies condition (a): $\Delta v + \Delta h(t_2) = h'_r(t_2)$ and the local SOC: $h''_r(t_2) \leq \Delta h'(t_2)$ in condition (c). Let $\Delta = t_2 - t_1$, and consider the FOC with respect to shifting $t_1$ and $t_2$ jointly:
    \begin{align*}
        \frac{\d V_\ell^{*}(t_1,t_1+\Delta)}{\d t_1} = \underbrace{(2\mu_0 - t_1-1)}_{< 0} \Big( \Psi(t_1,t_2) - h_\ell'(t_1) \Big) + e^{-t_2+t_1} \underbrace{\big(\Delta v + \Delta h(t_2) - h'_r(t_2)\big)}_{=0 \text{ from (a)}},
    \end{align*}
    which gives condition (b): $\Psi(t_1,t_2) \geq h_\ell'(t_1)$ with equality when $t_1 > 0$, and the local SOC: $h''_\ell(t_1) \cdot t_1 \leq 0$ in condition (c), as desired.
\end{proof}

\begin{proof}[Proof of sufficiency]
By the construction of the strategy, \eqref{eqn:OCC} is slack for $0 < t < t_1$. Therefore, the complementary slackness condition implies $\Lambda(t)$ must be constant prior to $t_1$, i.e., $\Lambda(0+) = \Lambda(t_1)$. The principal's derivative-of-Lagrangian from choosing stopping belief $\nu$ at time $t > t_1$ is
\begin{align*}
    l_{f,\Lambda}(\nu,t)=
    \begin{cases}
    v_\ell+h_\ell(t)+\Lambda(t)-\Lambda(t)c(t)+\int_{t_1}^t c(s)\d \Lambda(s)-\Lambda(t_1) - 2\nu(\Lambda(t) - \Lambda(t_1)) & \nu < 0.5, \\
    v_r+h_r(t)-\Lambda(t)c(t)+\int_{t_1}^t c(s)\d \Lambda(s) - \Lambda(t_1) + 2\nu \Lambda(t_1) & \nu \ge 0.5.
    \end{cases}
\end{align*}
Select $\nabla U(0,t)$ to be $\nabla U(\mu^*_{t_1}(t^*),t)$. Define $l^*_{f,\Lambda}(\nu,t) = l_{f,\Lambda}(\nu,t) - 2\nu \Lambda(t_1)$ as an affine transformation of $l_{f,\Lambda}$. The key observation is that $l^*_{f,\Lambda}$ as a function of $\nu$ takes a simple piecewise linear structure: it linearly decreases on $[0,0.5)$ at rate $2\Lambda(t)$ and stays constant on $[0.5,1]$, as is depicted by \cref{fig:alibi}.

The concavification method suggests that the following three conditions are sufficient for the optimality of $f$: (i) $l^*_{f,\Lambda}(0,t)$ is a constant function of $t$ for $t \ge t_1$, (ii) $l^*_{f,\Lambda}(0,t) \le l_{f,\Lambda}(0,t_1)$ for $t < t_1$, and (iii) $l^*_{f,\Lambda}(\mu^*_{t_1}(t_2),t)$ is maximized at $t_2$. We construct $\Lambda$:
\begin{align*}
    \Lambda(t)=
    \begin{dcases}
            \Psi(t,t_2):=\int_t^{t_2}e^{t-s}h'_\ell(s)\d s +e^{t-t_2}h'_r(t_2), & \forall t \ge t_1, \\
            \Psi(t_1,t_2), & \forall t \in (0,t_1),
    \end{dcases}
\end{align*}
where $t_2$ is given by $\Delta v+\Delta h(t_2) = h'_r(t_2)$. When $t_1 > 0$, condition (b) guarantees that $\Lambda'(t_1) = 0$, i.e.\ $\Lambda$ is a smooth function. Note that $\Lambda$ is monotonically increasing if for all $t \ge t_1$,
\begin{align*}
    \Lambda'(t)=\Psi'_t(t,t_2)\ge 0,
\end{align*}
which is given by condition (d). Next, we verify the three conditions.

\emph{Condition (i)} is satisfied if for all $t \ge t_1$, $\Lambda$ satisfies
\begin{align*}
    \frac{\partial l^*_{f,\Lambda}(0,t)}{\partial t}=h'_\ell(t)+\Lambda'(t)-\Lambda(t)=0,
\end{align*}
which can be verified by direct calculation.

\emph{Condition (ii)} has bite only when $t_1 > 0$, in which case, it is satisfied if $\frac{\partial l^*_{f,\Lambda}(0,t)}{\partial t}|_{t=t_1-} \ge 0$ and $\frac{\partial^2 l^*_{f,\Lambda}(0,t)}{\partial t^2} \le 0$ for $t < t_1$. The former condition is implied by condition (b). The latter condition is implied by (d).

\emph{Condition (iii)}: we verify the optimality of $\mu^*_{t_1}(t_2)$ and $t_2$ separately.
\begin{itemize}[noitemsep]
    \item $\mu^*_{t_1}(t_2) \in (0.5,1)$ is optimal $\iff$ $l^*_{f,\Lambda}(\cdot,t_2)$ attains the same value at $0$ and $\mu^*_{t_1}(t_2)$ $\iff$ $\Delta v + \Delta h(t_2) = \Lambda(t_2) = h'_r(t_2)$ (condition (a)). See the illustration in \cref{fig:alibi}, where the colored lines are $l^*_{f,\Lambda}(\nu,\widehat{t})$ and the dashed line is the concave envelope.
    \item $t_2$ is optimal if $\frac{\partial l^*_{f,\Lambda}(1,t)}{\partial t}|_{t=t_2}= 0$ and $\frac{\partial^2 l^*_{f,\Lambda}(1,t)}{\partial t^2} \le 0$. The former is equivalent to $h'_r(t_2) - \Lambda(t_2) = 0$, which is implied by the definition of $\Lambda$. The latter is implied by $h''_r(t) - \Lambda'(t) \le 0$, i.e.,
    \begin{align*}
        \begin{dcases}
            h''_r(t) \le 0, & \forall t < t_1, \\
            \Psi'_t(t,t_2) - h''_r(t) \ge 0, & \forall t \ge t_1,
        \end{dcases}
    \end{align*}
    which are given by condition (d).
\end{itemize}

\begin{figure}[htbp]
\centering
\vspace{-1em}
\tikzset{every picture/.style={line width=0.75pt}} 

\begin{tikzpicture}[x=0.75pt,y=0.75pt,yscale=-1,xscale=1]

\draw   (100,145.29) -- (440,145.29) -- (440,237.06) -- (100,237.06) -- cycle ;
\draw [color={rgb, 255:red, 0; green, 122; blue, 255 }  ,draw opacity=1 ][line width=2.25]    (100,183.53) -- (240,237.06) ;
\draw [color={rgb, 255:red, 208; green, 2; blue, 27 }  ,draw opacity=1 ][line width=2.25]    (280,183.53) -- (440,183.53) ;
\draw  [dash pattern={on 4.5pt off 4.5pt}]  (100,183.53) -- (280,183.53) ;
\draw    (360,168.24) -- (360,180.53) ;
\draw [shift={(360,183.53)}, rotate = 270] [fill={rgb, 255:red, 0; green, 0; blue, 0 }  ][line width=0.08]  [draw opacity=0] (8.93,-4.29) -- (0,0) -- (8.93,4.29) -- cycle    ;
\draw    (120,168.24) -- (102.38,181.71) ;
\draw [shift={(100,183.53)}, rotate = 322.59] [fill={rgb, 255:red, 0; green, 0; blue, 0 }  ][line width=0.08]  [draw opacity=0] (8.93,-4.29) -- (0,0) -- (8.93,4.29) -- cycle    ;

\draw (91,244.99) node [anchor=north west][inner sep=0.75pt]  [font=\small]  {$0$};
\draw (271,246.52) node [anchor=north west][inner sep=0.75pt]  [font=\small]  {$0.5$};
\draw (428,246.52) node [anchor=north west][inner sep=0.75pt]  [font=\small]  {$1$};
\draw (68,130.28) node [anchor=north west][inner sep=0.75pt]  [font=\small]  {$FOC$};
\draw (344,149.16) node [anchor=north west][inner sep=0.75pt]  [font=\small]  {$\mu^{*}_{t_1}(t_2)$};
\draw (113,153.99) node [anchor=north west][inner sep=0.75pt]  [font=\small]  {$\nu =0$};

\end{tikzpicture}
\vspace{-0.5em}
\caption{Concavification condition for \emph{Suspense}${}^{t_1}$-$\ell^{t_2}$ at $t=t_2$.}
\label{fig:alibi}
\end{figure}
\end{proof}

\subsection{Proof of \texorpdfstring{\cref{cor:binary_suspense_R}}{Corollary 2.1}}\label{appendix:corollary_3_1}

\Cref{cor:binary_suspense_R} follows from \cref{thm:binary_suspense_L} by exchanging the labels of the two states and the two actions, $L \leftrightarrow R$ and $\ell \leftrightarrow r$. Under this relabeling, the prior $\mu_0(R) \in (0, 1/2)$ becomes $\mu_0(L) \in (0, 1/2)$, the role of $\Delta v$ is reversed, and the marginal delay gains $h_\ell$ and $h_r$ swap. The Suspense${}^{t_1}$-$r^{t_2}$ strategy is the image of the Suspense${}^{t_1}$-$\ell^{t_2}$ strategy under this relabeling, with belief paths $\mu^L_t, \mu^R_t$ replaced by their symmetric counterparts and the calibrated Poisson reveal directed at state $R$ rather than state $L$. The conditions (a)--(d) of \cref{cor:binary_suspense_R} are precisely the conditions of \cref{thm:binary_suspense_L} after the relabeling: condition (a) becomes $\Delta v + \Delta h(t_2) = -h'_\ell(t_2)$, condition (b) involves the relabeled $\Psi(t_1,t_2) = \int_{t_1}^{t_2} e^{-s+t_1} h'_r(s) \,\d s + e^{-t_2+t_1} h'_\ell(t_2)$, and conditions (c) and (d) inherit the same structure with $h_\ell$ and $h_r$ swapped. The necessity argument of \cref{appendix:binary_proofs} and the construction of the multiplier $\Lambda$ apply verbatim to the relabeled environment, so the conclusion of \cref{thm:binary_suspense_L} transfers to \cref{cor:binary_suspense_R}.

\subsection{Proof of \texorpdfstring{\cref{thm:time_risk_loving}}{Theorem 3}}\label{appendix:time_risk_loving_proof}

\begin{proof}
    Let $\Lambda$ be the multiplier that solves \eqref{eqn:dual}. We claim that $\Lambda(t)$ must be strictly increasing on $[0,\overline{t}]$. Suppose for the purpose of contradiction that this is not true, i.e., there exists an interval $[t_1,t_2]$ such that $\Lambda(t) = \kappa$ on the interval. Without loss of generality, let $[t_1,t_2]$ be a maximal interval such that this is true. Since $\overline{t} < \overline{T}$, $t_2 < \overline{T}$.\footnote{In the special case where $\overline{T} = \infty$, $t_2 = \infty$ is already ruled out because this would imply $t_1 \ge \overline{t}$.}

    Define $\xi(t) = \max_{\mu}(l_{f,\Lambda}(\mu,t) - a \cdot \mu)$. Then \eqref{eqn:foc} implies that $\xi(t) \le 0$ and $\xi(t_2) = 0$. Note that $l_{f,\Lambda}(\mu,t)$ is left continuous and only jumps up in $t$. The envelope theorem implies
    \begin{align*}
        \xi'(t_2^-)&=V'_t(\mu,t_2)+\kappa U'_t(\mu,t_2),
    \end{align*}
    where $\mu$ attains $\xi(t_2)$. Next, we calculate a lower bound for ``$\xi'(t_2^+)$''.
    \begin{align*}
        &\xi(t_2+\delta)-\xi(t_2)\\
        \ge& \, l_{f,\Lambda}(\mu,t_2+ \delta)-l_{f,\Lambda}(\mu,t_2)\\
        =& \, V(\mu,t_2+\delta)-V(\mu,t_2)+\Big(\Lambda(t_2+\delta) - \kappa + \kappa \Big)U(\mu,t_2+\delta)-\kappa U(\mu,t_2) -\int_{t_2}^{t_2+\delta} \nabla_{\mu}U(\widehat{\mu}_{t},t)\cdot \mu \, \Lambda(\d t)\\
        =& \, V(\mu,t_2+\delta)-V(\mu,t_2)+\kappa\Big(U(\mu,t_2+\delta)-U(\mu,t_2)\Big)\\
        &+ \Big(\Lambda(t_2 + \delta) - \kappa \Big) U(\mu,t_2 + \delta) -\int_{t_2}^{t_2+\delta} \nabla_{\mu}U(\widehat{\mu}_{t},t)\cdot \mu \, \Lambda(\d t) \\
        \ge& \, V(\mu,t_2+\delta)-V(\mu,t_2)+\kappa\Big(U(\mu,t_2+\delta)-U(\mu,t_2)\Big) +\int_{t_2}^{t_2+\delta} \Big(U(\mu,t)-\nabla_{\mu}U(\widehat{\mu}_{t},t)\cdot \mu \Big) \Lambda(\d t)\\
        \ge& \, V(\mu,t_2+\delta)-V(\mu,t_2)+\kappa\Big(U(\mu,t_2+\delta)-U(\mu,t_2)\Big)\\
        =& \, \xi'(t_2^-)\delta+o(\delta),
    \end{align*}
    where the second inequality uses $U(\mu,t_2+\delta) \leq U(\mu,t_2)$, and the third inequality uses the convexity of $U$ in $\mu$.

    Suppose $V''_t > 0$ and $U''_t \ge 0$. Then $V + \kappa U$ is strictly convex in $t$. Suppose instead $V'_t > 0$. We claim that $\kappa > 0$. If not, $\xi(t_2+\delta) \ge V'_t(\mu,t_2)\delta + o(\delta) > 0$, violating \eqref{eqn:foc} around $t_2$. Then, since $\kappa > 0$, $V + \kappa U$ is strictly convex in $t$.

    Therefore, $l_{f,\Lambda}(\mu,t) - a \cdot \mu$ is strictly convex in $t$ on $[t_1,t_2]$ for fixed $\mu$. Since $\xi(t) := \max_{\mu}(l_{f,\Lambda}(\mu,t) - a \cdot \mu)$ is the upper envelope of a collection of strictly convex functions, it is also strictly convex on $[t_1,t_2]$. Hence, $\xi'(t_2^-) > 0$, which violates \eqref{eqn:foc} at $t_2^+$, leading to a contradiction.

    Since $\Lambda(t)$ must be strictly increasing, the complementary slackness condition implies that \eqref{eqn:OCC} must be binding throughout the support of $f$.
\end{proof}

\subsection{Proof of \texorpdfstring{\cref{thm:time_risk_averse}}{Theorem 4}}\label{appendix:time_risk_averse_proof}

\begin{proof}
    Let $\xi(t) = \max_{\mu}(l_{f,\Lambda}(\mu,t) - a \cdot \mu)$. Then \eqref{eqn:foc} implies that $\xi(t) \le 0$ and $\xi(\underline{t}) = \xi(\overline{t}) = 0$. The envelope theorem implies $\xi'(\underline{t}^-) = V'_t(\mu,\underline{t}) + \Lambda(0) U'_t(\mu,\underline{t}) \ge 0$, where $\mu$ is the maximizer that attains $\xi(\underline{t})$. Suppose for the purpose of contradiction that $\overline{t} > \max\left\{\overline{J}_V^{-1} \circ \underline{J}_V(\underline{t}), \overline{J}_U^{-1} \circ \underline{J}_U(\underline{t})\right\}$. Then there exists $\epsilon > 0$ such that for all $\mu \in \Delta(\Theta)$ and all $t > \overline{t} - \epsilon$,
    \begin{align*}
        V'_t(\mu,t) + \Lambda(t) U'_t(\mu,t) < -\epsilon.
    \end{align*}
    Since $\widehat{\mu}_t \to \mu^*$ and $\frac{\Lambda(\overline{t}) - \Lambda(t)}{\overline{t} - t}$ is bounded as $t \to \overline{t}^-$ (by \cref{defi:regular}), there exists $\delta > 0$ such that $\delta < \epsilon$ and for all $t \ge \overline{t} - \delta$,
    \begin{align*}
        \frac{\Lambda(\overline{t}) - \Lambda(\overline{t} - \delta)}{\delta}\Big(U(\mu^*,t) - \nabla_{\mu}U(\widehat{\mu}_t,t)\cdot \mu^*\Big) < \frac{1}{2}\epsilon.
    \end{align*}
    Then, for all $t_1, t_2$ such that $\overline{t} - \delta < t_1 < t_2 < \overline{t}$,
    \begin{align*}
        &(l_{f,\Lambda}(\mu^*,t_2) - a \cdot \mu^*)\\
        = & \, (l_{f,\Lambda}(\mu^*,t_1) - a \cdot \mu^*) + \int_{t_1}^{t_2} \big(V'_t(\mu^*,t) + \Lambda(t) U'_t(\mu^*,t)\big) \d t + \int_{t_1}^{t_2}\Big(U(\mu^*,t) - \nabla_{\mu}U(\widehat{\mu}_t,t)\cdot \mu^*\Big)\Lambda(\d t)\\
        <& \, -\epsilon (t_2 - t_1) + \frac{\epsilon \delta}{2(\Lambda(\overline{t}) - \Lambda(\overline{t} - \delta))}\int_{t_1}^{t_2}\Lambda(\d t)\\
        & \to -\frac{\epsilon \delta}{2} \quad \text{as } t_2 \to \overline{t},\ t_1 = \overline{t} - \delta.
    \end{align*}
    The analysis above implies that $l_{f,\Lambda}(\mu^*,t)$ is bounded away from $0$ as $t \to \overline{t}^-$, contradicting $(\mu^*, \overline{t}) \in \mathrm{supp}(f)$.
\end{proof}

\subsection{Formal version of \texorpdfstring{\cref{prop:binary_completeness_informal}}{Proposition 1}}\label{appendix:binary_completeness_formal}

We state and prove the formal version of the informal completeness \cref{prop:binary_completeness_informal} introduced in \cref{ssec:binary_characterization}. The formal statement requires a slight generalization of the Suspense-$\ell/r$ strategies of \cref{sec:binary} that accommodates several corner cases and a close variant. We first define the variants:

\begin{itemize}
    \item \textbf{Suspense${}^{t_1}$-Inconclusive-$\ell^{t_2}$} runs in two stages and is depicted in \cref{fig:inconclusive_L}. During the first stage $[0,t_1)$, the principal generates suspense with the two belief paths $\mu^L_t$ and $\mu^R_t$.\footnote{They are given by $\mu^L_t = \mu_{t_1}^L + \frac{t_1 - t}{A}(\frac{1}{2} - \mu_{t_1}^L)$ and $\mu^R_t = \mu_{t_1}^R - \frac{t_1 - t}{A}(\mu_{t_1}^R - \frac{1}{2})$, where $A = t_1 \Big(\frac{1/2 - \mu^L_{t_1}}{\mu_0 - \mu_{t_1}^L}\Big)$. $\mu^{L}_{t_1}$ and $\mu^R_{t_1}$ will be determined shortly.} The two paths are chosen such that the agent obtains no continuation surplus over the duration of suspense generation. At the end of the suspense stage $t = t_1$, either the agent is more confident that the state is $L$ ($\mu^L_{t_1} < \mu_0$) and will not be provided further information, which induces stopping and choosing $\ell$, or $\mu^R_{t_1} = e^{-t_2 + t_1}$, which leads to the second stage.\footnote{$\mu_{t_1}^L$ is computed from the agent's indifference condition between stopping at time $0$ and time $t_1$. We obtain $\mu_{t_1}^L = \frac{1}{2}\big(\frac{2\mu_0 - t_1 - p_+}{1 - p_+}\big)$, where $p_+ = \frac{t_1}{2 e^{-t_2 + t_1} - 1}$. Implicitly, we require $t_1 < \mu_0$ to guarantee there exists $t_2$ such that $\mu^R_{t_1} \in (0,1)$.}

    In the second stage $[t_1,t_2]$, the strategy reveals state $L$ at a rate chosen such as to keep the agent's continuation incentive \eqref{eqn:OCC} binding. This, in turn, pins down a unique continuing belief path $\mu^*_{t_1}(t) := \mu_{t_1}^R e^{t - t_1}$. By time $t_2$, state $L$ will be fully revealed. Thus, at time $t_2$ the agent is certain the state is $R$, which induces stopping and choosing $r$.

    \item \textbf{Suspense${}^{t_1}$-Inconclusive-$r^{t_2}$} is defined analogously. The strategy gradually reveals state $R$ in the second stage $[t_1,t_2]$, as depicted in \cref{fig:inconclusive_R}.
\end{itemize}

\begin{figure}[htbp]
    \centering
    \begin{minipage}{0.49\linewidth}
        \tikzset{every picture/.style={line width=0.75pt}} 

\begin{tikzpicture}[x=0.75pt,y=0.75pt,yscale=-1,xscale=1]

\draw [color={rgb, 255:red, 65; green, 117; blue, 5 }  ,draw opacity=1 ] [dash pattern={on 2.25pt off 1.5pt}]  (141,123.59) -- (60.16,109.4) ;
\draw [color={rgb, 255:red, 0; green, 122; blue, 255 }  ,draw opacity=1 ]   (60.16,111.48) -- (140.23,111.48) ;
\draw   (140.9,140.1) .. controls (140.9,144.77) and (143.23,147.1) .. (147.9,147.1) -- (173.7,147.1) .. controls (180.37,147.1) and (183.7,149.43) .. (183.7,154.1) .. controls (183.7,149.43) and (187.03,147.1) .. (193.7,147.1)(190.7,147.1) -- (219.5,147.1) .. controls (224.17,147.1) and (226.5,144.77) .. (226.5,140.1) ;
\draw   (60.1,140.1) .. controls (60.1,144.77) and (62.43,147.1) .. (67.1,147.1) -- (89.8,147.1) .. controls (96.47,147.1) and (99.8,149.43) .. (99.8,154.1) .. controls (99.8,149.43) and (103.13,147.1) .. (109.8,147.1)(106.8,147.1) -- (132.5,147.1) .. controls (137.17,147.1) and (139.5,144.77) .. (139.5,140.1) ;
\draw  [draw opacity=0][fill={rgb, 255:red, 0; green, 0; blue, 0 }  ,fill opacity=0.17 ] (60.09,20.4) -- (309.82,20.4) -- (309.82,80.4) -- (60.09,80.4) -- cycle ;
\draw   (60.09,20.4) -- (309.82,20.4) -- (309.82,140.4) -- (60.09,140.4) -- cycle ;
\draw  [dash pattern={on 0.84pt off 2.51pt}]  (165.53,46.57) .. controls (172.59,77.89) and (173.71,110.5) .. (166.65,137.99) ;
\draw [shift={(166.21,139.67)}, rotate = 285.16] [fill={rgb, 255:red, 0; green, 0; blue, 0 }  ][line width=0.08]  [draw opacity=0] (7.2,-1.8) -- (0,0) -- (7.2,1.8) -- cycle    ;
\draw  [color={rgb, 255:red, 255; green, 255; blue, 255 }  ,draw opacity=1 ] (19.4,0.4) -- (329.4,0.4) -- (329.4,160.4) -- (19.4,160.4) -- cycle ;
\draw  [color={rgb, 255:red, 208; green, 2; blue, 27 }  ,draw opacity=1 ][fill={rgb, 255:red, 208; green, 2; blue, 27 }  ,fill opacity=1 ] (139.21,123.59) .. controls (139.21,122.44) and (140.01,121.51) .. (141,121.51) .. controls (141.99,121.51) and (142.79,122.44) .. (142.79,123.59) .. controls (142.79,124.74) and (141.99,125.67) .. (141,125.67) .. controls (140.01,125.67) and (139.21,124.74) .. (139.21,123.59) -- cycle ;
\draw  [dash pattern={on 0.84pt off 2.51pt}]  (199.2,37.9) .. controls (206.25,69.22) and (206.77,110.4) .. (199.69,138.24) ;
\draw [shift={(199.25,139.93)}, rotate = 285.16] [fill={rgb, 255:red, 0; green, 0; blue, 0 }  ][line width=0.08]  [draw opacity=0] (7.2,-1.8) -- (0,0) -- (7.2,1.8) -- cycle    ;
\draw [color={rgb, 255:red, 0; green, 122; blue, 255 }  ,draw opacity=1 ]   (140.14,52.6) .. controls (174,40.67) and (193,33.67) .. (226,20.67) ;
\draw  [color={rgb, 255:red, 208; green, 2; blue, 27 }  ,draw opacity=1 ][fill={rgb, 255:red, 208; green, 2; blue, 27 }  ,fill opacity=1 ] (223.42,20.67) .. controls (223.42,19.52) and (224.22,18.59) .. (225.21,18.59) .. controls (226.2,18.59) and (227,19.52) .. (227,20.67) .. controls (227,21.82) and (226.2,22.75) .. (225.21,22.75) .. controls (224.22,22.75) and (223.42,21.82) .. (223.42,20.67) -- cycle ;
\draw [color={rgb, 255:red, 208; green, 2; blue, 27 }  ,draw opacity=1 ][line width=1.5]    (140.16,140.28) -- (226.33,140.28) ;
\draw [color={rgb, 255:red, 189; green, 16; blue, 224 }  ,draw opacity=1 ] [dash pattern={on 2.25pt off 1.5pt}]  (140.14,52.6) -- (59.75,72.68) ;
\draw  [fill={rgb, 255:red, 255; green, 255; blue, 255 }  ,fill opacity=1 ] (58.37,111.48) .. controls (58.37,110.33) and (59.17,109.4) .. (60.16,109.4) .. controls (61.15,109.4) and (61.95,110.33) .. (61.95,111.48) .. controls (61.95,112.63) and (61.15,113.56) .. (60.16,113.56) .. controls (59.17,113.56) and (58.37,112.63) .. (58.37,111.48) -- cycle ;
\draw  [dash pattern={on 0.84pt off 2.51pt}]  (84.42,66.65) .. controls (91.05,81.71) and (94.3,96.75) .. (87.65,113.84) ;
\draw [shift={(87,115.43)}, rotate = 292.91] [fill={rgb, 255:red, 0; green, 0; blue, 0 }  ][line width=0.08]  [draw opacity=0] (7.2,-1.8) -- (0,0) -- (7.2,1.8) -- cycle    ;
\draw  [dash pattern={on 0.84pt off 2.51pt}]  (115.25,119.18) .. controls (121.56,103.66) and (123.63,79.67) .. (115.76,60.22) ;
\draw [shift={(115,58.43)}, rotate = 66.1] [fill={rgb, 255:red, 0; green, 0; blue, 0 }  ][line width=0.08]  [draw opacity=0] (7.2,-1.8) -- (0,0) -- (7.2,1.8) -- cycle    ;
\draw [color={rgb, 255:red, 189; green, 16; blue, 224 }  ,draw opacity=1 ] [dash pattern={on 2.25pt off 1.5pt}]  (259,55.18) -- (240.5,55.18) ;
\draw [color={rgb, 255:red, 65; green, 117; blue, 5 }  ,draw opacity=1 ] [dash pattern={on 2.25pt off 1.5pt}]  (259,74.93) -- (240.5,74.93) ;
\draw [color={rgb, 255:red, 0; green, 122; blue, 255 }  ,draw opacity=1 ]   (241.41,94.48) -- (258.75,94.48) ;
\draw [color={rgb, 255:red, 208; green, 2; blue, 27 }  ,draw opacity=1 ]   (241.41,114.23) -- (258.75,114.23) ;
\draw  [color={rgb, 255:red, 0; green, 122; blue, 255 }  ,draw opacity=1 ][fill={rgb, 255:red, 255; green, 255; blue, 255 }  ,fill opacity=1 ] (138.44,111.48) .. controls (138.44,110.33) and (139.24,109.4) .. (140.23,109.4) .. controls (141.22,109.4) and (142.02,110.33) .. (142.02,111.48) .. controls (142.02,112.63) and (141.22,113.56) .. (140.23,113.56) .. controls (139.24,113.56) and (138.44,112.63) .. (138.44,111.48) -- cycle ;
\draw  [color={rgb, 255:red, 208; green, 2; blue, 27 }  ,draw opacity=1 ][fill={rgb, 255:red, 255; green, 255; blue, 255 }  ,fill opacity=1 ] (139.16,140.28) .. controls (139.16,139.13) and (139.97,138.2) .. (140.95,138.2) .. controls (141.94,138.2) and (142.75,139.13) .. (142.75,140.28) .. controls (142.75,141.43) and (141.94,142.36) .. (140.95,142.36) .. controls (139.97,142.36) and (139.16,141.43) .. (139.16,140.28) -- cycle ;

\draw (42.43,132.2) node [anchor=north west][inner sep=0.75pt]  [font=\footnotesize]  {$L$};
\draw (306.4,144.8) node [anchor=north west][inner sep=0.75pt]  [font=\small]  {$t$};
\draw (39.69,74.8) node [anchor=north west][inner sep=0.75pt]  [font=\small]  {$0.5$};
\draw (136.1,147.3) node [anchor=north west][inner sep=0.75pt]  [font=\small]  {$t_{1}$};
\draw (21.15,102.8) node [anchor=north west][inner sep=0.75pt]  [font=\small]  {$Prior$};
\draw (224.1,148.2) node [anchor=north west][inner sep=0.75pt]  [font=\small]  {$t_{2}$};
\draw (65.41,153.51) node [anchor=north west][inner sep=0.75pt]  [font=\footnotesize]  {$ \begin{array}{l}
Suspense-\\
generating
\end{array}$};
\draw (160.2,152.61) node [anchor=north west][inner sep=0.75pt]  [font=\footnotesize]  {$ \begin{array}{l}
\ell -\\
targeting
\end{array}$};
\draw (260.75,45.95) node [anchor=north west][inner sep=0.75pt]  [font=\footnotesize]  {$:\mu _{t}^{R}$};
\draw (260.5,65.2) node [anchor=north west][inner sep=0.75pt]  [font=\footnotesize]  {$:\mu _{t}^{L}$};
\draw (260.25,84.7) node [anchor=north west][inner sep=0.75pt]  [font=\footnotesize]  {$:\hat{\mu }_{t}$};
\draw (260.25,104.45) node [anchor=north west][inner sep=0.75pt]  [font=\footnotesize]  {$:supp( f)$};
\draw (42.63,17.4) node [anchor=north west][inner sep=0.75pt]  [font=\footnotesize]  {$R$};

\end{tikzpicture}
        \vspace{-2em}
        \caption{\emph{Suspense-Inconclusive-$\ell$} strategy.}
        \label{fig:inconclusive_L}
    \end{minipage}
    \begin{minipage}{0.49\linewidth}
        \tikzset{every picture/.style={line width=0.75pt}} 

\begin{tikzpicture}[x=0.75pt,y=0.75pt,yscale=-1,xscale=1]

\draw  [draw opacity=0][fill={rgb, 255:red, 0; green, 0; blue, 0 }  ,fill opacity=0.17 ] (62.09,11.23) -- (311.82,11.23) -- (311.82,71.23) -- (62.09,71.23) -- cycle ;
\draw [color={rgb, 255:red, 0; green, 122; blue, 255 }  ,draw opacity=1 ]   (63.48,94.74) -- (131.33,94.74) ;
\draw   (131.67,130.93) .. controls (131.67,135.6) and (134,137.93) .. (138.67,137.93) -- (157.58,137.93) .. controls (164.25,137.93) and (167.58,140.26) .. (167.58,144.93) .. controls (167.58,140.26) and (170.91,137.93) .. (177.58,137.93)(174.58,137.93) -- (196.49,137.93) .. controls (201.16,137.93) and (203.49,135.6) .. (203.49,130.93) ;
\draw   (62.1,130.93) .. controls (62.1,135.6) and (64.43,137.93) .. (69.1,137.93) -- (86.88,137.93) .. controls (93.55,137.93) and (96.88,140.26) .. (96.88,144.93) .. controls (96.88,140.26) and (100.21,137.93) .. (106.88,137.93)(103.88,137.93) -- (124.67,137.93) .. controls (129.34,137.93) and (131.67,135.6) .. (131.67,130.93) ;
\draw   (62.09,11.23) -- (311.82,11.23) -- (311.82,131.23) -- (62.09,131.23) -- cycle ;
\draw  [fill={rgb, 255:red, 255; green, 255; blue, 255 }  ,fill opacity=1 ] (59.9,94.74) .. controls (59.9,93.59) and (60.7,92.66) .. (61.69,92.66) .. controls (62.68,92.66) and (63.48,93.59) .. (63.48,94.74) .. controls (63.48,95.89) and (62.68,96.82) .. (61.69,96.82) .. controls (60.7,96.82) and (59.9,95.89) .. (59.9,94.74) -- cycle ;
\draw  [color={rgb, 255:red, 208; green, 2; blue, 27 }  ,draw opacity=1 ][fill={rgb, 255:red, 208; green, 2; blue, 27 }  ,fill opacity=1 ] (130.97,25.34) .. controls (130.97,24.19) and (131.77,23.26) .. (132.76,23.26) .. controls (133.75,23.26) and (134.55,24.19) .. (134.55,25.34) .. controls (134.55,26.49) and (133.75,27.42) .. (132.76,27.42) .. controls (131.77,27.42) and (130.97,26.49) .. (130.97,25.34) -- cycle ;
\draw [color={rgb, 255:red, 0; green, 0; blue, 0 }  ,draw opacity=1 ] [dash pattern={on 0.84pt off 2.51pt}]  (160.06,113.83) .. controls (167.11,82.87) and (169.07,40.87) .. (162.04,13.3) ;
\draw [shift={(161.6,11.63)}, rotate = 74.67] [fill={rgb, 255:red, 0; green, 0; blue, 0 }  ,fill opacity=1 ][line width=0.08]  [draw opacity=0] (7.2,-1.8) -- (0,0) -- (7.2,1.8) -- cycle    ;
\draw  [color={rgb, 255:red, 255; green, 255; blue, 255 }  ,draw opacity=1 ] (21.4,-8.77) -- (331.4,-8.77) -- (331.4,151.23) -- (21.4,151.23) -- cycle ;
\draw [color={rgb, 255:red, 0; green, 0; blue, 0 }  ,draw opacity=1 ] [dash pattern={on 0.84pt off 2.51pt}]  (194.34,124.83) .. controls (201.4,93.87) and (201.54,41.31) .. (194.44,13.32) ;
\draw [shift={(194,11.63)}, rotate = 74.67] [fill={rgb, 255:red, 0; green, 0; blue, 0 }  ,fill opacity=1 ][line width=0.08]  [draw opacity=0] (7.2,-1.8) -- (0,0) -- (7.2,1.8) -- cycle    ;
\draw [color={rgb, 255:red, 0; green, 122; blue, 255 }  ,draw opacity=1 ]   (130.06,108.12) .. controls (153,114) and (195.49,126.83) .. (203.49,130.26) ;
\draw [color={rgb, 255:red, 208; green, 2; blue, 27 }  ,draw opacity=1 ][line width=1.5]    (132.76,11.34) -- (204.06,11.34) ;
\draw [color={rgb, 255:red, 189; green, 16; blue, 224 }  ,draw opacity=1 ] [dash pattern={on 2.25pt off 1.5pt}]  (261,44.01) -- (242.5,44.01) ;
\draw [color={rgb, 255:red, 65; green, 117; blue, 5 }  ,draw opacity=1 ] [dash pattern={on 2.25pt off 1.5pt}]  (261,63.76) -- (242.5,63.76) ;
\draw [color={rgb, 255:red, 0; green, 122; blue, 255 }  ,draw opacity=1 ]   (243.41,83.31) -- (260.75,83.31) ;
\draw [color={rgb, 255:red, 208; green, 2; blue, 27 }  ,draw opacity=1 ]   (243.41,103.06) -- (260.75,103.06) ;
\draw [color={rgb, 255:red, 189; green, 16; blue, 224 }  ,draw opacity=1 ] [dash pattern={on 2.25pt off 1.5pt}]  (131.76,26.34) -- (61.32,47.94) ;
\draw [color={rgb, 255:red, 65; green, 117; blue, 5 }  ,draw opacity=1 ] [dash pattern={on 2.25pt off 1.5pt}]  (130.06,108.12) -- (61.69,94.74) ;
\draw  [dash pattern={on 0.84pt off 2.51pt}]  (83.85,40.2) .. controls (90.47,55.26) and (92.04,79.71) .. (85.28,97.36) ;
\draw [shift={(84.63,98.98)}, rotate = 292.91] [fill={rgb, 255:red, 0; green, 0; blue, 0 }  ][line width=0.08]  [draw opacity=0] (7.2,-1.8) -- (0,0) -- (7.2,1.8) -- cycle    ;
\draw  [dash pattern={on 0.84pt off 2.51pt}]  (107.77,103.55) .. controls (114.08,88.03) and (116.6,55.76) .. (108.76,35.81) ;
\draw [shift={(108,34)}, rotate = 66.1] [fill={rgb, 255:red, 0; green, 0; blue, 0 }  ][line width=0.08]  [draw opacity=0] (7.2,-1.8) -- (0,0) -- (7.2,1.8) -- cycle    ;
\draw  [color={rgb, 255:red, 208; green, 2; blue, 27 }  ,draw opacity=1 ][fill={rgb, 255:red, 208; green, 2; blue, 27 }  ,fill opacity=1 ] (201.69,130.26) .. controls (201.69,129.11) and (202.5,128.18) .. (203.49,128.18) .. controls (204.47,128.18) and (205.28,129.11) .. (205.28,130.26) .. controls (205.28,131.41) and (204.47,132.34) .. (203.49,132.34) .. controls (202.5,132.34) and (201.69,131.41) .. (201.69,130.26) -- cycle ;
\draw  [color={rgb, 255:red, 0; green, 122; blue, 255 }  ,draw opacity=1 ][fill={rgb, 255:red, 255; green, 255; blue, 255 }  ,fill opacity=1 ] (129.54,94.74) .. controls (129.54,95.88) and (130.34,96.8) .. (131.33,96.8) .. controls (132.32,96.8) and (133.12,95.88) .. (133.12,94.74) .. controls (133.12,93.61) and (132.32,92.69) .. (131.33,92.69) .. controls (130.34,92.69) and (129.54,93.61) .. (129.54,94.74) -- cycle ;
\draw  [color={rgb, 255:red, 208; green, 2; blue, 27 }  ,draw opacity=1 ][fill={rgb, 255:red, 255; green, 255; blue, 255 }  ,fill opacity=1 ] (130.97,11.34) .. controls (130.97,10.19) and (131.77,9.26) .. (132.76,9.26) .. controls (133.75,9.26) and (134.55,10.19) .. (134.55,11.34) .. controls (134.55,12.49) and (133.75,13.42) .. (132.76,13.42) .. controls (131.77,13.42) and (130.97,12.49) .. (130.97,11.34) -- cycle ;

\draw (45.86,124.6) node [anchor=north west][inner sep=0.75pt]  [font=\footnotesize]  {$L$};
\draw (308.4,135.63) node [anchor=north west][inner sep=0.75pt]  [font=\small]  {$t$};
\draw (45.48,5.23) node [anchor=north west][inner sep=0.75pt]  [font=\footnotesize]  {$R$};
\draw (41.69,65.63) node [anchor=north west][inner sep=0.75pt]  [font=\small]  {$0.5$};
\draw (127.1,138.8) node [anchor=north west][inner sep=0.75pt]  [font=\small]  {$t_{1}$};
\draw (23.25,92.2) node [anchor=north west][inner sep=0.75pt]  [font=\small]  {$Prior$};
\draw (201.24,137.8) node [anchor=north west][inner sep=0.75pt]  [font=\small]  {$t_{2}$};
\draw (62.46,145.3) node [anchor=north west][inner sep=0.75pt]  [font=\footnotesize]  {$ \begin{array}{l}
Suspense-\\
generating
\end{array}$};
\draw (149.29,145.4) node [anchor=north west][inner sep=0.75pt]  [font=\footnotesize]  {$ \begin{array}{l}
r-\\
targeting
\end{array}$};
\draw (262.75,34.79) node [anchor=north west][inner sep=0.75pt]  [font=\footnotesize]  {$:\mu _{t}^{R}$};
\draw (262.5,54.04) node [anchor=north west][inner sep=0.75pt]  [font=\footnotesize]  {$:\mu _{t}^{L}$};
\draw (262.25,73.54) node [anchor=north west][inner sep=0.75pt]  [font=\footnotesize]  {$:\hat{\mu }_{t}$};
\draw (262.25,93.29) node [anchor=north west][inner sep=0.75pt]  [font=\footnotesize]  {$:supp( f)$};

\end{tikzpicture}
        \vspace{-2em}
        \caption{\emph{Suspense-Inconclusive-$r$} strategy.}
        \label{fig:inconclusive_R}
    \end{minipage}
\end{figure}

The difference between Suspense${}^{t_1}$-Inconclusive-$\ell^{t_2}$ and Suspense${}^{t_1}$-$\ell^{t_2}$ is that they have imperfectly revealing signals at the beginning and the end of the persuasion window, respectively. We also extend the definitions of Suspense${}^{t_1}$-$\ell/r^{t_2}$ to accommodate the following corner cases:
\begin{itemize}
    \item $t_1 = t_2 = 0$: in this case, all strategies reduce to the static persuasion strategy of \citet{kamenica2011bayesian}, where beliefs $\{0,0.5\}$ are induced at $t = 0$ if $\Delta v > 0$, and beliefs $\{0.5,1\}$ are induced if $\Delta v < 0$.
    \item $t_2 = \mu^{*-1}_{t_1}(1)$ or $\mu^{\dagger-1}_{t_1}(1)$: in these cases, the terminal belief of the Suspense${}^{t_1}$-$\ell^{t_2}$ upon no revelation is either $0$ or $1$. In these cases, the Suspense${}^{t_1}$-$\ell/r^{t_2}$ strategy coincides with the Suspense${}^{t_1}$-Inconclusive-$\ell/r^{t_2}$ strategy.
\end{itemize}

In short, we say a strategy is a \emph{generalized Suspense}-$\ell/r$ strategy if it satisfies any of the above definitions for some $t_1,t_2$. The following result provides sufficient conditions for the optimality of each generalized \emph{Suspense}-$\ell/r$ strategy, including the case of \cref{thm:binary_suspense_L}.

\begin{prop}\label{prop:binary_completeness_sufficient}
    Suppose $t_1, t_2 \in [0, \mu^{*-1}_{t_1}(1)]$ with $t_1 < t_2$. Let conditions (b)--(d) from \cref{thm:binary_suspense_L} hold. Then the following statements are true:
    \begin{enumerate}
        \item If $t_1 < t_2 < \mu^{*-1}_{t_1}(1)$ and $\Delta v + \Delta h(t_2) = h'_r(t_2)$, then \emph{Suspense}${}^{t_1}$-$\ell^{t_2}$ is optimal.
        \item If $t_1 < t_2 = \mu^{*-1}_{t_1}(1)$ and $\Delta v + \Delta h(t_2) \in [h'_r(t_2) - 2\Psi(t_1,t_2), h'_r(t_2)]$, then \emph{Suspense}${}^{t_1}$-$\ell^{t_2}$ is optimal.
        \item If $t_2 < \mu^{*-1}_{t_1}(1)$ and $\Delta v + \Delta h(t_2) = h'_r(t_2) - 2\Psi(t_1,t_2)$, then \emph{Suspense}${}^{t_1}$-Inconclusive-$\ell^{t_2}$ is optimal.
    \end{enumerate}
    Moreover, if $h_\ell$ and $h_r$ are concave and $|\Delta v| \geq \max\{h'_r(0), h'_\ell(0)\}$, then the $\text{Suspense}^{0}\text{-}\ell^{0}$ strategy is optimal.
\end{prop}

\begin{proof}
Part 1 is the same as \cref{thm:binary_suspense_L}.

\textbf{Proof of part 2.} We use the same $\Lambda$ as in the proof of \cref{thm:binary_suspense_L} in \cref{appendix:binary_proofs}. It remains to verify that $\max_t l^*_{f,\Lambda}(\nu,t)$ is concavified at $0$ and $1$. If $\nu \geq 0.5$, $\max_t l^*_{f,\Lambda}(\nu,t) = \max_t l^*_{f,\Lambda}(1,t) = l^*_{f,\Lambda}(1,t_2)$. If $\nu < 0.5$,
\begin{align*}
    \frac{\partial l^*_{f,\Lambda}(\nu,t)}{\partial t}\bigg|_{t=t_1-} &= h'_\ell(t_1) - \Lambda(t_1) \geq 0, \\
    \frac{\partial l^*_{f,\Lambda}(\nu,t)}{\partial t}\bigg|_{t=t_1+} &= h'_\ell(t_1) + (1 - 2\nu)\Lambda'(t_1) - \Lambda(t_1) \leq h'_\ell(t_1) + \Lambda'(t_1) - \Lambda(t_1) = 0.
\end{align*}
Since $l^*_{f,\Lambda}(\nu,\cdot)$ is concave, we have $\max_t l^*_{f,\Lambda}(\nu,t) = l^*_{f,\Lambda}(\nu,t_1)$ if $\nu < 0.5$. Thus, the statement that $\max_t l^*_{f,\Lambda}(\nu,t)$ is concavified at $0$ and $1$ is equivalent to
\[l^*_{f,\Lambda}(0,t_2) - 2\Lambda(t_1) \leq l^*_{f,\Lambda}(1,t_2) \leq l^*_{f,\Lambda}(0,t_2).\]
This inequality is equivalent to $\Delta v + \Delta h(t_2) \in [h'_r(t_2) - 2\Psi(t_1,t_2), h'_r(t_2)]$, as required.

\textbf{Proof of part 3.} We use the same $\Lambda$ as before. We showed earlier that $\max_t l^*_{f,\Lambda}(\mu^*(t_1),t) = l^*_{f,\Lambda}(\mu^*(t_1),t_1)$. It remains to verify that $\max_t l^*_{f,\Lambda}(\nu,t)$ is concavified at $0$, $1$, and $\mu^*(t_1)$. Since $\mu^*(t_1) < 0.5$, the concavification condition is equivalent to
\begin{align*}
    & l^*_{f,\Lambda}(\mu^*(t_1),t_1) = \mu^*(t_1) l^*_{f,\Lambda}(0,t_1) + (1 - \mu^*(t_1)) l^*_{f,\Lambda}(1,t_2) \\
    \iff & \Delta v + \Delta h(t_2) = h'_r(t_2) - 2\Psi(t_1,t_2),
\end{align*}
as required.

\textbf{Proof of optimality of the $\text{Suspense}^{0}\text{-}\ell^{0}$ strategy.} Under the $\text{Suspense}^{0}\text{-}\ell^{0}$ strategy with $\Lambda(t) = |\Delta v|$, the principal's derivative-of-Lagrangian from choosing stopping belief $\nu$ at time $t$ is
\begin{align*}
    l_{f,\Lambda}(\nu,t) = \begin{cases}
        v_\ell + h_\ell(t) - |\Delta v| t, & \nu < 0.5, \\
        v_r + h_r(t) + |\Delta v|(2\nu - 1) - |\Delta v| t, & \nu > 0.5.
    \end{cases}
\end{align*}
Since $h_\ell$ and $h_r$ are concave and $|\Delta v| \geq \max\{h'_r(0), h'_\ell(0)\}$, we have $l_{f,\Lambda}(\nu,t) \leq l_{f,\Lambda}(\nu,0)$ for every $\nu \in [0,1]$. It is easy to check that $l^*_{f,\Lambda}(\cdot,0)$ is concavified at $0$, $0.5$, and $1$, as desired.
\end{proof}

We finally provide a full characterization of optimal information structures when the principal's value is concave and either complementary or substitutable with the agent's action.

\begin{prop}\label{prop:binary_completeness_formal}
    If $h_\ell$ and $h_r$ are concave and $\Delta v > 0$, then:
    \begin{enumerate}
        \item (Complements) If $\Delta h' > 0$, there exist $t_1$ and $t_2$ such that the $\text{Suspense}^{t_1}\text{-}\ell^{t_2}$ strategy is optimal. Moreover, $t_1, t_2$ are decreasing in $\Delta v$.
        \item (Substitutes) If $\Delta h' < 0$, there exist $t_1$ and $t_2$ such that either the $\text{Suspense}^{t_1}\text{-}r^{t_2}$ strategy or the Suspense${}^{t_1}$-Inconclusive-$r^{t_2}$ strategy is optimal.
    \end{enumerate}
\end{prop}

\begin{proof}
We prove the following equivalent statement (obtained by flipping labels): if $h_\ell$ and $h_r$ are concave and $\Delta h' > 0$, then
    \begin{enumerate}
        \item (Complements) If $\Delta v > 0$, there exist $t_1$ and $t_2$ such that the $\text{Suspense}^{t_1}\text{-}\ell^{t_2}$ strategy is optimal.
        \item (Substitutes) If $\Delta v < 0$, there exist $t_1$ and $t_2$ such that either the $\text{Suspense}^{t_1}\text{-}\ell^{t_2}$ strategy or the Suspense${}^{t_1}$-Inconclusive-$\ell^{t_2}$ strategy is optimal.
    \end{enumerate}

Since $h_\ell$ and $h_r$ are concave and $\Delta h' > 0$, conditions (c) and (d) in \cref{thm:binary_suspense_L} always hold. Let $F(t_1,t_2) = \Psi(t_1,t_2) - h'_\ell(t_1)$. Define $t_{1,\ell}^*(t_2) \coloneqq \min\{t_1 \geq 0 : F(t_1,t_2) \geq 0\}$. This is well-defined because $F(t_2,t_2) = \Delta h'(t_2) > 0$. Consider that
\[e^{-t_1} F(t_1,t_2) = \int_{t_1}^{t_2} e^{-s} h''_\ell(s) \d s + e^{-t_2} \Delta h'(t_2)\]
is strictly increasing and continuous in $t_1$. This implies (1) $F(t_1,t_2) \geq 0$ if and only if $t_1 \geq t^*_{1,\ell}(t_2)$, and (2) $t^*_{1,\ell}(t_2)$ is continuous.
\begin{itemize}
    \item $t^*_{1,\ell}(t_2)$ is increasing in $t_2$. To see this, $\frac{\partial F}{\partial t_2} = e^{t_1 - t_2}(-\Delta h'(t_2) + h''_r(t_2)) < 0$. Consider any $t_2 < t'_2$. If $t_1 > t^*_{1,\ell}(t'_2)$, then $F(t_1,t_2) > F(t_1,t'_2) \geq 0$ since $F$ is decreasing in $t_2$. This implies $t^*_{1,\ell}(t'_2) > t^*_{1,\ell}(t_2)$, as desired.
\end{itemize}

Define $\hat{t}_2 \coloneqq \max\{t_2 \geq 0 : t_2 \leq \mu^{*-1}_{t^*_{1,\ell}(t_2)}(1)\}$. This is well-defined because $\mu^{*-1}_t(1)$ is bounded above. Let $\hat{t}_1 \coloneqq t^*_{1,\ell}(\hat{t}_2)$. From the definition of $\hat{t}_2$ we must have $\hat{t}_2 = \mu^{*-1}_{\hat{t}_1}(1)$. Note that if $t_2 \leq \hat{t}_2$, then $t_2 \leq \mu^{*-1}_{t^*_{1,\ell}(t_2)}(1)$. This is because (1) $t^*_{1,\ell}(t_2)$ is increasing in $t_2$, and (2) $\mu^{*-1}_{t_1}(1) = t_1 + \log\big(\frac{2\mu_0 - t_1}{\mu_0}\big)$ is decreasing in $t_1$.

We consider the following three cases of $\Delta v$:
\begin{itemize}
    \item Case 1: $\Delta v \in [h'_r(\hat{t}_2) - 2\Psi(\hat{t}_1,\hat{t}_2) - \Delta h(\hat{t}_2), h'_r(\hat{t}_2) - \Delta h(\hat{t}_2)]$. \Cref{prop:binary_completeness_sufficient} part 2 implies that the $\text{Suspense}^{\hat{t}_1}\text{-}\ell^{\hat{t}_2}$ strategy is optimal. Note that $\text{Suspense}^{\hat{t}_1}\text{-}\ell^{\hat{t}_2}$ coincides with $\text{Suspense}^{\hat{t}_1}\text{-Inconclusive-}\ell^{\hat{t}_2}$ because $\hat{t}_2 = \mu^{*-1}_{\hat{t}_1}(1)$.
    \item Case 2: $\Delta v > h'_r(\hat{t}_2) - \Delta h(\hat{t}_2)$. We have the following subcases.
    \begin{itemize}
        \item Case 2.1: $\Delta v \geq h'_r(0)$. The last part of \cref{prop:binary_completeness_sufficient} implies that the $\text{Suspense}^{0}\text{-}\ell^{0}$ strategy is optimal.
        \item Case 2.2: $\Delta v < h'_r(0)$. This implies there exists $t^*_2 \in (0, \hat{t}_2)$ such that $\Delta v = h'_r(t^*_2) - \Delta h(t^*_2)$. Set $t^*_1 = t^*_{1,\ell}(t^*_2)$. \Cref{prop:binary_completeness_sufficient} part 1 implies that the $\text{Suspense}^{t^*_1}\text{-}\ell^{t^*_2}$ strategy is optimal.
    \end{itemize}
    \item Case 3: $\Delta v < h'_r(\hat{t}_2) - 2\Psi(\hat{t}_1,\hat{t}_2) - \Delta h(\hat{t}_2)$. This implies $\Delta v < 0$.\footnote{This is because $h'_r(t_2) - 2\Psi(t_1,t_2) \leq (1 - 2 e^{t_1 - t_2}) h'_r(t_2) \leq 0$ since $t_2 < \mu^{*-1}_{t_1}(1) = t_1 - \log \mu^R_{t_1} < t_1 + \log 2$.} We have the following subcases:
    \begin{itemize}
        \item Case 3.1: $\Delta v \leq -h'_r(0)$. The last part of \cref{prop:binary_completeness_sufficient} implies that the $\text{Suspense}^{0}\text{-Inconclusive-}\ell^{0}$ strategy is optimal.
        \item Case 3.2: $\Delta v > -h'_r(0)$. Note that $h'_r(0) - 2\Psi(0,0) - \Delta h(0) = -h'_r(0)$. This implies there exists $t^*_2 \in (0, \hat{t}_2)$ such that $\Delta v = h'_r(t^*_2) - \Psi(t^*_1,t^*_2) - \Delta h(t^*_2)$, where $t^*_1 = t^*_{1,\ell}(t^*_2)$. \Cref{prop:binary_completeness_sufficient} part 3 implies that the $\text{Suspense}^{t^*_1}\text{-}$Inconclusive-$\ell^{t^*_2}$ strategy is optimal.
    \end{itemize}
\end{itemize}
\end{proof}

\begin{figure}[htbp]
\centering
\vspace{-2em}
\tikzset{every picture/.style={line width=0.75pt}} 

\begin{tikzpicture}[x=0.75pt,y=0.75pt,yscale=-1,xscale=1]

\draw    (43.94,54.82) -- (533.66,53.71) ;
\draw [shift={(536.66,53.7)}, rotate = 179.87] [fill={rgb, 255:red, 0; green, 0; blue, 0 }  ][line width=0.08]  [draw opacity=0] (5.36,-2.57) -- (0,0) -- (5.36,2.57) -- cycle    ;
\draw [shift={(40.94,54.83)}, rotate = 359.87] [fill={rgb, 255:red, 0; green, 0; blue, 0 }  ][line width=0.08]  [draw opacity=0] (5.36,-2.57) -- (0,0) -- (5.36,2.57) -- cycle    ;
\draw    (288.8,42.02) -- (288.8,65.76) ;
\draw    (363.16,42.02) -- (363.16,65.76) ;
\draw    (439.42,42.02) -- (439.42,65.76) ;
\draw    (136.27,42.02) -- (136.27,65.76) ;
\draw    (212.53,42.02) -- (212.53,65.76) ;
\draw   (213.2,78.73) .. controls (213.2,83.4) and (215.53,85.73) .. (220.2,85.73) -- (278.53,85.73) .. controls (285.2,85.73) and (288.53,88.06) .. (288.53,92.73) .. controls (288.53,88.06) and (291.86,85.73) .. (298.53,85.73)(295.53,85.73) -- (356.87,85.73) .. controls (361.54,85.73) and (363.87,83.4) .. (363.87,78.73) ;
\draw   (364.43,78.43) .. controls (364.47,83.1) and (366.82,85.41) .. (371.49,85.36) -- (391.99,85.18) .. controls (398.66,85.12) and (402.01,87.42) .. (402.05,92.09) .. controls (402.01,87.42) and (405.32,85.06) .. (411.99,84.99)(408.99,85.02) -- (432.49,84.81) .. controls (437.16,84.77) and (439.47,82.42) .. (439.42,77.75) ;
\draw   (136.53,78.73) .. controls (136.53,83.4) and (138.86,85.73) .. (143.53,85.73) -- (164.2,85.73) .. controls (170.87,85.73) and (174.2,88.06) .. (174.2,92.73) .. controls (174.2,88.06) and (177.53,85.73) .. (184.2,85.73)(181.2,85.73) -- (204.87,85.73) .. controls (209.54,85.73) and (211.87,83.4) .. (211.87,78.73) ;
\draw   (439.87,78.07) .. controls (439.88,82.74) and (442.22,85.06) .. (446.89,85.04) -- (472.57,84.95) .. controls (479.24,84.92) and (482.58,87.24) .. (482.6,91.91) .. controls (482.58,87.24) and (485.9,84.9) .. (492.57,84.87)(489.57,84.88) -- (518.25,84.78) .. controls (522.92,84.76) and (525.24,82.42) .. (525.22,77.75) ;
\draw   (47.93,79.56) .. controls (47.97,84.23) and (50.32,86.54) .. (54.99,86.49) -- (81.96,86.24) .. controls (88.63,86.17) and (91.98,88.47) .. (92.03,93.14) .. controls (91.98,88.47) and (95.29,86.11) .. (101.96,86.05)(98.96,86.08) -- (128.93,85.79) .. controls (133.6,85.75) and (135.91,83.4) .. (135.87,78.73) ;

\draw (204.46,17.04) node [anchor=north west][inner sep=0.75pt]  [font=\normalsize] [align=left] {$\displaystyle time\ 0\ persuasion\ gain\ \Delta v$};
\draw (283.04,65.93) node [anchor=north west][inner sep=0.75pt]   [align=left] {$\displaystyle 0$};
\draw (250.97,97.6) node [anchor=north west][inner sep=0.75pt]  [font=\footnotesize]  {$ \begin{array}{l}
Suspense-\ell \\
with\ certain\ \\
stopping\ belief
\end{array}$};
\draw (364.52,98.86) node [anchor=north west][inner sep=0.75pt]  [font=\footnotesize]  {$ \begin{array}{l}
Suspense-\ell \\
\end{array}$};
\draw (142.75,98.18) node [anchor=north west][inner sep=0.75pt]  [font=\footnotesize]  {$ \begin{array}{l}
Suspense-\\
inconclusive-\ell \\
\end{array}$};
\draw (49.83,99.46) node [anchor=north west][inner sep=0.75pt]  [font=\footnotesize]  {$ \begin{array}{l}
static\ info.\ \\
to\ maximize\ \\
probability\ of\ \\
\ell 
\end{array}$};
\draw (457.71,95.99) node [anchor=north west][inner sep=0.75pt]  [font=\footnotesize]  {$ \begin{array}{l}
static\ info.\ \\
to\ maximize\ \\
probability\ of\ \\
r
\end{array}$};

\end{tikzpicture}
\caption{Strategies as $\Delta v$ varies (fix $\Delta h' > 0$).}
\label{fig:vary_deltav}
\end{figure}

Another way of thinking about \cref{prop:binary_completeness_formal} is to fix $h_r, h_\ell$ such that $\Delta h' > 0$ and vary $\Delta v$, the time-$0$ gain from the agent taking action $r$ over $\ell$; this is illustrated in \cref{fig:vary_deltav}. For the case in which $\Delta v > 0$, we are in case 1 of \cref{prop:binary_completeness_formal}, so the Suspense-$\ell$ strategy is optimal.\footnote{Note that time-$0$ static persuasion to maximize the probability the agent chooses $r$ is a special case.} Now consider the case in which $\Delta v < 0$; this corresponds to case 2 of \cref{prop:binary_completeness_formal} with the labels flipped, so either Suspense$^{t_1}$-Inconclusive-$\ell^{t_2}$ or Suspense$^{t_1}$-$\ell^{t_2}$ is optimal.\footnote{Even though $\Delta v < 0$, Suspense$^{t_1}$-$\ell^{t_2}$ can be optimal when $h'_r(\hat{t}_2) - \Delta h(\hat{t}_2) < 0$.}


\section{Proofs for the Linear Application}\label{appendix:linear_section}

This appendix collects the proofs of the two results stated in \cref{sec:linear}: the first-order characterization for the continuum (\cref{thm:foc_continuum}) and the optimality of dynamic tail-censorship (\cref{thm:tail_censorship}).

\subsection{Proof of \cref{thm:foc_continuum}}\label{appendix:foc_continuum_proof}

The proof of \cref{thm:foc_continuum} mirrors that of \cref{thm:foc}. It exploits the strong-duality argument of \cref{lem:duality}, augmented by the \eqref{eqn:MPS} constraint on the marginal over posterior means: the dual decouples into a time-dependent multiplier $\Lambda$ on the conditional obedience constraint and a time-invariant convex multiplier $\Gamma$ on the mean-preservation constraint. The complementary slackness condition $\mathcal{L}(g, \Lambda, \Gamma) = \E_g[V]$ together with the support equality $l_{g,\Lambda,\Gamma}(\mm, t) = 0$ on $\mathrm{supp}(g)$ then sandwiches the Lagrangian value between primal and dual feasible levels, certifying optimality.

\begin{proof}[Proof of \cref{thm:foc_continuum}]
Suppose $(g,\Lambda,\Gamma)$ satisfies FOC, OC-C, and the complementary slackness condition $\mathcal{L}(g,\Gamma,\Lambda) = \mathbb{E}_g[V]$. Note that
\[\mathcal{L}(f,\Lambda,\Gamma) \geq \int V(m ,t) f(\d m, \d t), \]
for every $f \in \overline{\Delta}_{\mu_0}$ that satisfies (OC-C) and (MPS) because
\begin{align*}
    \int_{0}^1 \bigg( \int_0^m \int_0^x \bigg(\mu_0(\d y) - \int_{\tau \in T} g(\d y,\d t) \bigg) \d x \bigg) \d^2 \Gamma(m) &= \int \Gamma(m) \mu_0(\d m)-\int \Gamma(m) g(\d m, \d t) \\
    &\geq 0,
\end{align*}
where the inequality follows the constraint that $g$ is a mean-preserving contraction of $\mu_0$ and $\Gamma$ is convex.

We show that $\mathcal{L}(g,\Gamma, \Lambda) \geq  \mathcal{L}(g',\Gamma, \Lambda )$ for  every $g' \in \Delta_{\mu_0}.$ Suppose for the purpose of contradiction that such $g'$ exists. Let $\Delta  \coloneqq g'-g$. Then, since $\mathcal{L}$ is concave in $g$, for every $\alpha \in (0,1),$
\begin{align}
&\frac{\mathcal{L}(\alpha f' + (1-\alpha)f,\Gamma,\Lambda) - \mathcal{L}(f,\Gamma,\Lambda)}{\alpha} \geq \mathcal{L}(f',\Gamma,\Lambda) - \mathcal{L}(f,\Gamma,\Lambda) > 0 \nonumber\\
\Longrightarrow & \underline{\lim}_{\alpha \to 0} \frac{\mathcal{L}(\alpha f' + (1-\alpha)f,\Gamma,\Lambda) - \mathcal{L}(f,\Gamma,\Lambda)}{\alpha} > 0 \nonumber\\
\Longrightarrow & \int V(m, t)\Delta(\d m,\d t)    +\int_{t \in T^\circ} \int_{\tau>t} U(m,\tau) \Delta(\d m, \d \tau ) \d\Lambda(t) \nonumber
\\
&\quad- \overline{\lim_{\alpha\to 0}} \int_{t \in T^\circ} \left( \mathcal{S}(g +  \Delta,t) - \mathcal{S}(g,t) \right) d \Lambda(t) - \int \Gamma(m) \Delta(\d m,\d t) > 0, \label{eqn: lagrangian_suff}
\end{align}
where
\begin{align*}
    \mathcal{S}(g,t) \coloneqq \int g (\d m,\d s) \cdot U \left( \widehat{m}_t(g),t \right).
\end{align*}
Consider that
\begin{align}
    \int_{t \in T^\circ} \int_{\tau>t} U(m,\tau) \Delta(\d m, \d \tau ) \d\Lambda(t)  = \int U(m,t) \Lambda(t) \Delta (\d m, \d t). \label{eqn: lagrangian: U}
\end{align}
Also,
\begin{align*}
    \widehat{m}_t(g + \alpha\Delta) = \widehat{m}_t(g) +
    \frac{ \int_{s>t} (m - \widehat{m}_t(g))  \Delta(\d m, \d s)}{\int g(\d m,\d s ) } \alpha +\mathcal{o}(\alpha),
\end{align*}
which implies
\begin{align*}
    U(\widehat{m}_t(g + \alpha\Delta)) \geq  U(\widehat{m}_t(g)) + \frac{ \int_{s>t} (m - \widehat{m}_t(g))  \Delta(\d m, \d s)}{\int g(\d m,\d s ) } U'_m(\widehat{m}_t(g),t)\alpha + \mathcal{o}(\alpha)
\end{align*}
for any selection of a subgradient $U'_m.$ Therefore,
\begin{align}
    \mathcal{S}(g+\alpha\Delta,t) \geq \mathcal{S}(g,t) + \int_{s>t} \left( U(\widehat{m}_t(g),t) + \left(  m - \widehat{m}_t(g) \right) U'_m(\hat{m}_t(g),t) \right)  \Delta(\d m,\d s) \cdot \alpha + \mathcal{o}(\alpha). \label{eqn: lagrangian: OC}
\end{align}
From \eqref{eqn: lagrangian_suff}, \eqref{eqn: lagrangian: U}, and \eqref{eqn: lagrangian: OC}, we obtain
\begin{align}
    &\int l_{g,\Lambda,\Gamma}(m,t)(g'-g) (\d m, \d t) > 0. \label{eqn: sufficiency_contradiction}
\end{align}
Since $g' \in \overline{\Delta}_{\mu_0}$ and $l_{g,\Lambda,\Gamma}(m,t) = 0$ on the support of $g$, we have
\begin{align*}
    \int l_{g,\Lambda,\Gamma}(m,t) g'(\d m,\d t) = 0 \geq \int l_{g,\Lambda,\Gamma}(m,t) g'(\d m,\d t),
\end{align*}
which contradicts \cref{eqn: sufficiency_contradiction}, as required.

\noindent Consider any $g' \in \overline{\Delta}_{\mu_0}$ that satisfies all constraints in (OC-C) and (MPS). Then,
\begin{align*}
    \int V(m,t) g'(\d m,\d t) \leq \mathcal{L}(g',\gamma ,\Lambda) \leq \mathcal{L}(g,\gamma,\Lambda) = \int V(m,t) g(\d m, \d t),
\end{align*}
where the last equality is implied by the complementary slackness condition.
\end{proof}

\subsection{Proof of \cref{thm:tail_censorship}}\label{appendix:tail_proof}

The proof of \cref{thm:tail_censorship} applies the first-order characterization \cref{thm:foc_continuum} to the central-bank--investor environment. The conjectured optimum is a tail-censorship policy with thresholds $\alpha(t)$ and $\beta(t)$ satisfying \cref{eqn:ODE_continuum}; the construction below specifies the multipliers $\Lambda$ and $\Gamma$ and verifies the FOC equality on the support and the FOC inequality off the support. The verification rests on the two characterizing properties highlighted in the body: the no-upward-surprise identity $\widehat{m}_t = \beta(t)$ and the investor indifference condition. We begin with a useful lemma stating the existence of a solution of \cref{eqn:ODE_continuum}.

\begin{lem} \label{lem: existence of ODE}
Suppose $a(\cdot)$ and $b(\cdot)$ solves the following system of equations
\begin{align}\label{eqn: backward ODE continuum}
   \begin{dcases}a'(t)&=\frac{\mathcal{M}(a(t))+1-\mathcal{M}(b(t))}{\mu(a(t))(0.5-a(t))} \times rb(t) \eqqcolon x(a(t),b(t))  \\
    b'(t)&=-\frac{b(t)-a(t)}{0.5-a(t)} \times rb(t) \eqqcolon y(a(t),b(t)) \end{dcases}
\end{align}
with initial conditions $b(0) = 1$ and $a(0) = 0.$ Define
\begin{align*}
m^* &\coloneqq b(T_0-), \text{where} \\
T_0 &\coloneqq \sup\{t_0\geq 0: \text{a solution of \cref{eqn: backward ODE continuum} is defined on $[0,t_0]$ and $b(t_0) \geq 0.5$}\}.
\end{align*}
Then,
\begin{enumerate}
    \item For every $m_0 \geq 0.5$, 1) if $m_0 \geq m^*$, then there exists $T \geq 0$ such that $\lim_{t \uparrow T} b(T) = m^*$ and $\lim_{t \uparrow T} a(T) \leq 0.5$, and 2) if $m_0 < m^*$, then there exists $T \geq 0$ such that $\lim_{t \uparrow T} a(t) = 0.5$ and $\lim_{t \uparrow T} b(t) > m_0.$
    \item Fix $m_0 > 0.5$. A sufficient condition of the existence of $T$ that $b(T) = m_0$ is
\[\mathbb{E}[(m_0-m) 1\{m < 0.5\}] > 1- m_0. \]
\end{enumerate}
\end{lem}
\begin{proof}[Proof of \cref{lem: existence of ODE}]
    We restrict the space of $(a(t),b(t))$ to an open set $E \coloneqq (-\epsilon, 0.5) \times (0.5,1+\epsilon)$ for small $\epsilon> 0.$ Note that we extend $\mu(\cdot)$ to support the domain $(-\epsilon, 1+\epsilon)$ in the way that $\mu(m) = \mu(1)$ when $m > 1$ and $\mu(m) = \mu(0)$ when $m<0$. This implies $\mu$ is still continuous and bounded away from $0$.

    Since $\mu$ and $g$ are continuous in $E$, the Peano existence theorem implies there exists a local solution around $t=0$. Let the maximal solution of the system of ODE with initial conditions $b(0) = 1$ and $a(0) = 0$ be defined on an interval $I \subset \mathbb{R}^+.$ This implies $T_0 = \sup I$ and $m^* = \lim_{t \uparrow T_0} b(t).$

    First, we show that $a$ and $b$ are increasing and decreasing in $I$, respectively. Since $a_0> b_0$ for every $(a_0,b_0) \in E$, we must have $b'(t) = y(a(t),b(t)) < 0$ for every $t \in I$. Thus, $b$ is strictly decreasing, and $b(t) < b(0) = 1$ for every $t \in I$. This implies, for every $t \in I,$
    \begin{align*}
        a'(t)&=\frac{\mathcal{M}(a(t))+1-\mathcal{M}(b(t))}{\mu(a(t))(0.5-a(t))} \times rb(t) > \frac{\mathcal{M}(a(t))}{\mu(a(t))(0.5-a(t))} \times 0.5r \geq 0.
    \end{align*}
    Thus, $a$ is weakly increasing over $I$, as desired.

    Next, we show that $I$ must be a finite interval. Consider that $b'(t) < 1$ for every $t \in I$, which implies $0.5 <b(t) < 1-t$. Therefore, $T_0 \leq 0.5$, as desired.

    First, we show either $\lim_{t \uparrow \sup I} a(t) = 0.5$ or $m^* = 0.5$. Suppose for the purpose of contradiction that
    \[a^* \coloneqq \lim_{t \uparrow \sup I} a(t) < 0.5 < \lim_{t \uparrow \sup I} b(t) = m^*.\]
    Then, the graph of the solution to the ODEs must satisfy the following
    \begin{align*}
        \{(t, a(t),b(t)) : t\in I\} \subset \underbrace{[0,T_0] \times [0, a^*] \times [m^* , 1]}_{\eqqcolon K} \subset \mathbb{R}^+ \times E.
    \end{align*}
    Since $K$ is compact, this contradicts the fact that the graph of the maximal solution cannot live in a compact subset. Thus, we must have either $a^* = 0.5$ or $m^* = 0.5.$

    If $m_0 \geq m^* = \lim_{t \uparrow T_0} b(t) $, then it is clear that there exists $T\geq 0 $ such that $b(T-) = m_0$ by the Intermediate Value Theorem. Suppose otherwise that $m_0 < m^*$. Since $m_0 \geq 0.5$, we have $m^* >0.5.$ This means $a^* = 0.5$, implying the first half of the lemma by choosing $T = T_0.$

     For the second part of the lemma, suppose there is no $T>0$ such that $b(T) = \underline{\theta}$. We wish to show $\lim_{t \uparrow \sup I} a(t) = 0.5$, and the necessary condition is
    \[\mathbb{E}[(m_0-m) 1\{m < 0.5\}] \leq 1- m_0.\]
    Since $b$ is continuous on $I$, we must have $b(t) > \underline{\theta}$ for every $t \in I.$

    Consider that
    \begin{align}
       m_0 \leq \lim_{t \uparrow T_0} b(t) &= 1 - \int_0^{T_0}
        \frac{b(s) - a(s)}{0.5-a(s)} \times rb(s)\d s \nonumber\\
        &= 1 - \int_0^{T_0} \frac{b(s)-a(s)}{\mathcal{M}(a(s)) + 1- \mathcal{M}(b(s))} \d \mathcal{M}(a(s))\nonumber\\
        &\leq 1 - \int_0^{T_0} (\underline{\theta}-a(s)) \d \mathcal{M}(a(s)) \nonumber\\
        &= 1+ \mathbb{E}[(m - m_0) 1\{m < 0.5\}], \nonumber
    \end{align}
    which directly implies our desired necessary condition.
\end{proof}

We are now ready to show \cref{thm:tail_censorship}.
\begin{proof}[Proof of \cref{thm:tail_censorship}]
Fix $\underline{\theta}.$ We choose $T$ and $m^*$ as defined in \cref{lem: existence of ODE}, where we set $m_0 = \max\{0.5, \underline{\theta}\}.$ We set $\alpha(t) = a(T-t)$ and $\beta(t) = b(T-t)$ for every $t \in [0,T],$ which satisfy \cref{eqn:ODE_continuum}. We will prove the proposition case-by-case.

\fbox{Case 1: $\underline{\theta} \geq m^*$} \cref{lem: existence of ODE} implies $\beta(0) = \underline{\theta}$ and $\alpha(0) \leq 0.5.$ We choose the following Lagrange coefficients:
\begin{align*}
    \Lambda(t) &= \begin{dcases} \frac{c(\beta(t)-\underline{\theta})}{r\beta(t)} e^{rt} &\text{ if $t \leq T$} \\ \frac{c(1-\underline{\theta})}{r} e^{rt} &\text{ if $t > T$}\end{dcases},\\
    \Gamma(m) &= \begin{dcases}
         \int_{\beta(0)}^m \gamma^+(m) \d m &\text{if $m > \beta(0)$}\\ 0 & \text{if $m \in [\alpha(0),\beta(0)]$} \\
         -\int_{m}^{\alpha(0)} \gamma^-(m) \d t & \text{if $m < \alpha(0)$},
    \end{dcases}
\end{align*}
where
\begin{align*}
    \gamma^+(m) = c\underline{\theta} \int_0^{\beta^{-1}(m)} \frac{1}{\beta(s)} \d s, \quad\quad \gamma^-(m) = c\underline{\theta} \int_0^{\alpha^{-1}(m)} \frac{1}{\beta(s)} \d s - \frac{c}{r}\bigg(1 - \frac{\underline{\theta}}{\beta(\alpha^{-1}(m))} \bigg).
\end{align*}
Note that $\Lambda$ is increasing and $\Lambda(0) = 0$, so $\Lambda$ is continuous at $0.$ Moreover, $\Gamma$ is convex because
\begin{align*}
    \Gamma'(m) = \begin{cases}
        \gamma^+(m) &\text{if $m > \beta(0)$} \\
        0 &\text{if $m \in [\alpha(0),\beta(0)]$} \\
        \gamma^-(m) &\text{if $m < \alpha(0)$},
    \end{cases}
\end{align*}
implying 1) $\Gamma'$ is increasing over the domain $[\alpha(0),1]$ because $\beta$ is increasing, and 2) $\Gamma'$ is increasing over the domain $[0,\alpha(0)]$ because $\lim_{m \uparrow \alpha(0)} \gamma'(m) = 0$ by $\beta(0) = \underline{\theta}$ and
\begin{align*}
    \frac{\d \gamma^-(\alpha (t))}{\d t} = \frac{c\underline{\theta}}{r\beta(t)^2} (r\beta(t) - \beta'(t)) < 0
\end{align*}
by the ODE given by \cref{eqn:ODE_continuum}.

Plugging the parametric assumptions to the FOC, we get: for every $t \in T$ and $m \in [0,1],$
\begin{align*}
    l_{g,\Lambda,\Gamma}(m,t)=
\begin{dcases}
    c(m-\underline{\theta}) t   - rm\int_0^t e^{-rs} \Lambda(s) \d s-\Gamma(m) & m>0.5\\
    c(m-\underline{\theta}) t + (0.5-m) \Lambda(t) e^{-rt}  - rm\int_0^t e^{-rs} \Lambda(s)\d s-\Gamma(m) & m\le 0.5
\end{dcases}
\end{align*}
\noindent\underline{Step 1: show $[\alpha(0),\beta(0)] \subset \text{argmax}_{m} l_{g,\Lambda,\Gamma} (m , 0)$.}

\noindent We can see that $l_{g,\Lambda,\Gamma} (m , 0) = - \Gamma(m) \leq 0$, where the maximum is attained at $[\alpha(0),\beta(0)] $ from our construction of $\Gamma,$ as required.

\noindent\underline{Step 2: show $\beta(t) \in \text{argmax}_{m \geq 1/2} l_{g,\Lambda,\Gamma} (m,t)$ and $\alpha(t) \in \text{argmax}_{m < 1/2} l_{g,\Lambda,\Gamma} (m,t)$ for every $t\leq T.$}

\noindent The FOC with respect to $m$ when $m > 0.5$ is given by
\begin{align*}
    \frac{\partial l_{g,\Lambda,\Gamma} (m,t)}{\partial m} = \gamma^+(\beta(t)) - \Gamma'(m) , \frac{\partial^2 l_{g,\Lambda,\Gamma} (m,t)}{\partial m^2} = -\Gamma''(m)< 0.
\end{align*}
The FOC equals $0$ when $m = \beta(t)$ and $l_{g,\Lambda,\Gamma}(\cdot,t)$ is concave over $[1/2,1]$ due to its negative second derivative. Thus, $\beta(t) \in \text{argmax}_{m \geq 1/2} l_{g,\Lambda,\Gamma} (m,t)$.

\noindent Similarly, the FOC with respect to $m$ when $m \leq 0.5$ is given by
\begin{align*}
    \frac{\partial l_{g,\Lambda,\Gamma} (m,t)}{\partial m} = \gamma^+(\beta(t))-  \Gamma'(m), \frac{\partial^2 l_{g,\Lambda,\Gamma} (m,t)}{\partial m^2} = -\Gamma''(m)< 0.
\end{align*}
The FOC equals $0$ when $m = \alpha(t)$ from our construction of $\Lambda$ and $\Gamma$. Also, $l_{g,\Lambda,\Gamma}(\cdot,t)$ is concave over $[0,1/2]$ due to its negative second derivative. Thus, $\alpha(t) \in \text{argmax}_{m < 1/2} l_{g,\Lambda,\Gamma}(m,t)$, as desired.

\noindent\underline{Step 3: show $l_{g,\Lambda,\Gamma}(\alpha (t),t) = l_{g,\Lambda,\Gamma}(\beta (t),t) = 0$ for every $t \in [0,T]$}

\noindent Consider that, by Step 2 and the envelope theorem,
\begin{align*}
    \frac{\d}{\d t} l_{g,\Lambda,\Gamma}(\beta (t),t)
    = \frac{\partial l_{g,\Lambda,\Gamma}}{ \partial t} (\beta(t),t)  =  c (\beta(t) - \underline{\theta}) - r\beta(t) e^{-rt} \Lambda(t) = 0.
\end{align*}
Thus, $l_{g,\Lambda,\Gamma}(\beta (t),t) = l_{g,\Lambda,\Gamma}(\beta (0),0) = 0$ from Step 1. Similarly,
\begin{align*}
    \frac{\d}{\d t} l_{g,\Lambda,\Gamma}(\alpha (t),t)
    &= \frac{\partial l_{g,\Lambda,\Gamma}}{ \partial t} (\alpha(t),t)  \\
    &= -c\underline{\theta}  + \frac{c\underline{\theta}}{r}(0.5 - \alpha(t)) \frac{\beta'(t)}{\beta(t)^2} + \frac{c\alpha(t) \underline{\theta}}{\beta(t)} \\
    &= 0,
\end{align*}
where the last inequality is implied be the ODE given by \cref{eqn:ODE_continuum}. Thus, $l_{g,\Lambda,\Gamma}(\alpha (t),t) = l_{g,\Lambda,\Gamma}(\alpha (0),0) = 0$ from Step 1, as desired.

\noindent\underline{Step 4: show $l_{g,\Lambda,\Gamma} (m,t) \geq 0$ for every $t>T.$}

\noindent Consider that, for every $t>T$ and $m \in [0,1]$, $\frac{\partial l_{g,\Lambda,\Gamma}}{\partial t} (m,t) = c\underline{\theta}(m - 1) \leq 0 $. Thus, if $t>T$, then
\begin{align*}
    l_{g,\Lambda,\Gamma}(m ,t) \leq l_{g,\Lambda,\Gamma}(m ,T) \overset{\text{Step 2}}{\leq}  l_{g,\Lambda,\Gamma}(\beta(T) ,T) = 0,
\end{align*}
as desired.

All steps together imply $l_{f,\Gamma,\Lambda}(m,t) \leq 0,$ where the equality is attained if $(m,t) \in \text{supp } g.$ Moreover, the complementary slackness condition is clear since the obedience constraint binds for every time $t \in T^\circ$ and the mean-preserving spread constraint also binds for every $m \in m$ because our conjectured solution eventually gives full information, i.e., $\text{marg}_m g = \mu.$ \cref{thm:foc_continuum} implies $g$ must be the optimal information structure.

\fbox{Case 2: $\underline{\theta} < m^*$ and $m^* > 0.5$}
\cref{lem: existence of ODE} implies $\alpha(0) = 0.5$ and $\beta(0) > 0.5$. We choose the same Lagrange coefficients except small changes:
\begin{align*}
    \Lambda(t) &= \begin{dcases} 0 &\text{if $t = 0$}\\  \frac{c(\beta(t)-\underline{\theta})}{r\beta(t)} e^{rt} &\text{ if $0 < t \leq T$} \\ \frac{c(1-\underline{\theta})}{r} e^{rt} &\text{ if $t > T$}\end{dcases}
\end{align*}
All earlier arguments can be directly applied except the discontinuity of $\Lambda$ at $t=0$. However, $l_{g,\Lambda,\Gamma}(m,t)$ is still continuous at $t= 0$ when $ m \geq 0.5$, and $[\alpha(0),\beta(0)] \subset [1/2,1]$. Thus, we still have $l_{g,\Lambda,\Gamma}(\alpha(t),t) = l_{f,\Lambda, \Gamma}(\beta(t),t) = 0$ for every $t \in [0,T]$ from Step 3.

\fbox{Case 3: $m^* = 0.5 \geq \underline{\theta} $}
\cref{lem: existence of ODE} implies $\alpha(0) \leq 0.5$ and $\beta(0) = 0.5$ We choose the same Lagrange coefficients except small changes:
\begin{align*}
    \Lambda(t) &= \begin{dcases} 0 &\text{if $t = 0$}\\  \frac{c(\beta(t)-\underline{\theta})}{r\beta(t)} e^{rt} &\text{ if $0 < t \leq T$} \\ \frac{c(1-\underline{\theta})}{r} e^{rt} &\text{ if $t > T$}\end{dcases}
\end{align*}
Since $\widehat{m}_0 = 0.5$, we adjust a subgradient at time $0$ as $U'_m(0.5,0) = 0$. The FOC becomes, for every $t > 0$ and $m \in [0,1],$ $l_{g,\Lambda,\Gamma} (m , 0 ) = -\Gamma(m),$ and
\begin{align*}
    l_{g,\Lambda,\Gamma}(m,t)=
\begin{dcases}
    c(m-\underline{\theta}) t + (m-0.5)\Lambda(0+)e^{-rt}   - rm\int_{0+}^t e^{-rs} \Lambda(s) \d s-\Gamma(m) & m>0.5\\
    c(m-\underline{\theta}) t - (m-0.5) (\Lambda(t)-\Lambda(0+)) e^{-rt}  - rm\int_{0+}^t e^{-rs} \Lambda(s) \d s-\Gamma(m) & m\le 0.5
\end{dcases}
\end{align*}
All earlier arguments can be directly applied except the discontinuity of $\Lambda$ and $l$ at $t=0$. However, $l_{g,\Lambda,\Gamma}(m,t)$ is still continuous at $t= 0$ when $ m \leq 0.5$, and $[\alpha(0),\beta(0)] \subset [\alpha(0),1/2]$. Thus, we still have $l_{g,\Lambda,\Gamma}(\alpha(t),t) = l_{f,\Lambda, \Gamma}(\beta(t),t) = 0$ for every $t \in [0,T]$ from Step 3.
\end{proof}


\section{Extensions}\label{appendix:extensions}

The two extensions developed in \cref{sec:extensions} inherit from the baseline framework without further appendix machinery. In \cref{ssec:evolving_state}, the modification is a restriction of the feasibility set from $D$ to $\overline D$ or $\underline D$; both restricted sets remain compact, and the time-section $\overline D(t)$ (respectively $\underline D(t)$) is monotone in $t$ under set inclusion. The direct-communication reduction of \cref{appendix:reduction} and the strong-duality argument of \cref{appendix:duality_proof} therefore apply to the restricted problem verbatim, and \cref{thm:foc} characterizes its solutions. In \cref{ssec:costly_info}, the costly-information variant is absorbed into the principal's effective payoff $\widehat V = V - H$ by the chain rule for uniform-posterior-separable costs; the constrained variant adds a convex inequality \eqref{eqn:icc_c} for which the Lagrangian acquires an extra multiplier without disturbing the FOC structure. Detailed proofs of the topological prerequisites that ensure existence and tightness in the evolving-state case are developed in the supplementary material of \citet{koh2024attention}, to which we refer the reader for the underlying measure-theoretic constructions.

\end{document}